\documentclass[conference]{IEEEtran}

\usepackage{times}
\usepackage{dsfont}
\usepackage{times}
\usepackage{color}
\usepackage{balance}
\usepackage{url}
\let\labelindent\relax
\usepackage{enumitem}
\usepackage{siunitx}
\usepackage{tabularx,amssymb,amsmath,subfigure,multirow,algorithm,algorithmic,graphicx}

\newcommand{\nop}[1]{}

\newtheorem{definition}{\bf Definition}

\newtheorem{lemma}{\bf Lemma}[section]

\title{Prediction-Based Task Assignment in Spatial Crowdsourcing (Technical Report)}

\author{
    {Peng Cheng$^{\#}$, Xiang Lian$^{*}$, Lei Chen$^{\#}$, Cyrus Shahabi$^{\dagger}$}
    \vspace{1.6mm}\\
    \fontsize{10}{10}\selectfont\itshape
    $^{\#}$Hong Kong University of Science and Technology, Hong Kong, China\\
    \fontsize{9}{9}\selectfont\ttfamily\upshape
    \{pchengaa,  leichen\}@cse.ust.hk
    \vspace{1.2mm}\\
    \fontsize{10}{10}\selectfont\rmfamily\itshape
    $^{*}$Kent State University, Ohio, USA\\
    \fontsize{9}{9}\selectfont\ttfamily\upshape
    xlian@kent.edu
    \vspace{1.2mm}\\
    \fontsize{10}{10}\selectfont\rmfamily\itshape
    $^{\dagger}$University of Southern California, California, USA\\
    \fontsize{9}{9}\selectfont\ttfamily\upshape
    shahabi@usc.edu
}

\begin{document}

\maketitle

\hyphenpenalty=5000

\begin{abstract}
With the rapid advancement of mobile devices and crowdsourcing
platforms, \textit{spatial crowdsourcing} has attracted
much attention from various research communities. A
spatial crowdsourcing system periodically matches a
number of location-based workers with nearby spatial tasks (e.g.,
taking photos or videos at some specific locations). Previous studies
on spatial crowdsourcing focus on task assignment
strategies that maximize an assignment score based  solely  on the available information about workers/tasks at the time of assignment. These strategies can only achieve local
optimality by neglecting the workers/tasks that may join
the system in a future time. In contrast, in this paper, we aim to improve the global assignment, by considering both present and future (via predictions) workers/tasks. In particular, we
formalize a new optimization problem, namely \textit{maximum quality task assignment} (MQA). The optimization objective of MQA is to maximize a global assignment
quality score,
under a traveling budget constraint.
To tackle this problem, we design an effective grid-based
prediction method to estimate the spatial distributions of workers/tasks
in the future, and then utilize the predictions to assign workers to tasks at any given time instance. We prove that the MQA problem is
NP-hard, and thus intractable. Therefore, we propose efficient
heuristics to tackle the MQA problem, including
\textit{MQA greedy} and \textit{MQA divide-and-conquer} approaches,
which can efficiently assign workers to spatial tasks with
high quality scores and low budget consumptions. Through extensive
experiments, we demonstrate the efficiency and effectiveness of our approaches on both real and synthetic datasets.
\end{abstract}

\section{Introduction}
\label{sec:intro}

Mobile devices not only bring  convenience to
our daily life, but also enable people to easily perform location-based tasks in their vicinity, such as taking photos/videos
(e.g., street view of Google Maps \cite{GoogleMapStreetView}),
reporting traffic conditions (e.g., Waze \cite{waze}), and identifying the status of
display shelves at neighborhood stores (e.g., Gigwalk
\cite{gigwalk}). Recently, to exploit these phenomena, a new
framework, namely \textit{spatial crowdsourcing}
\cite{kazemi2012geocrowd},  for requesting workers to perform
spatial tasks, has drawn much attention from both academia (e.g.,
MediaQ \cite{kim2014mediaq}) and industry (e.g., TaskRabbit
\cite{taskrabbit}). A typical spatial crowdsourcing system (e.g.,
gMission \cite{chen2014gmission}) utilizes a number of
dynamically moving workers to accomplish spatial tasks (e.g., taking
photos/videos), which requires workers to physically go to the
specified locations to complete these tasks.

We first provide the following motivating example.

\vspace{0.5ex}\noindent\textbf{Example 1 (The Spatial Crowdsourcing
Problem with Multiple Time Instances)} {\it Consider a scenario of
spatial crowdsourcing in Figure \ref{fig:example}, where spatial
tasks, $t_1$ $\sim$ $t_3$, are represented by red circles, and
workers, $w_1$ $\sim$ $w_3$, are denoted by blue triangles. In
particular, Figure \ref{fig:fig1a} shows a worker, $w_1$, and two
tasks, $t_1$ and $t_2$, which join the system  at
a timestamp $p$ (denoted by markers with solid border). Figure
\ref{fig:fig1b} depicts two more workers, $w_2$ and $w_3$, and one
more task, $t_3$, which arrive at the system at a future timestamp $(p+1)$ (denoted by markers with the
dashed border).

Each worker $w_i$ ($1\leq i\leq 3$) has a particular expertise to
perform different types of spatial tasks $t_j$ ($1\leq j\leq 3$).
Thus, we assume that the competence of a worker $w_i$ to perform task $t_j$
can be captured by a quality score $q_{ij}$ (as shown in Table
\ref{tab:distance_quality}). Furthermore, each worker, $w_i$, is
provided with some reward (e.g., monetary or points) to
cover the traveling cost from $w_i$ to $t_j$, which is proportional
to the traveling distance, $dist(w_i, t_j)$ (as depicted in Table
\ref{tab:distance_quality}), where $dist(x, y)$ is a distance
function from $x$ to $y$. Thus, workers are persuaded
to perform tasks even if they are far away from the tasks' locations, which
improves the task completion rate, especially for
workers/tasks with unbalanced location distributions.

\begin{figure}[t!]
	\centering
	\setcounter{subfigure}{-1}
	\subfigure{
		\includegraphics[scale=0.2]{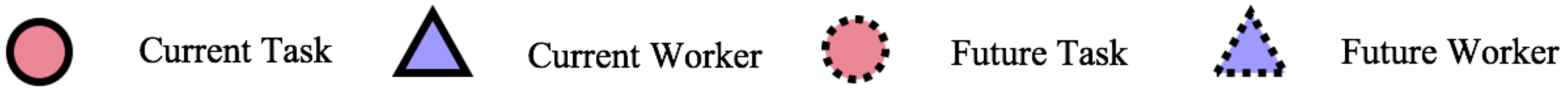}
		\label{fig:bar}
	}\\\vspace{-1ex}
	\hspace{-4ex}\subfigure[][{\scriptsize assignment at timestamp $p$}]{
		\scalebox{0.19}[0.19]{\includegraphics{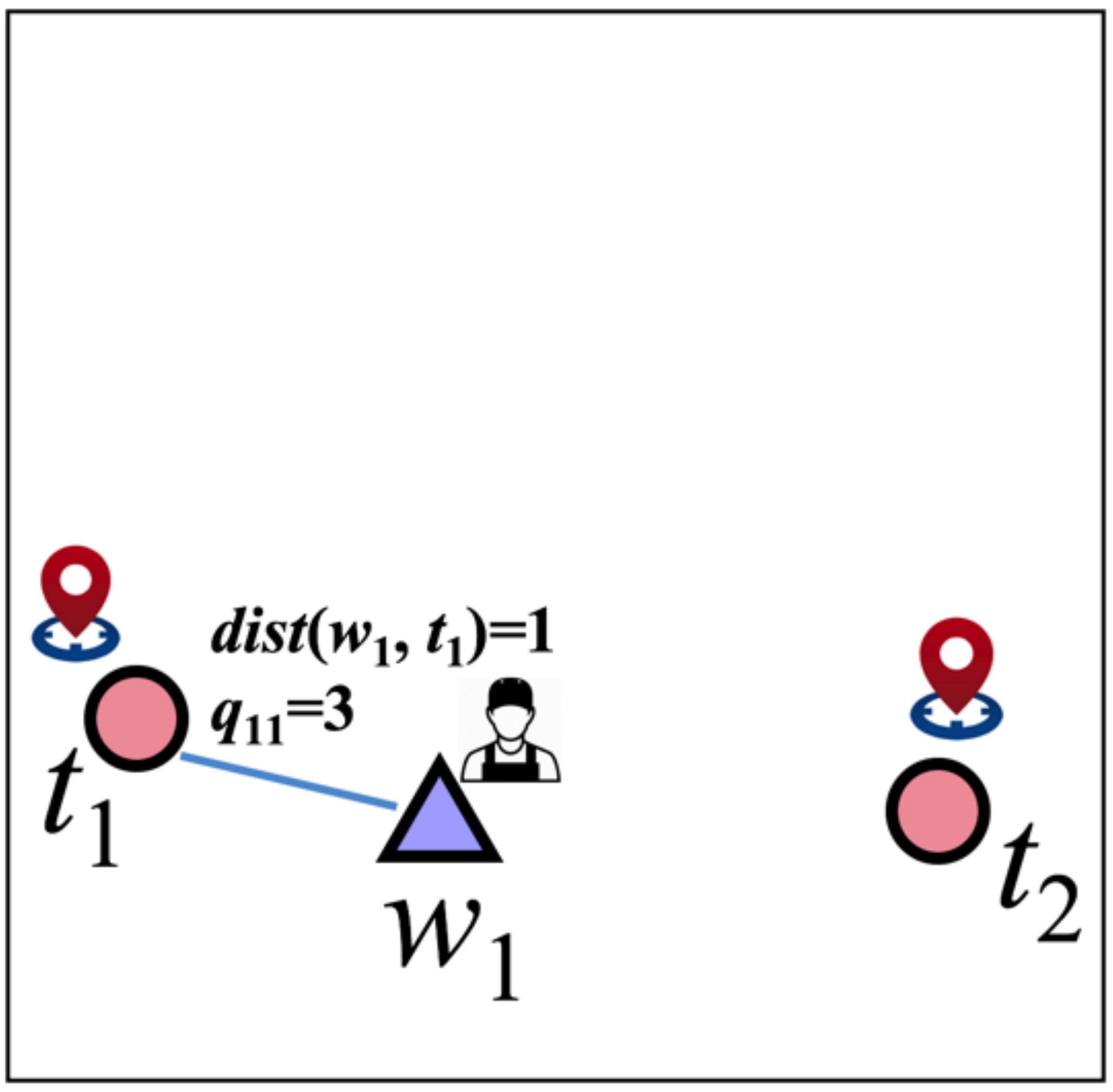}}
		\label{fig:fig1a}
	}
	\hspace{0.5ex}\subfigure[][{\scriptsize assignment at timestamp $(p+1)$}]{
		\scalebox{0.19}[0.19]{\includegraphics{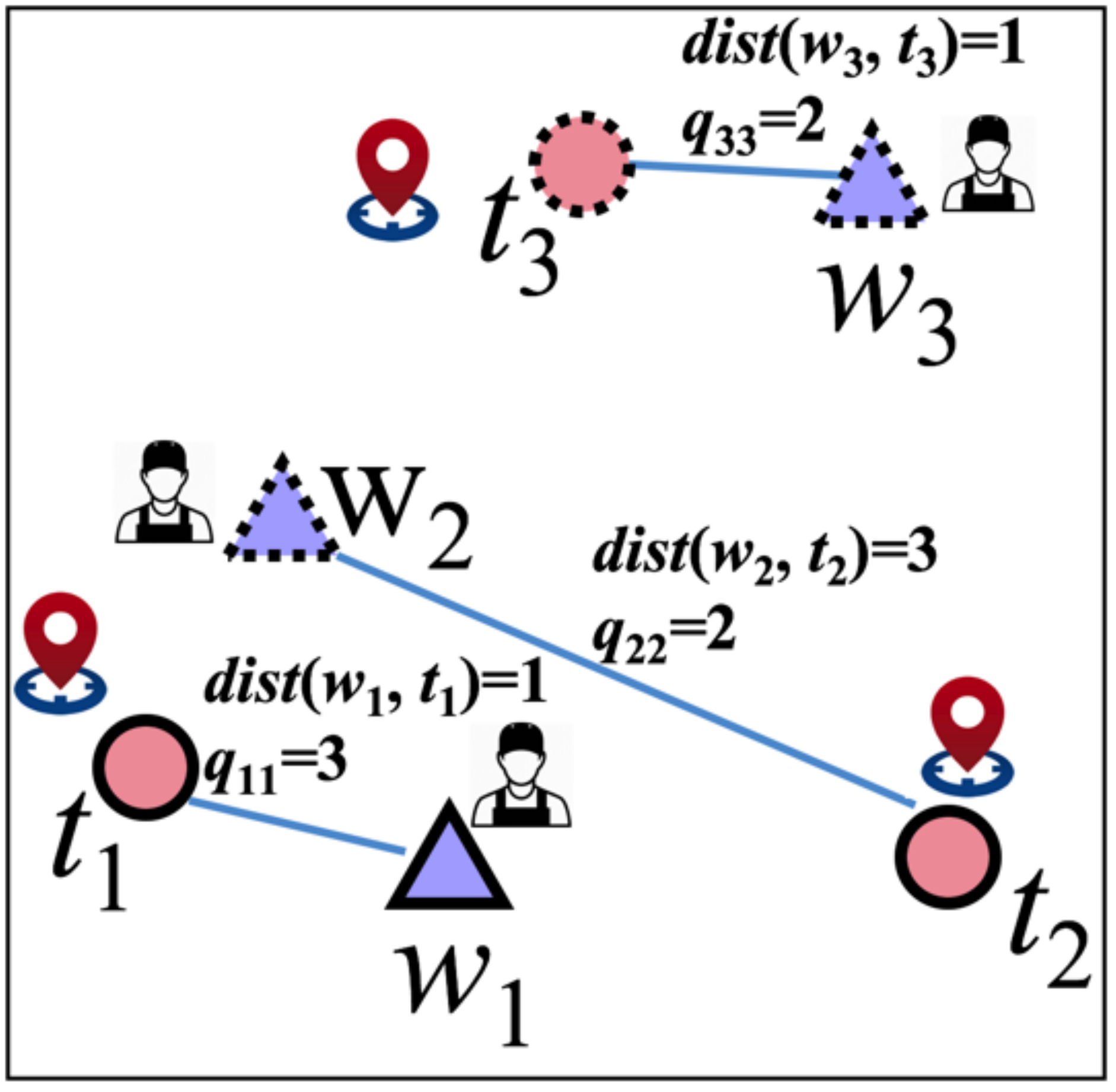}}
		\label{fig:fig1b}
	}
	\vspace{-2ex}
	\caption{\small Locally Optimal Worker-and-Task Assignments in the Spatial Crowdsourcing System.}
	\label{fig:example}\vspace{-5ex}
\end{figure}

Given a maximum reward budget at each time instance, a spatial
crowdsourcing problem is to assign workers to tasks at
both current and future timestamps, $p$ and $(p+1)$, respectively,
such that the overall quality score of assignments is maximized while
the reward for workers is under an allocated budget. 

}

In Example 1, the traditional spatial crowdsourcing approaches
\cite{cheng2014reliable, kim2014mediaq}
considered the worker-and-task assignments only
based on workers/tasks that are available in the system at each
time instance. For example, in the first time instance, $p$, of Figure \ref{fig:fig1a},
only $t_1$, $t_2$, and $w_1$ exist.
Therefore, we assign worker $w_1$ to task $t_1$ (rather than task
$t_2$), since it holds that $dist(w_1, t_1)< dist(w_1, t_2)$ and
$q_{11}>q_{12}$ (i.e., worker $w_1$ can accomplish task $t_1$ with
lower budget consumption and higher quality, compared with task
$t_2$, as given in Table \ref{tab:distance_quality}). Next, in the
second time instance of Figure \ref{fig:fig1b}, two new tasks, $t_2$ and $t_3$,
and two new workers, $w_2$ and $w_3$, become available in the system at
timestamp $(p+1)$. Traditional spatial crowdsourcing approaches
would create assignment pairs $\langle w_2, t_2\rangle$ and $\langle
w_3, t_3\rangle$ at this time instance (see lines in Figure
\ref{fig:fig1b}). As a result, such an assignment strategy at two
separate time instances leads to the overall traveling cost 5 $(=1+3+1)$ and
overall quality score 7 $(=3+2+2)$.

Note that the assignment strategy above does not take into account
future workers/tasks that may join the system at a later time instance. In \cite{kazemi2012geocrowd}, it is shown that an optimal assignment strategy exists for a single time instance but this local optimality may not yield a global optimal assignment across all time instances. 
In our running example, at
timestamp $p$, a clairvoyant algorithm that knows the future workers/tasks that arrive
at timestamp $(p+1)$, may provide a better 
global assignment strategy (i.e., with lower overall traveling cost and higher
overall quality score). Based on this observation, in this paper, we will
formulate a new optimization problem, namely
\textit{maximum quality task assignment} (MQA) with the goal of
assigning moving workers to spatial tasks under a budget constraint to achieve a better global assignment across multiple time instances.

\begin{table}
	\centering \vspace{-7ex}
	{\scriptsize
		\caption{\small Distances and Quality Scores of Worker-and-Task
			Pairs}\label{tab:distance_quality}\vspace{-2ex}
		\begin{tabular}{c|c|c}
			{\bf worker-and-task pair, $\langle w_i, t_j\rangle$} & {\bf  distance, $dist(w_i, t_j)$} & {\bf quality score, $q_{ij}$}\\
			\hline \hline
			$\langle w_1, t_1\rangle$ & 1 & 3\\
			$\langle w_1, t_2\rangle$ & 2 & 2\\
			$\langle w_1, t_3\rangle$ & 4 & 2\\\hline
			$\langle w_2, t_1\rangle$ & 1 & 4\\
			$\langle w_2, t_2\rangle$ & 3 & 2\\
			$\langle w_2, t_3\rangle$ & 2 & 1\\\hline
			$\langle w_3, t_1\rangle$ & 5 & 2\\
			$\langle w_3, t_2\rangle$ & 3 & 1\\
			$\langle w_3, t_3\rangle$ & 1 & 2\\
			\hline
		\end{tabular}
	}
	\vspace{-7ex}
\end{table}

\vspace{0.5ex}\noindent\textbf{Example 2 (The Maximum Quality Task Assignment Problem)} {\it In the example of Figure
\ref{fig:example}, the MQA problem is to maximize the overall quality score of assignments under traveling cost budget constraints across multiple time instances. As shown in Figure
\ref{fig:tasks_m}, at timestamp $p$, we can have the
prediction-based assignments: $\langle w_2, t_1\rangle$, $\langle
w_1, t_2\rangle$, and $\langle w_3, t_3\rangle$, which can achieve
a better global assignment with smaller traveling cost $4$ $(=1+2+1)$ and
higher quality score $8$ $(=4+2+2)$, compared to locally optimal
assignments (without prediction) in Figure \ref{fig:example} with
the traveling cost, $5$, and the quality score $7$.

}

As shown in \cite{kazemi2012geocrowd} and from the example above, we can see that even an optimal local
assignment based on the currently available tasks/workers
at each time instance may not lead to a global optimal solution. 
In \cite{kazemi2012geocrowd} , some heuristics based on {\bf past} worker/task distributions (e.g., based on location entropy) were proposed to address this challenge.  Here, instead, we strive to predict/estimate {\bf future} task/worker distributions to provide a better global assignment strategy.  Moreover, in \cite{kazemi2012geocrowd}, the objective was to maximize the number of assigned tasks without
considering the traveling budget constraint and the quality score of task assignment. In contrast, the optimization objective of MQA is to maximize the assignment quality score under budget constraints.

\begin{figure}[t!]\vspace{-7ex}
    \centering
 \setcounter{subfigure}{-1}
    \subfigure{
       \includegraphics[scale=0.2]{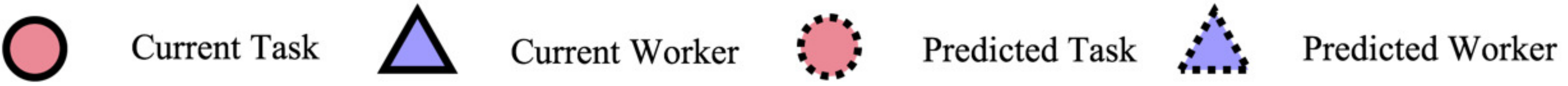}
       \label{fig:bar}
        }\\\vspace{-1ex}
       \scalebox{0.17}[0.17]{\includegraphics{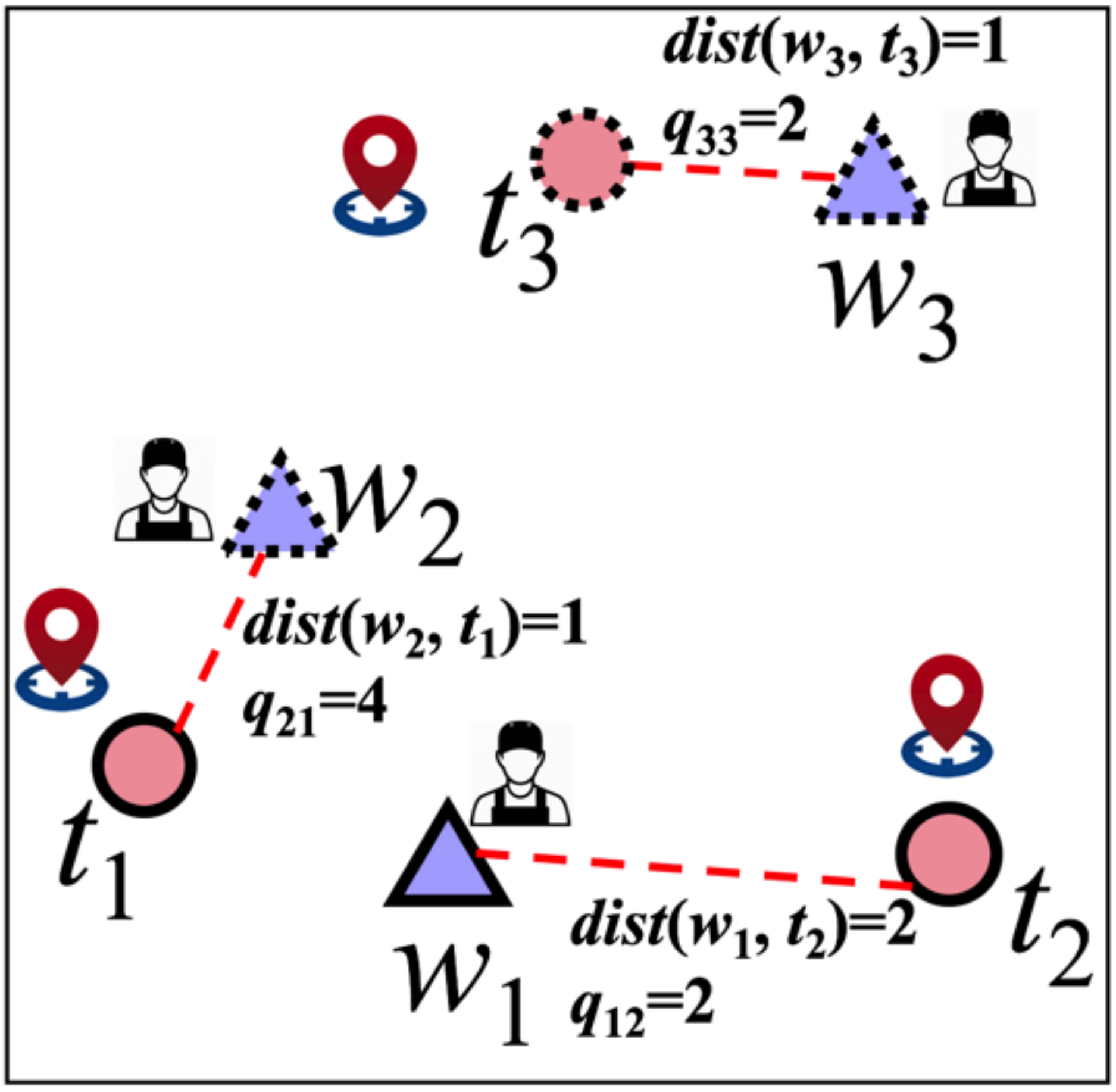}}
\vspace{-2ex}
    \caption{\small Globally Optimal Assignments for the Maximum Quality Task Assignment (MQA) Problem.}
    \label{fig:tasks_m}\vspace{-4ex}
\end{figure}

Different from prior studies in spatial crowdsourcing, the MQA problem requires designing an
accurate prediction approach for estimating future location
distributions of tasks and workers (and the quality distributions of
worker-and-task assignment pairs as well), and considering the
assignments over the estimated location/quality variables of future
tasks/workers,
which is quite challenging. Furthermore, in this paper, we prove
that the MQA problem is NP-hard for any given time instance, by reducing it from the 0-1
Knapsack problem \cite{vazirani2013approximation}. As a result, the
MQA problem is not tractable. Therefore, in order to efficiently
tackle the MQA problem, we will propose effective heuristics, including \textit{MQA greedy} and
\textit{MQA divide-and-conquer} approaches, over both current and
predicted workers/tasks, which can efficiently compute a better global assignment with high quality scores under the budget
constraints.

Specifically, we make the following contributions.

\begin{itemize}

\item We formally define the \textit{maximum quality task assignment} (MQA) problem in Section \ref{sec:problem_def}. We prove that the
MQA problem is NP-hard in Section
\ref{sec:reduction}.

\item We propose an effective grid-based prediction approach to estimate location/quality distributions/statistics for future tasks /workers in Section
\ref{sec:prediction}.

\item We design an efficient \textit{MQA greedy } algorithm to iteratively select
one best assignment pair each time over current/future tasks/workers
in Section \ref{sec:greedy}.

\item We illustrate a novel \textit{MQA divide-and-conquer } algorithm
to recursively divide the problem into subproblems and merge
assignment results in subproblems in Section
\ref{sec:D&C}.

\item We verify the effectiveness
and efficiency of our proposed MQA approaches with extensive
experiments on real and synthetic data sets in Section
\ref{sec:exp}.\vspace{-1ex}
\end{itemize}

In addition to the contributions listed above, in this paper, we
review previous works in spatial crowdsourcing in Section
\ref{sec:related}, and conclude in Section \ref{sec:conclusion}.

\vspace{-2ex}
\section{Problem Definition}
\label{sec:problem_def}

%

\subsection{Dynamically Moving Workers}

\begin{definition} $($Dynamically Moving Workers$)$ Let
    $W_p$ $=\{w_1,$ $w_2,$ $...,$ $w_n\}$ be a set of $n$ moving
    workers at timestamp $p$. Each worker $w_i$ ($1\leq i\leq n$)  is located at position
    $l_i(p)$ at timestamp $p$, and freely moves with velocity $v_i$.
     \label{definition:worker}
\end{definition}

In Definition \ref{definition:worker}, worker $w_i$ can dynamically
move with speed $v_i$ in any direction. At each timestamp $p$,
worker $w_i$ are located at positions $l_i(p)$. Workers can freely
join or leave the spatial crowdsourcing system.

\subsection{Time-Constrained Spatial Tasks}


\begin{definition}
$($Time-Constrained Spatial Tasks$)$ Let $T_p=\{t_1,$ $t_2,$ $...,$
$t_m\}$ be a set of time-constrained spatial tasks at timestamp $p$.
Each spatial task $t_j$ ($1\leq j\leq m$) is located at a specific
location $l_j$, and workers are expected to reach the location of
task $t_j$ before the deadline $e_j$.
\label{definition:task}
\end{definition}

In Definition \ref{definition:task}, a task requester creates a
time-constrained spatial task $t_j$ (e.g., taking a photo), which
requires workers to be physically at the specific location $l_j$ before the deadline $e_j$. In this
paper, we assume that each spatial task can be performed by a single
worker.

\subsection{The Maximum Quality Task Assignment Problem}

In this section, we will formalize the problem of
\textit{maximum quality task assignment} (MQA), which
assigns time-constrained spatial tasks to spatially scattered
workers with the objective of maximizing the overall quality score of assignments under the traveling budget constraint.

\vspace{0.5ex}\noindent {\bf Task Assignment Instance Set.} Before
we present  the MQA problem, we first introduce the notion of the
\textit{task assignment instance set}.

\begin{definition}
$($Task Assignment Instance Set, $I_p$$)$ At timestamp $p$, given a
worker set $W_p$ and a task set $T_p$, a \textit{task assignment
instance set}, $I_p$, is a set of valid worker-and-task
assignment pairs in the form $\langle w_i, t_j\rangle$, where each
worker $w_i \in W_p$ is assigned to at most one task $t_j\in T_p$,
and each task $t_j \in T_p$ is accomplished by at most one worker
$w_i \in W_p$.
\label{definition:instance}
\end{definition}\vspace{-2ex}

Intuitively, $I_p$ in Definition \ref{definition:instance} is one
possible (valid) worker-and-task assignment between worker set $W_p$
and task set $T_p$. A valid assignment pair $\langle w_i,
t_j\rangle$ is in $I_p$, if and only if this pair satisfies the
condition that worker $w_i$ can reach the location $l_j$ of task
$t_j$ before the deadline $e_j$.

{\bf \it \underline{The Reward of the Worker-and-Task Assignment
Pair}.} As discussed earlier in Section \ref{sec:intro}, we assume
that each worker-and-task assignment pair $\langle w_i, t_j\rangle$
is associated with a traveling cost, $c_{ij}$, which corresponds to
the reward (e.g., monetary or points) to cover the
transportation fee from $l_i(p)$ to $l_j$ (e.g., cost of gas or
public transportation). Without loss of generality we assume that the reward, $c_{ij}$, is computed by:
$c_{ij} = C\cdot dist(l_i(p), l_j)$, where $C$ is the unit cost per
mile, and $dist(l_i(p),$ $l_j)$ is the distance between locations of
worker $w_i$ and task $t_j$. For simplicity, we will use the
Euclidean distance function $dist(l_i(p),$ $l_j)$ in this paper.

{\bf \it \underline{The Quality Score of the Worker-and-Task
Assignment Pair}.} Each worker-and-task assignment pair $\langle
w_i, t_j\rangle$ is also associated with a \textit{quality score},
denoted as $q_{ij}$ ($\in [0, 1]$), which indicates the quality of
the task $t_j$ that is completed by worker $w_i$. 
Due to different types of tasks with various difficulty
levels (e.g., taking photos vs. repairing houses) and different expertise or competence of workers in performing
tasks, we assume that
different worker-and-task pairs $\langle w_i, t_j\rangle$ may be
associated with diverse quality scores $q_{ij}$.

\vspace{0.5ex}\noindent {\bf The MQA Problem.} Next, we provide the
definition of our \textit{maximum quality task assignment}
(MQA) problem.

\begin{definition}
$($Maximum Quality Task Assignment (MQA) $)$ Given a set of time
instances $P$ and a maximum reward budget $B$ for each time instance, the
problem of \textit{maximum quality task assignment} (MQA)
is to assign the available workers in $W_p$ to tasks in
$T_p$ to provide a task assignment instance
set, $I_p$, at timestamp $p \in P$, such
that:

    \begin{enumerate}
        \item at any timestamp $p\in P$, each worker $w_i\in W_p$ is assigned to at most one spatial
        task $t_j\in T_p$ such that his/her arrival time at location
        $l_j$ is before deadline $e_j$;
        \item at timestamp $p\in P$, the total reward (i.e., the traveling cost) of all the assigned workers does
        not exceed budget $B$, that is, $\sum_{\forall \langle w_i, t_j\rangle \in
                I_p}c_{ij}$ $\leq B$; and
        \item the overall quality score of the assigned tasks of all timestamps in $P$ is maximized, that is,\vspace{-1ex}
        {\scriptsize
        \begin{eqnarray}
        \text{maximize } \sum_{\forall p \in P}\sum_{\forall \langle w_i, t_j\rangle \in I_p} q_{ij}.\label{eq:tbc_sc}\vspace{-4ex}
        \end{eqnarray}
    }
    \end{enumerate}
    \label{definition:PA_SC}\vspace{-1ex}
\end{definition}

Intuitively, the MQA problem (given in Definition
\ref{definition:PA_SC}) assigns workers to tasks at multiple time instances in  $P$, such that (1) at each time instance, an assignment
pair $\langle w_i, t_j\rangle$ is valid (i.e., satisfying the time
constraints, $e_j$); (2) at each time instance, the total reward of
assignment pairs is less than or equal to the maximum budget $B$;
and (3) the overall quality score of assignments of all the time instances is
maximized.

As discussed earlier in Section \ref{sec:intro}, it may not be globally
optimal to simply optimize the
assignments at each individual
time instance separately. The MQA approaches aim to provide a better global assignment by taking into account tasks/workers
not only at the current timestamp $p$, but also at the future
timestamp $(p+i) \in P$.

Table \ref{table0} summarizes the commonly used symbols.

\begin{table}
    \begin{center}\vspace{-5ex}
        \caption{\small Symbols and Descriptions.} \vspace{-2ex}\label{table0}
        {\small\scriptsize
            \begin{tabular}{l|l}
                {\bf Symbol} & {\bf \qquad \qquad \qquad\qquad\qquad Description} \\ \hline \hline
                $T_p$   & a set of $m$ time-constrained spatial tasks $t_j$ at timestamp $p$\\
                $W_p$   & a set of $n$ dynamically moving workers $w_i$ at timestamp $p$\\
                $\hat{w_i}$ (or $\hat{t_j}$) & the predicted worker (or task) \\
                $\tilde{w_i}$ (or $\tilde{t_j}$) & the worker (or task) in either current or next time instance\\
                $I_p$   & the task assignment instance set at timestamp $p$ \\
                $e_j$   & the deadline of arriving at the location of task $t_j$\\
                $l_i(p)$  & the position of worker $w_i$ at timestamp $p$\\
                $l_j$   & the position of task $t_j$\\
                $v_i$   & the velocity of worker $w_i$ \\
                $c_{ij}$ & the traveling cost from the location of worker $w_i$ to that of task $t_j$\\
                $q_{ij}$ & the quality score of assigning worker $w_i$ to perform task $t_j$\\
                $B$   & the maximum reward budget (i.e., traveling cost) at each time instance \\
                $C$ & the unit price of the traveling cost by workers\\
                $P$ & a set of time instances\\
                $w$ & the size of the sliding window to do the
                prediction\\
                \hline
                \hline
            \end{tabular}
        }
    \end{center}\vspace{-7ex}
\end{table}

\subsection{Hardness of the MQA Problem}
\label{sec:reduction}

In order to give a sense on the size of the solution space, with $n$ dynamically moving workers and $m$ time-constrained spatial
tasks, in the worst case, there are an exponential number of
possible worker-and-task assignment strategies, which leads to high
time complexity (i.e., $O((m + 1)^n)$). Then, we prove
that the MQA problem is NP-hard, by reducing it from a well-known
NP-hard problem, \textit{0-1 Knapsack problem}
\cite{vazirani2013approximation}.

\begin{lemma} (Hardness of the MQA Problem)
The maximum quality task assignment (MQA) problem is
NP-hard. \label{lemma:np}
\end{lemma}

{\small
\begin{proof}
	Please refer to Appendix A.
\end{proof}
}

From Lemma \ref{lemma:np}, we can see that the MQA problem is
NP-hard, and thus intractable. Alternatively, we will later design
efficient heuristics to tackle the MQA problem and
achieve better global assignment strategies.

\vspace{-1ex}
\subsection{Framework}
\label{subsec:framework}

Figure \ref{alg:framework} illustrates a general framework, namely
procedure {\sf MQA\_Framework}, for solving the MQA problem,
which assigns workers with spatial tasks at multiple time instances in $P$, based on the predicted location/quality
distributions of workers/tasks.

Specifically, at timestamp $p$, we first retrieve a
set, $T_p$, of available spatial tasks, and a set, $W_p$, of
available workers (lines 1-3). Here, the task set $T_p$ (or worker
set $W_p$) contains both tasks (or workers) that have not been
assigned at the last time instance and the ones that newly arrive at the
system after the last time instance (note: those workers who finished tasks
in previous time instances are also treated as ``new workers'' that join the
system, thus the workers can continuously contribute to the platform). In order to achieve better global assignments, we need to predict workers and tasks at a future timestamp $(p+1)$
that newly join the system, and obtain two sets $W_{p+1}$ and
$T_{p+1}$, respectively (line 4).

Subsequently, with both current and future sets of tasks/workers (i.e.,
$T_p$/$W_p$ and $T_{p+1}$/$W_{p+1}$, respectively), we can apply our
proposed algorithms in this paper (including \textit{MQA greedy} and
\textit{MQA divide-and-conquer}), and retrieve better
assignment pairs in assignment instance set $I_p$ (line 5). Finally, we notify each worker $w_i$
to do his/her task $t_j$ (lines 6-7).

\begin{figure}[ht]
    \begin{center}\vspace{-3ex}
        \begin{tabular}{l}
            \parbox{3.1in}{
                \begin{scriptsize}
                    \begin{tabbing}
                        12\=12\=12\=12\=12\=12\=12\=12\=12\=12\=12\=\kill
                        {\bf Procedure {\sf MQA\_Framework}} \{ \\
                        \> {\bf Input:} a set of time instances $P$\\
                        \> {\bf Output:} a worker-and-task assignment strategy across the time instances in $P$\\
                        \> (1) \> \> \textbf{for} time instance $p \in P$\\
                        \> (2) \> \> \> retrieve all the available spatial tasks in set $T_p$\\
                        \> (3) \> \> \> retrieve all the available workers in set $W_p$\\
                        \> (4) \> \> \> predict new future tasks/workers in $T_{p+1}$ and $W_{p+1}$ at the next time instance\\
                        \> (5) \> \> \> apply the \textit{MQA greedy} or \textit{MQA divide-and-conquer} approach to obtain a  \\
                        \> \> \> \> \> assignment instance set, $I_p$, w.r.t. $T_p$, $W_p$, $T_{p+1}$ and $W_{p+1}$\\
                        \> (6) \> \> \> \textbf{for} each pair $\langle w_i, t_j \rangle$ in $I_p$\\
                        \> (7) \> \> \> \> inform worker $w_i$ to perform task $t_j$\}\\
                        
                    \end{tabbing}
                \end{scriptsize}
            }
        \end{tabular}
    \end{center}\vspace{-5ex}
    \caption{\small A Framework for Tackling the MQA Problem.}
    \label{alg:framework}\vspace{-2ex}
\end{figure}

\section{The Grid-based Worker/Task \\Prediction Approach}
\label{sec:prediction}

In order to achieve better global assignments in our MQA problem, we
need to accurately predict the future status of workers/tasks that
newly join the spatial crowdsourcing system. Specifically, in this
section, we will introduce a grid-based worker/task prediction
approach, which can efficiently and effectively estimate the number
of future workers/tasks, location distributions of future
workers/tasks, and quality score distributions w.r.t. future
worker-and-task pairs (and their existence probabilities as well).

\begin{table}[t!]
	\centering \vspace{-5ex}
	{\scriptsize
		\caption{\small Example of Grid-Based Prediction}\label{tab:sample_prediction}\vspace{-2ex}
		\begin{tabular}{c|c|c}
			{\bf Cell} & {\bf  Previous Worker Numbers} & {\bf Predicted Worker Number}\\
			\hline \hline
			$C_1$ & [4, 3, 4] & 4\\
			$C_2$ & [2, 3, 3] & 3\\
			$C_3$ & [0, 1, 0] & 0\\
			$C_4$ & [1, 1, 1] & 1\\
			\hline
		\end{tabular}
	}
	\vspace{-6ex}
\end{table}

\begin{figure}[t!]\vspace{-6ex}
	\centering
	\scalebox{0.18}[0.18]{\includegraphics{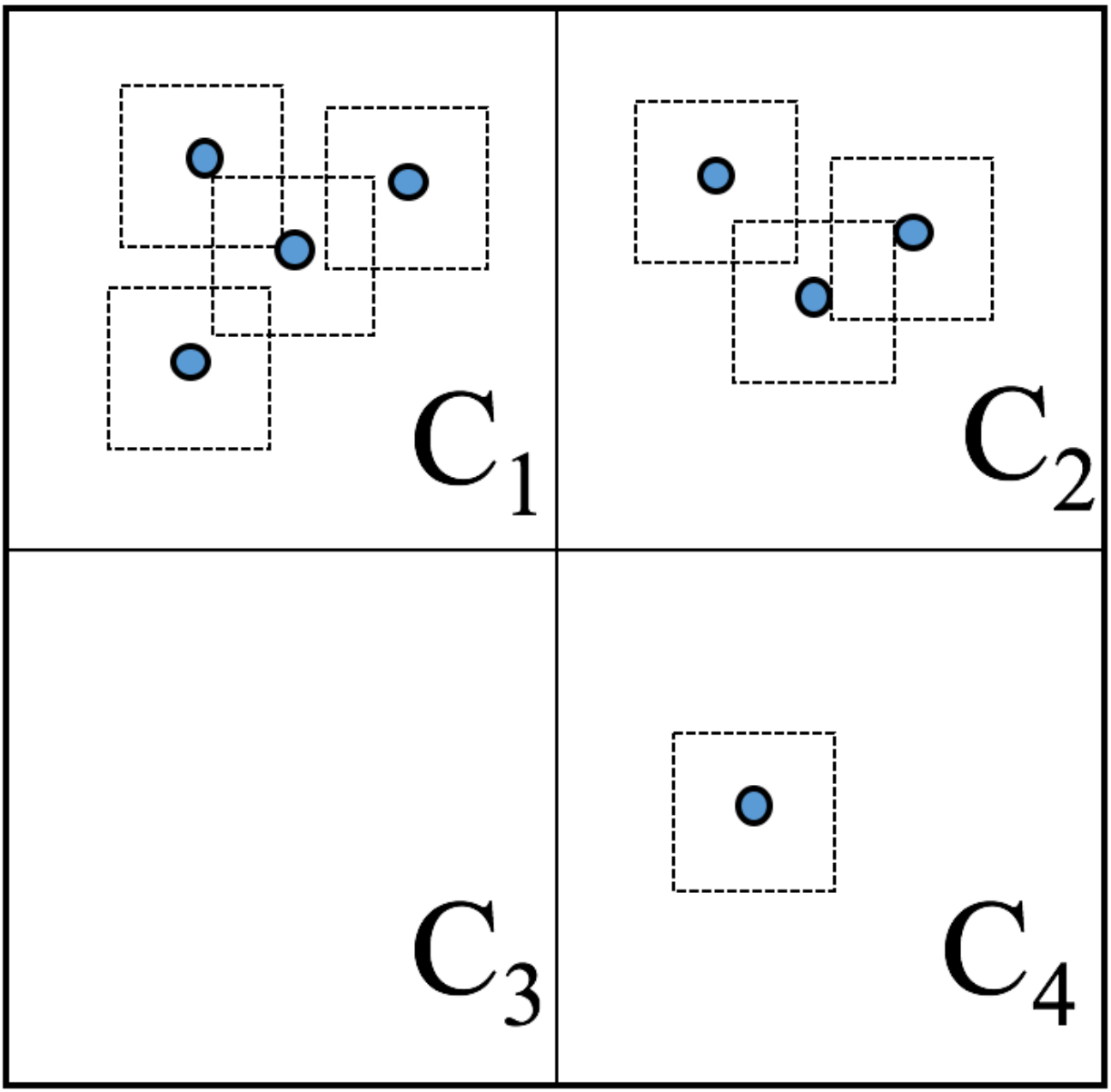}}
	\vspace{-2ex}
	\caption{\small Example of Predicted Workers.}
	\label{fig:prediction_worker}\vspace{-5ex}
\end{figure}

\subsection{The Grid-Based Prediction Algorithm}
\label{subsec:prediction}

In this section, we discuss how to predict the number of tasks
/workers, and their location distributions in a 2-dimensional data
space $\mathcal{U}=[0,1]^2$. In particular, we consider a grid
index, $\mathcal{I}$, over tasks and workers, which divides the data
space $\mathcal{U}$ into $\gamma^2$ cells, each with the side length
$1/\gamma$, where the selection of the best $\gamma$ value can be
guided by a cost model in \cite{cheng2014reliable}. Next, we
estimate the potential workers/tasks that may fall into each cell at a
future timestamp, which is inferred from historical data of
the most recent sliding window of size $w$.

In particular, our grid-based prediction algorithm first predicts
the future numbers of workers/tasks from the latest $w$ worker/task
sets in each cell, then generates possible worker (or task) samples
in $W_{p+1}$ (or $T_{p+1}$) for each cell of the grid index
$\mathcal{I}$.

First, we initialize a worker set $W_{p+1}$ and a task set $T_{p+1}$
at the future time instance with empty sets. Subsequently, within each cell
$cell_i$, we can obtain its $w$ latest worker counts,
$|W_{p-w+1}^{(i)}|$, $|W_{p-w+2}^{(i)}|$, ..., and $|W_{p}^{(i)}|$,
which form a \textit{sliding window} of a time series (with size
$w$). Due to the temporal correlation of worker counts in the
sliding window, in this paper, we utilize the \textit{linear
regression} \cite{lawson1974solving} over these $w$ worker counts to
predict the future number, $|W_{p+1}^{(i)}|$, of workers in cell,
$cell_i$, that newly join the system at timestamp $(p+1)$. Note
that other prediction methods can be also plugged into our
grid-based prediction framework, which we plan to study in
our future work. Similarly, we can estimate the number,
$|T_{p+1}^{(i)}|$, of tasks in $cell_i$ at
timestamp $(p+1)$.

According to the predicted numbers of workers/tasks, we can
uniformly generate $|W_{p+1}^{(i)}|$ worker samples (or
$|T_{p+1}^{(i)}|$ task samples) within each cell $cell_i$, and add
them to the predicted worker set $W_{p+1}$ (or task set $T_{p+1}$).
We use sampling with replacement to generate worker/task samples, which means two samples can be generated at the same location. For the
pseudo code of the MQA prediction algorithm, please refer to
Appendix B.

\vspace{1ex}\noindent \textbf{Example 3 (The Grid-Based Prediction)} 
{\it 
	Consider that a space is divided into 4 cells, $C_1$ to $C_4$, as shown in Figure \ref{fig:prediction_worker}. Based on the historical records on the number of workers in each cell $C_i$, we want to predict the workers that will appear in $C_i$ at the next time instance. Table \ref{tab:sample_prediction} presents the number of workers for each cell in current time instance $p$ and previous two time instances $p-1$ and $p-2$. For example, there 4 workers at time $p-2$, 3 workers at time $p-1$ and 4 workers at time $p$ in cell $C_1$. Then, we predict the number of workers in cell $C_1$ at time $p+1$ is 4. We uniformly generate 4 worker samples. After generating predicted workers for each cell, we can capture the distribution of workers as shown in Figure \ref{fig:prediction_worker}.
}

\vspace{0.5ex} \noindent \textbf{Location Distributions of the
Predicted Workers/Tasks.} As each cell is small, from the global view, the distributions of tasks/workers in the entire space are approximately captured. However, in each cell, discrete samples may be of
small sample sizes, which may lead to low prediction accuracy. For
example, if the sample size is only 1, then different possible
locations of this one single sample (generated in the cell) may
dramatically affect our MQA assignment results.

Inspired by this, instead of using discrete samples predicted, in
this paper, we will alternatively consider continuous
\textit{probability density function} (pdf) for location
distributions of workers/tasks. Specifically, we apply the
\textit{kernel density estimation} over samples in each cell to
describe the distributions of samples' locations. That is, centered
at each sample $s_i$ generated in $cell_i$, we can
obtain the continuous pdf function of this worker/task sample, $f(x)
= \prod_{r=1}^2 \left(\frac{1}{h_r} K\left(
\frac{x[r]-s[r]}{h_r}\right)\right)$, where $h_r$ ($\in (0, 1)$) is
the bandwidth on the $r$-th dimension, and function $K(\cdot)$ is a
\textit{uniform kernel function} \cite{hansen2009lecture}, given by
$K(u)=\frac{1}{2}\cdot \mathbf{1}{(|u|\leq 1)}$. Here,
$\mathbf{1}{(|u|\leq 1)} = 1$, when $|u|\leq 1$ holds. Note that
the choice of the kernel function is not significant for
approximation results \cite{hansen2009lecture}, thus, in this paper,
we use uniform kernel function $K(\cdot)$.
From the pdf function $f(x)$ of each sample $s_i\in cell_i$,
each dimension $r$ can be bounded by an interval $[s_i[r]- h_r,
s_i[r]+ h_r]$.

Typically, in the literature \cite{hansen2009lecture}, we set $h_r = \hat{\sigma}
C_v(k)n^{-1/(2v+1)}$, where $\hat{\sigma}$ is the standard deviation of
samples (derived from current worker/task statistics), $v$ is the
order of the kernel ($v$ is set to 2 here), and $C_v(k) = 1.8431$ ($=
2\left(\frac{\pi^{1/2}(v!)^3R(k)}{2v(2v)!k_v^2(k)}\right)^{1/(2v+1)}$,
for $k_v(k) = 1/3, R(k) = 1/2$ with Uniform kernel functions).

\subsection{Statistics of the Predicted Workers/Tasks}
\label{subsec:statistics}

For the ease of the presentation, in this paper, we denote those
future workers $w_i$ and tasks $t_j$ that are predicted as
$\widehat{w_i}$ and $\widehat{t_j}$, respectively.

Due to the predicted future workers/tasks, in our PB-SC problem, we
need to consider worker-and-task assignment pairs that may involve
predicted workers/tasks. That is, we have 3 cases, $\langle
\widehat{w_i}, t_j\rangle$, $\langle w_i, \widehat{t_j}\rangle$, and
$\langle \widehat{w_i}, \widehat{t_j}\rangle$, where $\widehat{w_i}$
and $\widehat{t_j}$ are the predicted worker and task samples,
respectively, following uniform distributions represented by kernel
functions $K(\cdot)$ (as mentioned in Section
\ref{subsec:prediction}).

Due to the existence of these predicted workers/tasks, the traveling
costs and the quality scores of assignment pairs now become random
variables (rather than fixed values). In this section, we will
discuss how to obtain statistics (e.g., mean and variance) of
traveling costs, quality scores, and confidences, associated with
assignment pairs.

\vspace{0.5ex} \noindent \textbf{The Traveling Cost of Pairs with
the Predicted Workers/Tasks.} The traveling cost,
$\widehat{c_{ij}}$, of worker-and-task pairs involving the predicted
workers/tasks (i.e., $\langle \widehat{w_i}, t_j\rangle$, $\langle
w_i, \widehat{t_j}\rangle$, or $\langle \widehat{w_i},
\widehat{t_j}\rangle$) can be given by $C\cdot dist(\widehat{w_i},
t_j)$, $C\cdot dist(w_i, \widehat{t_j})$, or $C\cdot
dist(\widehat{w_i}, \widehat{t_j})$, respectively.

We discuss the general case of computing statistics of variable
$\widehat{c_{ij}}=C\cdot dist(\widehat{w_i}, \widehat{t_j})$. Since
it is nontrivial to calculate the statistics of the Euclidean
distance between variables $\widehat{w_i}$ and $\widehat{t_j}$, we
alternatively consider statistics (mean and variance) of the squared
Euclidean distance variable $Z^2=dist^2(\widehat{w_i},
\widehat{t_j})$ ($ = \sum_{r=1}^2 (\widehat{w_i}[r] -
\widehat{t_j}[r])^2$), where $\widehat{w_i}$ and $\widehat{t_j}$ are
two variables uniformly residing in a 2D space.

{\it \underline{The Computation of Mean $E(Z^2)$.}} Let variable
$Z_r = \widehat{w_i}[r] - \widehat{t_j}[r]$, for $r=1, 2$, whose
mean $E(Z_r)$ and variance $Var(Z_r)$ can be easily computed (i.e.,
$E(Z_r)=s_i[r] - s_j[r]$ and $Var(Z_r)
=\frac{(h_r(\widehat{w_i}[r]))^2+(h_r(\widehat{t_j}[r]))^2}{3}$
respectively).

Then, we have $Z^2 = Z_1^2 + Z_2^2$. Thus, the mean $E(Z^2)$ can be
given by:\vspace{-5ex}

{\scriptsize
\begin{eqnarray}
E(Z^2) = E(Z_1^2) + E(Z_2^2).\label{eq:mean_Z_square}
\end{eqnarray}\vspace{-5ex}
}

{\it \underline{The Computation of Variance $Var(Z^2)$.}} Moreover,
for variance $Var(Z^2)$, it holds that:\vspace{-2ex}

{\scriptsize
\begin{eqnarray}
Var(Z^2) &\hspace{-2ex}=& \hspace{-2ex} E(Z^4) - (E(Z^2))^2 \label{eq:variance_Z_square}\\
&\hspace{-2ex}=& \hspace{-2ex}E((Z_1^2 + Z_2^2)^2) - (E(Z^2))^2\notag\\
&\hspace{-2ex}=& \hspace{-2ex}E(Z_1^4)+2\cdot E(Z_1^2)\cdot E(Z_2^2)
+ E(Z_2^4) - (E(Z^2))^2.\notag
\end{eqnarray}\vspace{-4ex}
}

From Eqs.~(\ref{eq:mean_Z_square}) and (\ref{eq:variance_Z_square})
above, the remaining issues are to compute $E(Z_r^2)$ and $E(Z_r^4)$
(for $r=1, 2$).

{\it \underline{The Computation of $E(Z_r^2)$.}} For $E(Z_r^2)$,
since $Z_r = \widehat{w_i}[r] - \widehat{t_j}[r]$, we
have:\vspace{-2ex}

{\scriptsize
\begin{eqnarray}
E(Z_r^2) &=& Var(Z_r) + (E(Z_r))^2 \label{eq:mean_Zr_square}\\
&=& Var(\widehat{w_i}[r]) + Var(\widehat{t_j}[r]) +
(E(\widehat{w_i}[r])- E(\widehat{t_j}[r]))^2.\notag
\end{eqnarray}\vspace{-4ex}
}

{\it \underline{The Computation of $E(Z_r^4)$.}} For $E(Z_r^4)$, we
can derive that:\vspace{-2ex}

{\scriptsize
\begin{eqnarray}
E(Z_r^4) &=& E((\widehat{w_i}[r] - \widehat{t_j}[r])^4)\notag\\
&=& E(\widehat{w_i}[r]^4)-4\cdot E(\widehat{w_i}[r]^3)\cdot
E(\widehat{t_j}[r])\notag\\
&& + 6\cdot E(\widehat{w_i}[r]^2)\cdot E(\widehat{t_j}[r]^2) -
4\cdot E(\widehat{w_i}[r])\cdot E(\widehat{t_j}[r]^3)\notag\\
&& + E(\widehat{t_j}[r]^4).\label{eq:mean_Zr_4}
\end{eqnarray}\vspace{-4ex}
}

In Eq.~(\ref{eq:mean_Zr_4}), variable $\widehat{w_i}[r]$ follows the
uniform distribution within bound $[lb\_w, ub\_w]$ (for the $r$-th
dimension of uniform kernel function $K(\cdot)$ in Section
\ref{subsec:prediction}). We can thus infer that:\vspace{-2ex}

{\scriptsize
\begin{eqnarray}
E(\widehat{w_i}[r]^4) &=& \int_{lb\_w}^{ub\_w} x^4
\frac{1}{ub\_w-lb\_w} dx = \frac{ub\_w^5
-lb\_w^5}{5(ub\_w-lb\_w)},\notag\\
E(\widehat{w_i}[r]^3) &=& \int_{lb\_w}^{ub\_w} x^3
\frac{1}{ub\_w-lb\_w} dx = \frac{ub\_w^4
-lb\_w^4}{4(ub\_w-lb\_w)},\notag\\
E(\widehat{w_i}[r]^2) &=& \int_{lb\_w}^{ub\_w} x^2
\frac{1}{ub\_w-lb\_w} dx = \frac{ub\_w^3
-lb\_w^3}{3(ub\_w-lb\_w)},\notag
\end{eqnarray}\vspace{-2ex}
}

\noindent where $[lb\_w, ub\_w] = [s_i[r]-h_r(\widehat{w_i}[r]),
s_i[r]+h_r(\widehat{w_i}[r])]$.

Similarly, we can also obtain $E(\widehat{t_j}[r]^4)$,
$E(\widehat{t_j}[r]^3)$, and $E(\widehat{t_j}[r]^2)$ for task
$\widehat{t_j}[r]$. We omit it here.

This way, by substituting Eqs.~(\ref{eq:mean_Zr_square}) and
(\ref{eq:mean_Zr_4}) into Eqs.~(\ref{eq:mean_Z_square}) and
(\ref{eq:variance_Z_square}), we can obtain mean $E(Z^2)$ and
variance $Var(Z^2)$ of the squared Euclidean distance between two
Uniform distributions.

\vspace{0.5ex} \noindent \textbf{Quality Scores of Pairs with the
Predicted Workers/Tasks.} We consider the three cases to compute
statistics of quality score distributions.

{\it \underline{Case 1: $\langle \widehat{w_i}, t_j\rangle$.}} In
this case, at the current timestamp $p$, we can obtain all
the $n_i$ workers $w_i$ that can reach task $t_j$. Then, we use
quality scores, $q_{ij}$, of their corresponding worker-and-task
pairs $\langle w_i, t_j\rangle$ as samples (each with probability
$1/n_i$), which can describe/estimate future distributions of
quality scores. Correspondingly, with these samples, we can obtain
mean and variance of quality scores between the predicted worker
$\widehat{w_i}$ and the current task $t_j$.

{\it \underline{Case 2: $\langle w_i, \widehat{t_j}\rangle$.}}
Similar to Case 1, we can obtain $m_j$ spatial tasks $t_j$ that can
be reached by worker $w_i$. Then, we obtain $m_j$ quality score
samples from valid pairs  $\langle w_i, t_j\rangle$ (each sample
with probability $1/m_j$), whose mean and variance can be used to
capture the quality score distribution between the current worker
$w_i$ and a predicted task $\widehat{t_j}$.

{\it \underline{Case 3: $\langle \widehat{w_i},
\widehat{t_j}\rangle$.}} Since both worker $\widehat{w_i}$ and task
$\widehat{t_j}$ have predicted distributions, we cannot directly
obtain quality score distributions. Thus, our basic idea is to infer
future quality scores by existing workers $w_i$ and tasks $t_j$ at
the current timestamp $p$. That is, at the current time instance, we
collect quality scores $q_{ij}$ of all pairs $\langle w_i,
t_j\rangle$ as samples, and use them to represent probabilistic
distributions of quality scores of both worker $\widehat{w_i}$ and
task $\widehat{t_j}$ at the future time instance.

\vspace{0.5ex} \noindent \textbf{Existence Probabilities of Pairs
with the Predicted Workers/Tasks.} Some assignment pairs
that involve the predicted worker/task may not be valid, due to the
time constraints of spatial tasks $t_j$ or $\widehat{t_j}$ (i.e.,
deadline $e_j$). Thus, we will associate each pair (with either
worker or task in future) with an existence probability,
$\widehat{p_{ij}}$.

For pair $\langle \widehat{w_i}, t_j\rangle$, we let
$\widehat{p_{ij}} = \min\{\frac{n_i}{|W_p|}, 1\}$, where $n_i$ is
the number of valid workers who can reach task $t_j$ at the current
timestamp $p$, and $|W_p|$ is the total number of (estimated)
workers at timestamp $p$.

Similarly, for pair $\langle w_i, \widehat{t_j}\rangle$, we can
obtain: $\widehat{p_{ij}} = \min\{\frac{m_j}{|T_p|}, 1\}$, where
$m_j$ is the number of valid tasks that worker $w_i$ can reach
before the deadlines, and $|T_p|$ is the total number of (estimated)
tasks at timestamp $p$.

For pair $\langle \widehat{w_i}, \widehat{t_j}\rangle$, let $u_{ij}$
be the total number of valid pairs between $W_p$ and $T_p$ at timestamp $p$. Then, we can estimate the existence
probability of pair $\langle \widehat{w_i}, \widehat{t_j}\rangle$
by: $\widehat{p_{ij}} = \frac{u_{ij}}{|W_p|\cdot |T_{p}|}$.

\section{The MQA Greedy Approach}
\label{sec:greedy}

In this section, we propose an efficient  MQA greedy algorithm (GREEDY) to solve
the MQA problem, which iteratively finds one ``best''
worker-and-task assignment pair, $\langle w_i, t_j\rangle$, each
time, with the highest increase of the quality score and under the
budget constraint of high confidences. Here, in order to achieve high total quality scores, GREEDY is applied over both
current and predicted future workers/tasks.

After all assignment pairs (involving current/future workers/tasks)
are selected, we will only insert into the set $I_p$ those pairs,
$\langle w_i, t_j\rangle$, with both workers and tasks at current
timestamp $p$ (i.e., $w_i\in W_p$ and $t_j\in T_p$).

\subsection{The Comparisons of the Quality Score Increases / Traveling Cost Increases}
\label{subsec:score_increase}

Since MQA greedy algorithm needs to select one worker-and-task
assignment pair, $\langle \tilde{w_i}, \tilde{t_j}\rangle$, at a
iteration with the highest increase of the total quality score, in this
section, we will first formalize the increase of the quality
score, $\Delta_q(\tilde{w_i}, \tilde{t_j})$, for a pair $\langle
\tilde{w_i}, \tilde{t_j}\rangle$, and then compare the increases of
overall quality scores between two pairs $\langle \tilde{w_i},
\tilde{t_j}\rangle$ and $\langle \tilde{w_a}, \tilde{t_b}\rangle$,
where $\tilde{w_i}$, $\tilde{t_j}$, $\tilde{w_a}$, and $\tilde{t_b}$
can be either current or predicted workers/tasks.

\vspace{0.5ex} \noindent {\bf The Calculation of the Quality Score
Increase, $\Delta_q(\tilde{w_i}, \tilde{t_j})$.} Based on
Eq.~(\ref{eq:tbc_sc}), the overall quality score is given by summing
up all quality scores of the selected assignment pairs. Thus, when
we choose a new assignment pair $\langle \tilde{w_i},
\tilde{t_j}\rangle$, the increase of the quality score,
$\Delta_q(\tilde{w_i}, \tilde{t_j})$, is exactly equal to the
quality score of this new pair $\langle \tilde{w_i},
\tilde{t_j}\rangle$, denoted as $\tilde{q_{ij}}$. That
is,\vspace{-3ex}

{\scriptsize
\begin{eqnarray}
\Delta_q(\tilde{w_i}, \tilde{t_j}) &=& \tilde{q_{ij}},
\label{eq:score_increase}
\end{eqnarray}\vspace{-6ex}
}

\noindent where $\tilde{q_{ij}}$ is a fixed value, if both
$\tilde{w_i}$ and $\tilde{t_j}$ are current worker and task,
respectively; otherwise, $\tilde{q_{ij}}$ is a random variable whose
distribution can be given by samples discussed in Section
\ref{subsec:statistics}.

\vspace{0.5ex} \noindent {\bf The Comparisons of the Quality Score
Increase Between Two Pairs $\langle \tilde{w_i}, \tilde{t_j}\rangle$
and $\langle \tilde{w_a}, \tilde{t_b}\rangle$.} Next, we discuss how
to decide which worker-and-task assignment pair is better, in terms
of the quality score increase, between two pairs $\langle
\tilde{w_i}, \tilde{t_j}\rangle$ and $\langle \tilde{w_a},
\tilde{t_b}\rangle$.

Specifically, if both pairs have workers/tasks at current timestamp
$p$, then the quality score increases, $\tilde{q_{ij}}$ and
$\tilde{q_{ab}}$ (given in Eq.~(\ref{eq:score_increase})), are fixed
values. In this case, the pair with higher quality score increase is
better.

On the other hand, in the case that either of the two pairs involves
the predicted workers/tasks, their corresponding quality score
increases, that is, $\tilde{q_{ij}}$ and/or $\tilde{q_{ab}}$, are
random variables. To compare the two quality score increases, we can
compute the probability, $Pr_{\Delta_q(\tilde{w_i}, \tilde{t_j})}$,
that pair $\langle \tilde{w_i}, \tilde{t_j}\rangle$ has the increase
greater than that of the other one. That is, by applying the
\textit{central limit theorem} (CLT)
\cite{grinstead2012introduction, jovanovic1997demo}, we
have:\vspace{-2ex}

{\scriptsize
\begin{eqnarray}
Pr_{\Delta_q(\tilde{w_i}, \tilde{t_j})} &=&
Pr\{\Delta_q(\tilde{w_i}, \tilde{t_j}) >
\Delta_q(\tilde{w_a}\tilde{t_b})\}\label{eq:quality_increase_comparison}\\
&=& Pr\{\tilde{q_{ij}}> \tilde{q_{ab}}\}\notag \\
&=& 1-Pr\{\tilde{q_{ij}}\leq \tilde{q_{ab}}\}\notag\\
&=& 1-Pr\left\{\frac{\tilde{q_{ij}} - \tilde{q_{ab}} -
(E(\tilde{q_{ij}}) - E(\tilde{q_{ab}}))}{Var(\tilde{q_{ij}}) +
Var(\tilde{q_{ab}})}\right.\notag\\
&& \left.\leq \frac{ - (E(\tilde{q_{ij}}) - E(\tilde{q_{ab}}))}{Var(\tilde{q_{ij}}) + Var(\tilde{q_{ab}})}\right\}\notag\\
&=& 1- \Phi\left(\frac{ - (E(\tilde{q_{ij}}) -
E(\tilde{q_{ab}}))}{Var(\tilde{q_{ij}}) +
Var(\tilde{q_{ab}})}\right),\notag
\end{eqnarray}\vspace{-2ex}
}

\noindent where $\Phi(\cdot)$ is the \textit{cumulative density
function} (cdf) of a standard normal distribution.

With Eq.~(\ref{eq:quality_increase_comparison}), we can compute
the probability, $Pr_{\Delta_q(\tilde{w_i}, \tilde{t_j})}$, that
pair $\langle \tilde{w_i}, \tilde{t_j}\rangle$ is better than (i.e.,
with higher score than) pair $\langle \tilde{w_a},
\tilde{t_b}\rangle$. If it holds that $Pr_{\Delta_q(\tilde{w_i},
\tilde{t_j})} > 0.5$, then we say that pair $\langle \tilde{w_i},
\tilde{t_j}\rangle$ is expected to have higher quality score
increase; otherwise, pair $\langle \tilde{w_a}, \tilde{t_b}\rangle$
has higher quality score increase.

\vspace{0.5ex} \noindent {\bf The Comparisons of the Traveling Cost
Increase Between Two Pairs $\langle \tilde{w_i}, \tilde{t_j}\rangle$
and $\langle \tilde{w_a}, \tilde{t_b}\rangle$.} Similar to the
quality score, we can also compute the probability,
$Pr_{\Delta_c(\tilde{w_i}, \tilde{t_j})}$, that the increase of the
traveling cost for pair $\langle \tilde{w_i}, \tilde{t_j}\rangle$ is
smaller than that of pair $\langle \tilde{w_a}, \tilde{t_b}\rangle$.
That is, we can obtain:\vspace{-3ex}

{\scriptsize
\begin{eqnarray}
Pr_{\Delta_c(\tilde{w_i}, \tilde{t_j})} &=&
Pr\{\Delta_c(\tilde{w_i}, \tilde{t_j}) \leq
\Delta_c(\tilde{w_a}, \tilde{t_b})\}\label{eq:cost_increase_comparison}\\
&=& Pr\{\tilde{c_{ij}}\leq \tilde{c_{ab}}\}\notag \\
&=& \Phi\left(\frac{ - (E(\tilde{c_{ij}}) -
E(\tilde{c_{ab}}))}{Var(\tilde{c_{ij}}) +
Var(\tilde{c_{ab}})}\right).\notag
\end{eqnarray}\vspace{-8ex}
}

\subsection{The Pruning Strategy}

As discussed in Section \ref{subsec:score_increase}, one
straightforward method for selecting a ``good'' assignment pair at a
iteration is as follows. We sequentially scan each valid worker-and-task
pair $\langle \tilde{w_i}, \tilde{t_j}\rangle$, and compare its
quality score increase, $\Delta_q(\tilde{w_i}, \tilde{t_j})$, with
that of the best-so-far pair $\langle \tilde{w_a},
\tilde{t_b}\rangle$, in terms of the probability
$Pr_{\Delta_q(\tilde{w_i}, \tilde{t_j})}$. If the pair $\langle
\tilde{w_i}, \tilde{t_j}\rangle$ expects to have higher quality
score (and moreover satisfy the budget constraint), then we consider
it as the new best-so-far pair.

The straightforward method mentioned above considers all possible
valid assignment pairs, and computes their comparison probabilities,
which requires high time complexity, that is, $O(m'\cdot n')$, where
$m'$ and $n'$ are the numbers of tasks and workers at both current
and future time instances, respectively. Therefore, in this section, we
will propose effective pruning methods to quickly discard those
false alarms of assignment pairs, with both high traveling costs and
low quality scores.

\vspace{0.5ex} \noindent {\bf Pruning with Bounds of Quality and
Traveling Cost.} Without loss of generality, for each pair $\langle
\tilde{w_i}, \tilde{t_j}\rangle$, assume that we can obtain its
lower and upper bounds of the traveling cost, $\tilde{c_{ij}}$, and
quality score, $\tilde{q_{ij}}$, where $\tilde{c_{ij}}$ and
$\tilde{q_{ij}}$ are either fixed values (if $\tilde{w_i}$ and
$\tilde{t_j}$ are worker/task at the current time instance) or random
variables (if worker and/or task are from the future time instance). That is, we
denote $\tilde{c_{ij}} \in [lb\_\tilde{c_{ij}}, ub\_\tilde{c_{ij}}]$
and $\tilde{q_{ij}} \in [lb\_\tilde{q_{ij}}, ub\_\tilde{q_{ij}}]$.

This way, in a 2D quality-and-travel-cost space, each
worker-and-task assignment pair, $\langle \tilde{w_i},
\tilde{t_j}\rangle$, corresponds to a rectangle,
$[lb\_\tilde{q_{ij}},$ $ub\_\tilde{q_{ij}}]$ $\times
[lb\_\tilde{c_{ij}}, ub\_\tilde{c_{ij}}]$. Then, based on the idea
of the \textit{skyline} query \cite{borzsony2001skyline}, we can safely
prune those pairs that are \textit{dominated} by candidate pairs, in
terms of the traveling cost and quality score.

\begin{lemma} (The Dominance Pruning) Given a candidate pair $\langle
\tilde{w_a}, \tilde{t_b}\rangle$, a valid worker-and-task pair
$\langle \tilde{w_i}, \tilde{t_j}\rangle$ can be safely pruned, if
and only if it holds that: (1) $ub\_\tilde{c_{ab}} <
lb\_\tilde{c_{ij}}$, and (2) $lb\_\tilde{q_{ab}} >
ub\_\tilde{q_{ij}}$.\label{lemma:dominate_pair}
\end{lemma}
\begin{proof} Since it holds that $\tilde{c_{ij}} \in [lb\_\tilde{c_{ij}}, ub\_\tilde{c_{ij}}]$ and
$\tilde{q_{ij}} \in [lb\_\tilde{q_{ij}}, ub\_\tilde{q_{ij}}]$, by
lemma assumptions and inequality transition, we have:\vspace{-2ex}

{\scriptsize
$$\tilde{c_{ab}}\leq ub\_\tilde{c_{ab}} < lb\_\tilde{c_{ij}}\leq
\tilde{c_{ij}}, \text{ and}$$ $$\tilde{q_{ab}}\geq
lb\_\tilde{q_{ab}}
> ub\_\tilde{q_{ij}}\geq \tilde{q_{ij}}.$$\vspace{-5ex}
}

As a result, we can see that, compared to pair $\langle \tilde{w_a},
\tilde{t_b}\rangle$, pair $\langle \tilde{w_i}, \tilde{t_j}\rangle$
has both higher traveling cost $\tilde{c_{ij}}$ and lower quality
score $\tilde{q_{ij}}$. Since our GREEDY algorithm only selects one
best pair each iteration (which can maximally increase the quality score
and minimally increase the traveling cost), pair $\langle
\tilde{w_i}, \tilde{t_j}\rangle$ is definitely worse than $\langle
\tilde{w_a}, \tilde{t_b}\rangle$ in both quality and traveling cost
dimensions, and thus can be safely pruned.
\end{proof}

\vspace{0.5ex} \noindent {\bf Pruning with the Increase
Probability.} Lemma \ref{lemma:dominate_pair} utilizes the
lower/upper bounds of the quality score and traveling cost to enable
the dominance pruning. If a pair cannot be simply pruned by Lemma
\ref{lemma:dominate_pair}, we will further consider a more costly
pruning method, by consider the probabilistic information.

\begin{lemma} (The Increase Probability Pruning) Given a candidate
pair $\langle \tilde{w_a}, \tilde{t_b}\rangle$, a valid pair
$\langle \tilde{w_i}, \tilde{t_j}\rangle$ can be safely pruned, if
and only if it holds that: (1) $Pr_{\Delta_q(\tilde{w_i},
\tilde{t_j})}$ (w.r.t. $\langle \tilde{w_a}, \tilde{t_b}\rangle$) is
greater than 0.5, and (2) $Pr_{\Delta_c(\tilde{w_i},
    \tilde{t_j})}$ (w.r.t. candidate pair $\langle \tilde{w_a},
\tilde{t_b}\rangle$) is greater than 0.5 , where
$Pr_{\Delta_q(\tilde{w_i}, \tilde{t_j})}$ and
$Pr_{\Delta_c(\tilde{w_i}, \tilde{t_j})}$
 are given by Eqs.~(\ref{eq:quality_increase_comparison}) and (\ref{eq:cost_increase_comparison}), respectively.\label{lemma:probabilistic_dominate}
\end{lemma}

Intuitively, Lemma \ref{lemma:probabilistic_dominate} filters out
those pairs $\langle \tilde{w_i}, \tilde{t_j}\rangle$ that have
higher probabilities to be inferior to other candidate pairs
$\langle \tilde{w_a}, \tilde{t_b}\rangle$, in terms of both
traveling cost and quality score.

Based on Lemmas \ref{lemma:dominate_pair} and
\ref{lemma:probabilistic_dominate}, we can obtain a set, $S_p$, of
candidate pairs that cannot be dominated by other pairs.

\vspace{0.5ex} \noindent {\bf Selection of the Best Pair Among
Candidate Pairs.} Given a number of candidate pairs in set $S_p$,
GREEDY still needs to identify one ``best'' pair with a
high quality score and under the budget constraint. Specifically, we
will first filter out those false alarms in $S_p$ with high
traveling costs (i.e., violating the budget constraint), and then
return one pair with the highest probability to have larger quality
score than others in $S_p$.

Assume that in GREEDY, we have so far selected $L$
pairs, denoted as $\langle \tilde{w_a}, \tilde{t_b}\rangle$. Then,
with a new assignment pair $\langle \tilde{w_i}, \tilde{t_j}\rangle
\in S_p$, if it holds that:\vspace{-2ex}

{\scriptsize
\begin{eqnarray}
Pr\left\{\left(\sum_{\forall \langle \tilde{w_a},
\tilde{t_b}\rangle} \tilde{lb\_c_{ab}}\right) +\tilde{c_{ij}}\leq
B_{max}\right\}\leq \delta,\label{eq:prune_high_travel_cost}
\end{eqnarray}\vspace{-2ex}
}

\noindent then pair $\langle \tilde{w_i}, \tilde{t_j}\rangle \in
S_p$ can be safely ruled out from candidate set $S_p$, where
$\delta$ is a user-specified confidence level that the selected
assignment satisfies the budget constraint $B_{max}$ for both
(remaining) current- and next- time instance budgets.
Eq.~(\ref{eq:prune_high_travel_cost}) can be computed via CLT
\cite{grinstead2012introduction, jovanovic1997demo}.

Next, in the remaining candidate pairs in $S_p$, we will select one
pair $\langle \tilde{w_i}, \tilde{t_j}\rangle$ with the highest
probability, $Pr_{q, max} (\langle \tilde{w_i},
\tilde{t_j}\rangle)$, of having the largest high quality. That is,
we have:\vspace{-3ex}

{\scriptsize
\begin{eqnarray}
Pr_{q, max} (\langle \tilde{w_i}, \tilde{t_j}\rangle) =
\prod_{\forall \langle \tilde{w_a}, \tilde{t_b}\rangle}
Pr_{\Delta_q(\tilde{w_i}, \tilde{t_j})} (\langle \tilde{w_a},
\tilde{t_b}\rangle) \label{eq:max_quality_prob}
\end{eqnarray}\vspace{-3ex}
}

\noindent where $Pr_{\Delta_q(\tilde{w_i}, \tilde{t_j})} (\langle
\tilde{w_a}, \tilde{t_b}\rangle)$ is the probability of quality
score increase, compared with pair $\langle \tilde{w_a},
\tilde{t_b}\rangle$, given by
Eq.~(\ref{eq:quality_increase_comparison}).

Finally, among all the remaining candidate pairs in set $S_p$, we
will choose the one, $\langle \tilde{w_i}, \tilde{t_j}\rangle$, with
the highest probability $Pr_{q, max} (\langle \tilde{w_i},$
$\tilde{t_j}\rangle)$, which will be included as a selected best
assignment pair in GREEDY.

\begin{figure}[t!]
    \begin{center}\vspace{-4ex}
        \begin{tabular}{l}
            \parbox{3.1in}{
                \begin{scriptsize}
                    \begin{tabbing}
                        12\=12\=12\=12\=12\=12\=12\=12\=12\=12\=12\=\kill
                        {\bf Procedure {\sf MQA\_Greedy}} \{ \\
                        \> {\bf Input:} current and predicted workers $\tilde{w_i}$  in $W$,  current and predicted tasks $\tilde{t_j}$\\
                        \> \> \> \>  in $T$, and the maximum possible budget $B_{max}$\\
                        \> {\bf Output:} a worker-and-task assignment instance set, $I_p$\\
                        \> (1) \> \> $I_p=\emptyset$; \\
                        \> (2) \> \> obtain a list, $\mathcal{L}$, of valid worker-and-task pairs for $\tilde{w_i}\in W$ and $\tilde{t_j}\in T$\\
                        \> (3) \> \> \textbf{for} $k=1$ to $\min\{|W|, |T|\}$ \\
                        \> (4) \> \> \> $S_p = \emptyset$;\\
                        \> (5) \> \> \> \textbf{for} each valid assignment pair $\langle \tilde{w_i}, \tilde{t_j}\rangle \in \mathcal{L}$\\
                        \> (6) \> \> \> \> \textbf{if} $\langle \tilde{w_i}, \tilde{t_j}\rangle$ has $lb\_\tilde{c_{ij}}$ greater than the remaining budget, then \textbf{continue};\\
                        \> (7) \> \> \> \> \textbf{if} pair $\langle \tilde{w_i}, \tilde{t_j}\rangle$ cannot be pruned w.r.t. $S_p$ by Lemma \ref{lemma:dominate_pair}\\
                        \> (8) \> \> \> \> \> \textbf{if} pair $\langle \tilde{w_i}, \tilde{t_j}\rangle$ cannot be pruned w.r.t. $S_p$ by Lemma \ref{lemma:probabilistic_dominate}\\
                        \> (9) \> \> \> \> \> \> add $\langle \tilde{w_i}, \tilde{t_j}\rangle$ to $S_p$\\
                        \> (10)\> \> \> \> \> \> prune other candidate pairs in $S_p$ with $\langle \tilde{w_i}, \tilde{t_j}\rangle$\\
                        \> (11)\> \> \> select one best assignment pair $\langle \tilde{w_i}, \tilde{t_j}\rangle$ in $S_p$ satisfying the budget  \\
                        \>     \> \> \> \> constraint $B_{max}$ in Eq.~(\ref{eq:prune_high_travel_cost}) and with the highest probability  \\
                        \>     \> \> \> \> $Pr_{q, max} (\langle \tilde{w_i}, \tilde{t_j}\rangle)$ in Eq.~(\ref{eq:max_quality_prob})\\
                        \> (12)\> \> \> add the selected pair $\langle \tilde{w_i}, \tilde{t_j}\rangle$ to $I_p$ \\
                        \> (13)\> \> \> remove all those valid pairs $\langle \tilde{w_i}, - \rangle$ or $\langle -, \tilde{t_j}\rangle$ from $\mathcal{L}$\\
                        \> (14)\> \> remove those worker-and-task pairs with the predicted workers/tasks from $I_p$\\
                        \> (15)\> \> \textbf{return} $I_p$\}\\
                    \end{tabbing}
                \end{scriptsize}
            }
        \end{tabular}
    \end{center}\vspace{-5ex}
    \caption{\small The MQA Greedy Algorithm.}
    \vspace{-4ex}
    \label{alg:greedy}
\end{figure}

\subsection{The MQA Greedy Algorithm}

In this subsection, we propose \textit{MQA greedy} algorithm, which iteratively
assigns a worker to a spatial task greedily that can obtain a high the overall
quality score under the budget constraint each iteration.

Figure \ref{alg:greedy} presents the pseudo code of our \textit{MQA greedy}
algorithm, namely {\sf MQA\_Greedy}, which obtains one best
worker-and-task assignment pair each time over both current and
predicted future workers and tasks, where the selected pair
satisfies the budget constraint $B_{max}$ and has the largest
quality score with high confidence, where $B_{max}$ is the available
budget in both current and next time instances.

We first initialize the worker-and-task assignment instance set
$I_p$ with an empty set (line 1). Then, we obtain a list,
$\mathcal{L}$, of valid worker-and-task assignment pairs, which may
involve either current or future workers/tasks, that is,
$\tilde{w_i}\in W$ and $\tilde{t_j}\in T$ (line 2). Next, for each
iteration, we find one best assignment pair with high quality score
and low traveling cost (satisfying the budget constraint) (lines
3-13). In particular, we check each valid assignment pair $\langle
\tilde{w_i}, \tilde{t_j}\rangle$ in the list $\mathcal{L}$ (line 5).
If this pair has the lower bound, $lb\_\tilde{c_{ij}}$, of the
traveling cost greater than (the upper bound of) the remaining
budget (w.r.t. $I_p$ and $B_{max}$), then it does not satisfy the
budget constraint, and we can continue to check the next assignment
pair (line 6).
Then, if the pair $\langle \tilde{w_i}, \tilde{t_j}\rangle$ cannot
be pruned by dominance and increase probability pruning methods in
Lemmas \ref{lemma:dominate_pair} and
\ref{lemma:probabilistic_dominate}, respectively, then $\langle
\tilde{w_i}, \tilde{t_j}\rangle$  is a candidate pair, and we
include in an initially empty candidate set $S_p$  (lines 7-9). In
addition, we can also use candidate pair $\langle \tilde{w_i},
\tilde{t_j}\rangle$ to prune other pairs in set $S_p$ (line 10).
After that, we can insert the best pair from the candidate set $S_p$
into set $I_p$ (lines 11-12), such that the budget constraint
$B_{max}$ is satisfied in Eq.~(\ref{eq:prune_high_travel_cost}) and
the probability $Pr_{q, max} (\langle \tilde{w_i},
\tilde{t_j}\rangle)$ in Eq.~(\ref{eq:max_quality_prob}) is
maximized. Since each worker can be assigned with at most one task
and each task is accomplished by at most one worker, we remove those
valid pairs from $\mathcal{L}$ that contains either worker
$\tilde{w_i}$ or task $\tilde{t_j}$ (line 13). Finally, we remove
those worker-and-task pairs involving future workers/tasks from
$I_p$, and return the set $I_p$ as the solution of the \textit{MQA greedy}
algorithm (lines 14-15).

\section{The MQA Divide-and-Conquer Approach}
\label{sec:D&C}

In this section, we propose an efficient \textit{MQA divide-and-conquer
algorithm} (D\&C), which partitions the MQA problem into $g$
subproblems, recursively conquers the subproblems, and merges
assignment results from subproblems. In this paper, we will divide
the MQA problem with $m'$ current/future tasks into $g$
subproblems, each involving $\lceil m'/g\rceil$ tasks. The D\&C
process continues, until the subproblem sizes become 1 (i.e., with
one single spatial task in subproblems, which can be easily solved
by GREEDY). We will later discuss how to utilize a
cost model to estimate the best $g$ value that can achieve low MQA
processing cost.

\subsection{The Decomposition of the MQA Problem}
\label{subsec:PA_SC_decomposition}

\vspace{0.5ex}\noindent {\bf Decomposing the MQA Problem.}
Specifically, assume that the original MQA problem involves $m'$
current/future spatial tasks for both current and next time instances. Our
goal is to divide this problem into $g$ subproblems $M_s$ (for
$1\leq s\leq g$), such that each subproblem $M_s$ involves a
disjoint subgroup of $\lceil m'/g\rceil$ spatial tasks,
$\tilde{t_j}$, each of which is associated with potentially valid
worker(s) $\tilde{w_i}$ (i.e., with valid worker-and-task assignment
pairs $\langle \tilde{w_i}, \tilde{t_j}\rangle$).

After the decomposition, each subproblem $M_s$ contains
all valid pairs, $\langle \tilde{w_i}, \tilde{t_j}\rangle$, w.r.t.
the decomposed $\lceil m'/g\rceil$ tasks. Since tasks in different
subproblems may be reachable by the same workers, different
subproblems may involve the same (conflicting) workers, whose
conflictions should be resolved when we merge solutions (the
selected assignment pairs) to these subproblems (as
discussed in Section \ref{subsec:merge}).

\vspace{1ex}\noindent \textbf{Example 4 (The MQA Problem
Decomposition)} {\it Figure \ref{fig:decomposing} shows an example
of decomposing the MQA problem (as shown in Figure
\ref{subfig:beforedecomposing}) into 3 subproblems (as depicted in
Figure \ref{subfig:decomposed}), where each subproblem contains one
single spatial task (i.e., subproblem size = 1), associated with its
related valid workers. Here, the dashed border indicates the
predicted future workers (i.e., $w_4$ and $w_5$) or task (i.e.,
$t_3$). In this example, the first subproblem in Figure
\ref{subfig:decomposed} contains task $t_1$, which can be reached by
workers $w_1$ and $w_4$. Different tasks may have conflicting
workers, for example, tasks $t_1$ and $t_2$ from subproblems $M_1$
and $M_2$, respectively, share the same (conflicting) worker $w_1$.

}

\begin{figure}[t!]\vspace{-4ex}
    \centering
    \setcounter{subfigure}{-1}
    \subfigure{
        \scalebox{.085}[.085]{\includegraphics{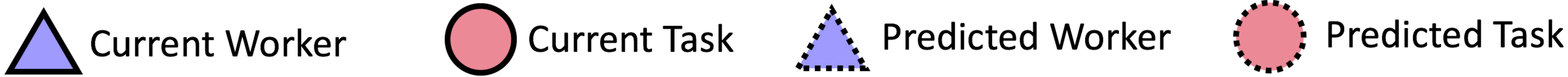}}
        \label{subfig:composing_bar}}\vspace{-1ex}
    \subfigure[][{\scriptsize The MQA Problem}]{
        \scalebox{.085}[.085]{\includegraphics{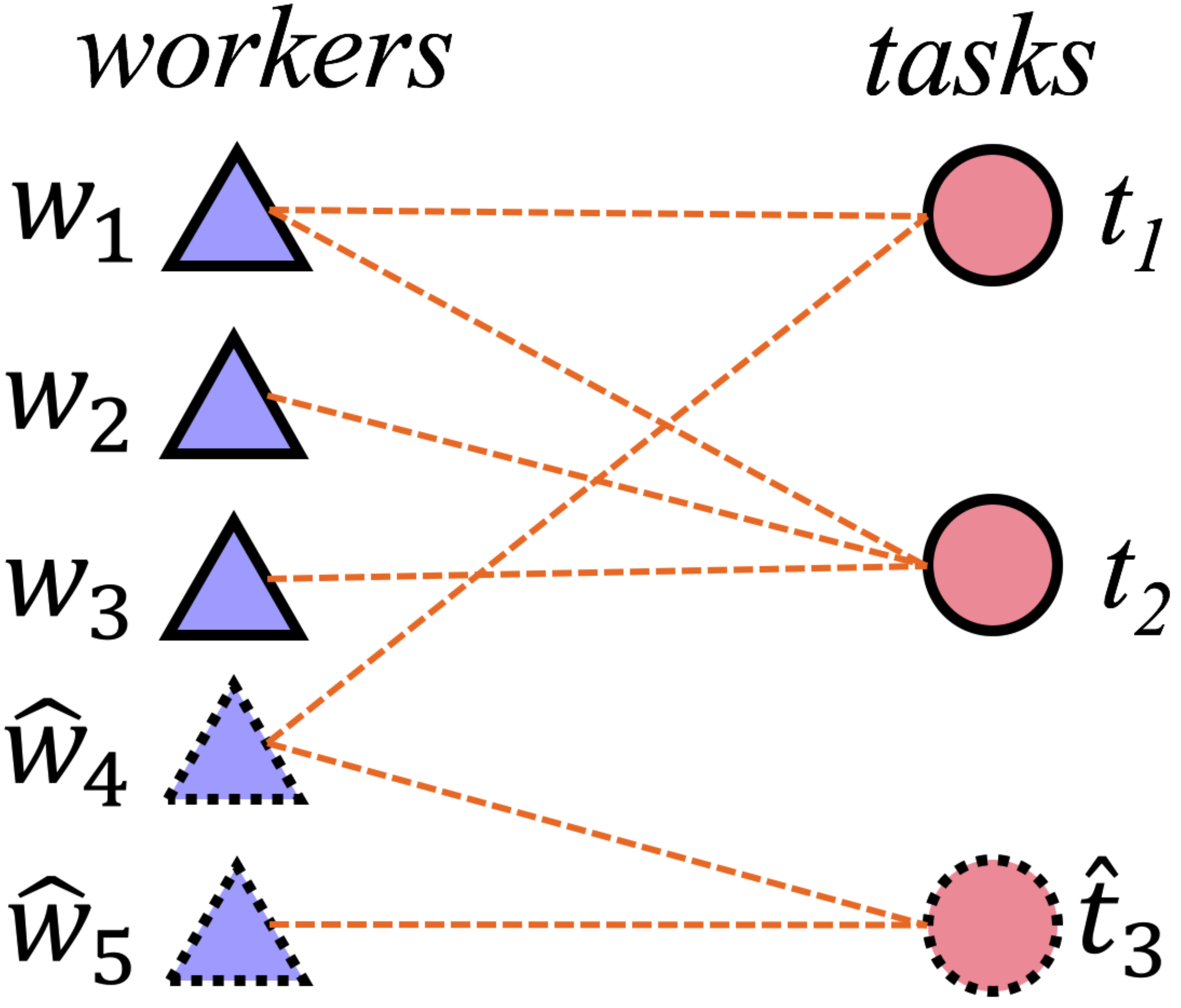}}
        \label{subfig:beforedecomposing}}\hspace{5ex}
    \subfigure[][{\scriptsize The Decomposed Subproblems}]{
        \scalebox{0.08}[0.08]{\includegraphics{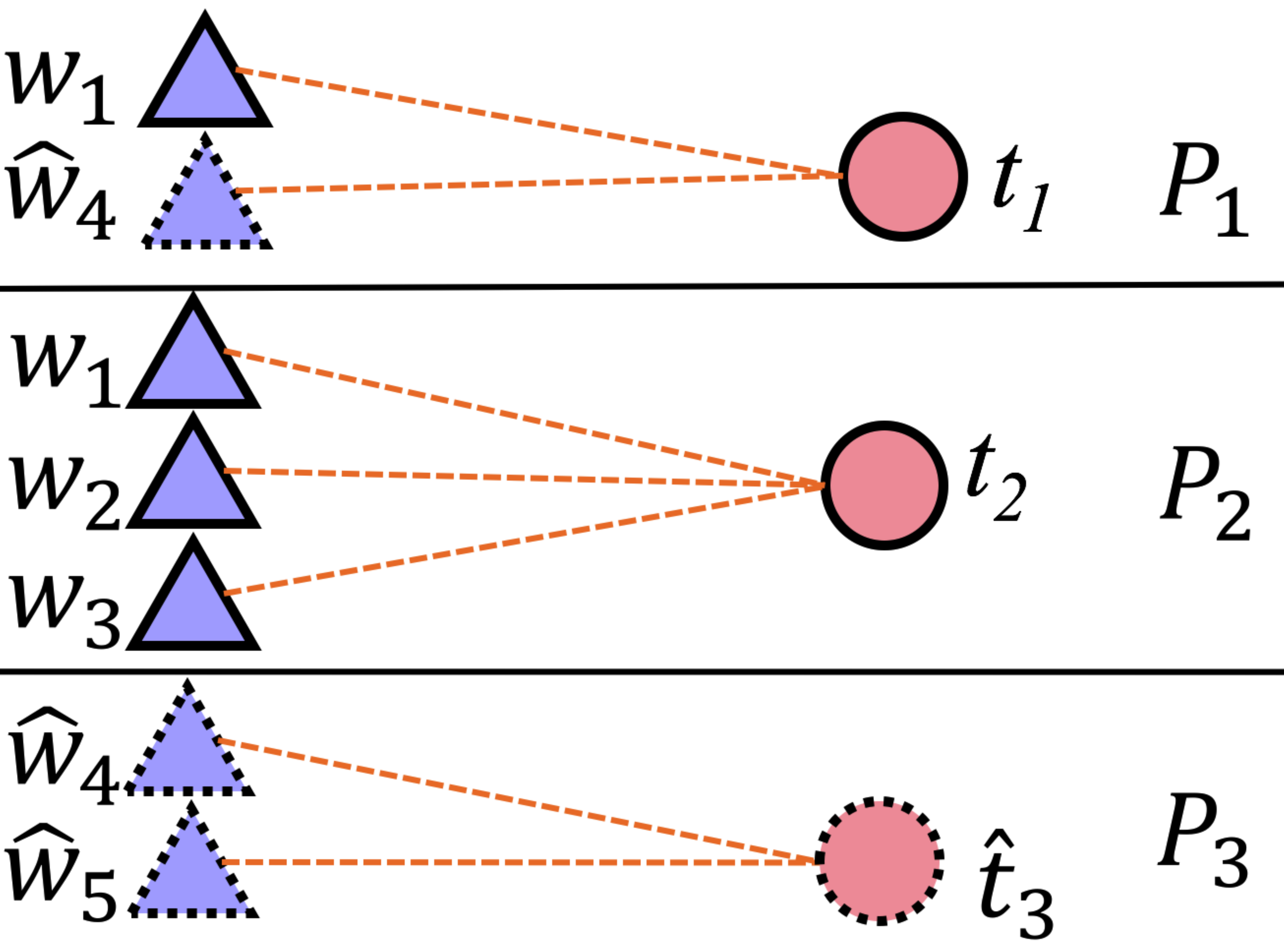}}
        \label{subfig:decomposed}}\vspace{-2ex}
    \caption{\small Illustration of Decomposing the MQA Problem.}\vspace{-4ex}
    \label{fig:decomposing}
\end{figure}

\vspace{0.5ex}\noindent {\bf The MQA Decomposition Algorithm.}
Figure \ref{alg:decomposing} illustrates the pseudo code of our
MQA decomposition algorithm, namely {\sf MQA\_Decomposition},
which decomposes the MQA problem (with $m'$ tasks), and returns
$g$ MQA subproblems, $M_s$ (each having $\lceil m'/g \rceil$
tasks).

Specifically, we first initialize $g$ empty subproblems, $M_s$,
where $1\leq s\leq g$ (lines 1-2). Then, we find all valid
worker-and-task assignment pairs $\langle \tilde{w_i},
\tilde{t_j}\rangle$ for both current and predicted workers/tasks in
sets $W$ and $T$, respectively (line 3).

Next, we want to iteratively retrieve $g$ subproblems, $M_s$, from
the original MQA problem (lines 4-8). That is, for the $s-th$
iteration, we first obtain an anchor task $\tilde{t_j}$ and its
$(\lceil m/g \rceil -1)$ nearest tasks, and add them to set
$T_p^{(s)}$ (line 5), where anchor tasks $\tilde{t_j}$ are chosen in
a \textit{sweeping style} (starting with the smallest longitude, or
mean of the longitude for future tasks; in the case that multiple
tasks have the same longitude, we choose the one with smallest
latitude).

For each task $\tilde{t_j}\in T_p^{(s)}$, we obtain all its related
workers $\tilde{w_i}$ who can reach task $\tilde{t_j}$, and add
pairs $\langle \tilde{w_i}, \tilde{t_j} \rangle$ to subproblem $M_s$
(lines 6-8). Finally, we return the $g$ decomposed subproblems $M_s$
(for $1\leq s\leq g$) (line 9).

\subsection{The MQA Merge Algorithm}
\label{subsec:merge}

As mentioned earlier, we can execute the decomposition algorithm to
recursive divide the MQA problem, until each subproblem only
involves one single task (which can be easily processed by the
greedy algorithm). After we obtain solutions to the decomposed MQA
subproblems (i.e., a number of selected assignment pairs in
subproblems), we need to merge these solutions into the one to the
original MQA problem.

\vspace{0.5ex}\noindent {\bf Resolving Worker-and-Task Assignment
Conflicts.} During the merge process, some workers are conflicting,
that is, they are assigned to different tasks in the solutions to
distinct subproblems at the same time. This contradicts with the
requirement that each worker can only be assigned with at most one
spatial task at each time instance. Thus, to merge solutions of these
conflicting subproblems, we resolve the conflicts.

Next, we use an example to illustrate how to resolve the conflicts
between two (or multiple) pairs $\langle \tilde{w_i}, \tilde{t_j}
\rangle$ and $\langle \tilde{w_i}, \tilde{t_b} \rangle$ (w.r.t. the
conflicting worker $\tilde{w_i}$), by selecting one ``best'' pair
with low traveling cost and high quality score.

\vspace{1ex}\noindent \textbf{Example 5 (The Merge of Subproblems)}
{\it In the example of Figure \ref{subfig:decomposed}, assume that
in subproblems $M_1$ and $M_2$, we selected pairs $\langle w_1, t_1
\rangle$ and $\langle w_1, t_2 \rangle$ as the best assignment,
respectively, which contains the conflicting worker $w_1$. When we
merge the two subproblems $M_1$ and $M_2$, we need to resolve such a
conflict by deciding which task should be assigned to the
conflicting worker $w_1$. By using Lemmas \ref{lemma:dominate_pair}
and \ref{lemma:probabilistic_dominate}, we can first prune pairs
that are not dominated by others. Then, among the remaining
candidate pairs, we can select one best pair satisfying
Eq.~(\ref{eq:prune_high_travel_cost}) and maximizing
Eq.~(\ref{eq:max_quality_prob}). In this example, assume that pair
$\langle w_1, t_1\rangle$ dominates pair $\langle w_1, t_2 \rangle$
by Lemma \ref{lemma:dominate_pair}. Then, we will assign worker
$w_1$ with task $t_1$ (since $\langle w_1, t_1\rangle$ is the best
pair), and find another worker (e.g., $w_2$) with the largest
quality score under the budget constraint to do task $t_2$. This
way, after solving conflicts, we obtain two updated pairs $\langle
w_1, t_1\rangle$ and $\langle w_2, t_2\rangle$ for subproblems $M_1$
and $M_2$, respectively.

}

\begin{figure}[t!]
	\begin{center}\vspace{-4ex}
		\begin{tabular}{l}
			\parbox{3.1in}{
				\begin{scriptsize} 
					\begin{tabbing}
						12\=12\=12\=12\=12\=12\=12\=12\=12\=12\=12\=\kill
						{\bf Procedure {\sf MQA\_Decomposition}} \{ \\
						\> {\bf Input:} $n'$ current/future workers $\tilde{w_i}$ in $W$, $m'$ current/future spatial tasks $\tilde{t_j}$\\
						\> \> \> \>   in $T$, and the number of subproblems $g$\\
						\> {\bf Output:} the decomposed MQA subproblems, $M_s$ (for $1\leq s\leq g$)\\
						\> (1) \> \> \textbf{for} $s$ = 1 to $g$\\
						\> (2) \> \> \> $M_s = \emptyset$\\
						\> (3) \> \> compute all valid worker-and-task pairs $\langle \tilde{w_i}, \tilde{t_j}\rangle$ from $W$ and $T$\\
						\> (4) \> \> \textbf{for} $s = 1$ to $g$\\
						\> (5) \> \> \> add an anchor task $\tilde{t_j}$ and find its $(\lceil m'/g \rceil -1)$ nearest tasks to set $T_p^{(s)}$ \\
						\>  \> \>  \> \> // the task, $\tilde{t_j}$, whose longitude (or mean of the longitude) is the smallest\\
						\> (6) \> \> \> \textbf{for} each current/future task $\tilde{t_j} \in T_p^{(s)}$\\
						\> (7) \> \> \> \> obtain all valid workers $\tilde{w_i}$ that can reach task $\tilde{t_j}$\\
						\> (8) \> \> \> \> add these pairs $\langle \tilde{w_i}, \tilde{t_j} \rangle$ to subproblem $M_s$\\
						\> (9) \> \>  \textbf{return} subproblems $M_1$, $M_2$, ..., and $M_g$\}\\
					\end{tabbing}
				\end{scriptsize}
			}
		\end{tabular}
	\end{center}
	\vspace{-5ex}
	\caption{\small The MQA Problem Decomposition Algorithm.}
	\vspace{-4ex}
	\label{alg:decomposing}
\end{figure}

\vspace{0.5ex}\noindent {\bf The MQA Merge Algorithm.} Figure
\ref{alg:conflict_reconcile} illustrates the MQA merge algorithm,
namely {\sf MQA\_Merge}, which resolves the conflicts between the
current assignment instance set $I_p$ (that we have merged
subproblems $M$) and that, $I_p^{(s)}$, of subproblem $M_s$, and
returns a merged set without conflicts.

First, we obtain a set, $W_c$, of conflicting workers between $I_p$
and $I_p^{(s)}$ (line 1), which are assigned with different tasks in
different subproblems, $M$ and $M_s$. Then, in each iteration, we
select one conflicting worker $\tilde{w_i}\in W_c$ with the highest
traveling cost in $I_p^{(s)}$, and choose one best pair between
$\langle \tilde{w_i}, \tilde{t_j} \rangle \in I_p$ and $\langle
\tilde{w_i}, \tilde{t_k}\rangle \in I_p^{(s)}$, in terms of budget
and quality score (which can be achieved by checking Lemmas
\ref{lemma:dominate_pair} and \ref{lemma:probabilistic_dominate},
and finding the one that satisfies
Eq.~(\ref{eq:prune_high_travel_cost}) and maximizes
Eq.~(\ref{eq:max_quality_prob}) (lines 2-4).
When $\langle \tilde{w_i}, \tilde{t_k}\rangle$ in subproblem $M_s$
is selected as the best pair, we can resolve the conflicts by
replacing $\langle \tilde{w_i}, \tilde{t_j} \rangle$ with $\langle
\tilde{w'_i}, \tilde{t_j} \rangle$ in $I_p$; otherwise, we can
replace $\langle \tilde{w_i}, \tilde{t_k}\rangle$ with $\langle
\tilde{w''_i}, \tilde{t_k} \rangle$ in $I_p^{(s)}$ (lines 5-8).
Then, we remove worker $\tilde{w_i}$ from set $W_c$ (line 9).

After resolving all conflicting workers in $W_c$ between $I_p$ and
$I_p^{(s)}$, we can merge them together, and return an updated
merged set $I_p$ (lines 10-11).

\begin{figure}[t!]
    \begin{center}
        \begin{tabular}{l}
            \parbox{3.1in}{
                \begin{scriptsize} \vspace{-4ex}
                    \begin{tabbing}
                        12\=12\=12\=12\=12\=12\=12\=12\=12\=12\=12\=\kill
                        {\bf Procedure {\sf MQA\_Merge}} \{ \\
                        \> {\bf Input:} the current assignment instance set, $I_p$, of the merged subproblems $M$, \\
                        \> \> \> \> and the assignment instance set, $I_p^{(s)}$, of subproblem $M_s$\\
                        \> {\bf Output:} a merged worker-and-task assignment instance set, $I_p$\\
                        \> (1) \> \> let $W_c$ be a set of conflicting workers between $I_p$ and $I_p^{(s)}$\\
                        \> (2) \> \> \textbf{while} $W_c \neq \emptyset$\\
                        \> (3) \> \> \> choose a worker $\tilde{w_i}\in W_c$ with the highest traveling cost in $I_p^{(s)}$\\
                        \> \> \> \> \textit{// assume $\tilde{w_i}$ is assigned to $\tilde{t_j}$ in $I_p$ and to $\tilde{t_k}$ in $I_p^{(s)}$}\\
                        \> (4) \> \> \> select one best pair between $\langle \tilde{w_i}, \tilde{t_j} \rangle \in I_p$ and $\langle \tilde{w_i}, \tilde{t_k}\rangle \in
                        I_p^{(s)}$\\
                        \> \> \> \> \> \> \> \textit{// by using Lemmas \ref{lemma:dominate_pair} and \ref{lemma:probabilistic_dominate}, and finding the one satisfying} \\
                        \> \> \> \> \> \> \> \textit{// Eq.~(\ref{eq:prune_high_travel_cost}) and
                        maximizing Eq.~(\ref{eq:max_quality_prob})}\\
                        \> (5) \> \> \> \textbf{if} pair $\langle \tilde{w_i}, \tilde{t_k}\rangle$ in subproblem $M_s$ is selected\\
                        \> (6) \> \> \> \> find another best worker $\tilde{w'_i}$ in $M$ and substitute $\langle \tilde{w'_i}, \tilde{t_j} \rangle$ in $I_p$\\
                        \> (7) \> \> \> \textbf{else} \\
                        \> (8) \> \> \> \> find another best worker $\tilde{w''_i}$ in $M_s$ and substitute $\langle \tilde{w''_i}, \tilde{t_k} \rangle$ in $I_p^{(s)}$\\
                        \> (9) \> \> \> $W_c = W_c - \{\tilde{w_i}\}$\\
                        \> (10)\> \> $I_p = I_p \cup I_p^{(s)}$\\
                        \> (11) \> \>  \textbf{return} $I_p$\}
                    \end{tabbing}
                \end{scriptsize}
            }
        \end{tabular}
    \end{center}\vspace{-3ex}
    \caption{\small The Merge Algorithm.}
    \vspace{-3ex}
    \label{alg:conflict_reconcile}
\end{figure}

\subsection{The D\&C Algorithm}
\label{subsec:D&C}

Up to now, we have discussed how to decompose and merge subproblems.
In this section, we will illustrate the detailed
\textit{MQA divide-and-conquer} (D\&C) algorithm, which partitions the
original MQA problem into subproblems, recursively solves each
subproblem, and merges assignment results of subproblems by
resolving conflicts and adjusting the assignments under the budget
constraint.

Figure \ref{alg:dc} shows the pseudo code of our D\&C algorithm,
namely procedure {\sf MQA\_D\&C}. We first initialize an empty set
$rlt$, which is used for store candidate pairs chosen by our D\&C
algorithm (line 1). Then, we utilize a novel cost model (discussed
in Appendix C to estimate the best number of the decomposed
subproblems, $g$, with respect to current/future sets, $W$ and $T$,
of workers and tasks, respectively (line 2). With this parameter
$g$, we can invoke the {\sf MQA\_Decomposition} algorithm (as
mentioned in Figure \ref{alg:decomposing}), and obtain $g$
subproblems $M_s$ (line 3).

For each subproblem $M_s$, if $M_s$ involves more than 1 task, then
we can recursively call procedure {\sf MQA\_D\&C} $(\cdot)$ to
obtain the best assignment pairs from subproblem $M_s$ (lines 4-6).
Otherwise, if subproblem $M_s$  only contains one single spatial
task $\tilde{t_j}$, then we apply the greedy algorithm (in Figure
\ref{alg:greedy}) to select one ``best'' worker for task
$\tilde{t_j}$ (lines 7-8). Here, the best worker means, the
corresponding pair has the highest quality score under the budget
constraint $B_{max}$.

After that, the selected assignment pairs in the $s$-th subproblem
are kept in set $rlt^{(s)}$, where $1\leq s\leq g$. Then, we can
invoke procedure {\sf MQA\_Merge} $(\cdot)$ to merge these $g$
sets $rlt^{(s)}$ into a set $rlt$, by resolving the conflicts (lines
9-11).

Due to the budget constraint $B_{max}$, we may still need to adjust
assignment pairs in set $rlt$ such that the total traveling cost is
below the maximum budget $B_{max}$. If the upper bound of the
traveling cost in set $rlt$ does not exceed budget $B_{max}$, then
we can directly return $rlt$ as $I_p$ (lines 12-13). Otherwise,
similar to GREEDY, we need to select ``best''
assignment pairs from set $rlt$ that are under the budget constraint
$B_{max}$ (maximizing the total quality score), and add them to the
set $I_p$, by calling procedure {\sf
MQA\_Budget\_Constrained\_Selection} (lines 14-15).
In particular, to adjust the budget, we select a ``best'' pair from
set $rlt$ each time that satisfies the budget constraint $B_{max}$
and with the highest quality score.
Please refer to details of procedure {\sf
MQA\_Budget\_Constrained\_Selection} in lines 17-28 of Figure
\ref{alg:dc}.

\noindent {\bf Discussions on Estimating the Best Number, $g$, of
the Decomposed Subproblems.} In order to reduce the computation cost
in our D\&C algorithm, we aim to select a best $g$ value that
minimizes the processing cost, in light of our proposed cost model.
Specifically, we formally model the computation cost,
$cost_{D\&C}$, of the D\&C algorithm, with respect to $g$, take the
derivative of $cost_{D\&C}$ over $g$, and then let the derivative be
0. This way, we can find the best $g$ value that minimizes the cost
of the D\&C algorithm. For details, please refer to Appendix C.

\begin{figure}[t!]
    \begin{center}
        \begin{tabular}{l}
            \parbox{3.1in}{
                \begin{scriptsize} \vspace{-4ex}
                    \begin{tabbing}
                        12\=12\=12\=12\=12\=12\=12\=12\=12\=12\=12\=\kill
                        {\bf Procedure {\sf MQA\_D\&C}} \{ \\
                        \> {\bf Input:} $n'$ current/future workers in $W$, and $m'$ current/future spatial tasks \\
                        \> \> \>\>   in $T$,  and a maximum budget $B_{max}$\\
                        \> {\bf Output:} an assignment instance set, $I_p$, with current/future workers/tasks\\
                        \> (1) \> \> $rlt = \emptyset$\\
                        \> (2) \> \> estimate the best number of subproblems, $g$, w.r.t. $W$ and $T$\\
                        \> (3) \> \> invoke {\sf MQA\_Decomposition}$(W, T, g)$, and obtain $g$ subproblems $M_s$\\
                        \> (4) \> \> \textbf{for} $s=1$ to $g$\\
                        \> (5) \> \> \> \textbf{if} the number of tasks in subproblem $M_s$ is greater than 1\\
                        \> (6) \> \> \> \> $rlt^{(s)}$ = {\sf MQA\_D\&C}($W (M_s)$, $T(M_s)$, $B_{max}$)\\
                        \> (7) \> \> \> \textbf{else}\\
                        \> (8) \> \> \> \> $rlt^{(s)} =$ {\sf MQA\_Greedy}($W (M_s)$, $T(M_s)$, $B_{max}$) \text{ }\\
                        \> (9) \> \> \textbf{for} $s=1$ to $g$\\
                        \> (10) \> \> \> find the next subproblem, $M_s$\\
                        \> (11) \> \> \> $rlt$ = {\sf MQA\_Merge} ($rlt$, $rlt^{(s)}$) \\
                        \> (12)\> \> \textbf{if} the upper bound of the traveling cost of $rlt$ $\leq B_{max}$\\
                        \> (13)\> \> \> \textbf{return} $rlt$\\
                        \> (14)\> \> \textbf{else}\\
                        \> (15)\> \> \> $I_p$ = {\sf MQA\_Budget\_Constrained\_Selection} ($rlt$, $B_{max}$)\\
                        \> (16)\> \> \textbf{return} $I_p$\}\\\\
                        {\bf Procedure {\sf MQA\_Budget\_Constrained\_Selection}} \{ \\
                        \> {\bf Input:} candidate pairs in $rlt$ and a maximum budget $B_{max}$\\
                        \> {\bf Output:} a worker-and-task assignment instance set, $I_p$, under the budget constraint\\
                        \> (17)\> \> $I_p=\emptyset$;\\
                        \> (18)\> \> \textbf{for} $k$ = 1 to $|rlt|$\\
                        \> (19)\> \> \> $S_p = \emptyset$\\
                        \> (20)\> \> \> \textbf{for} each assignment pair $\langle \tilde{w_i}, \tilde{t_j}\rangle \in rlt$\\
                        \> (21)\> \> \> \> \textbf{if} $\langle \tilde{w_i}, \tilde{t_j}\rangle$ has $lb\_\tilde{c_{ij}}$ greater than the remaining budget, then continue;\\
                        \> (22)\> \> \> \> \textbf{if} pair $\langle \tilde{w_i}, \tilde{t_j}\rangle$ cannot be pruned w.r.t. $S_p$ by Lemma \ref{lemma:dominate_pair}\\
                        \> (23)\> \> \> \> \> \textbf{if} pair $\langle \tilde{w_i}, \tilde{t_j}\rangle$ cannot be pruned w.r.t. $S_p$  by Lemma \ref{lemma:probabilistic_dominate}\\
                        \> (24)\> \> \> \> \> \> add $\langle \tilde{w_i}, \tilde{t_j}\rangle$ to $S_p$\\
                        \> (25)\> \> \> \> \> \> prune other candidate pairs in $S_p$ with $\langle \tilde{w_i}, \tilde{t_j}\rangle$\\
                        \> (26)\> \> \> select one best assignment pair $\langle \tilde{w_i}, \tilde{t_j}\rangle$ in $S_p$ satisfying the budget  \\
                        \>     \> \> \> \> constraint $B_{max}$ in Eq.~(\ref{eq:prune_high_travel_cost}) and with the highest probability  \\
                        \>     \> \> \> \> $M_{q, max} (\langle \tilde{w_i}, \tilde{t_j}\rangle)$ in Eq.~(\ref{eq:max_quality_prob})\\
                        \> (27)\> \> \> add the selected pair $\langle \tilde{w_i}, \tilde{t_j}\rangle$ to $I_p$\\
                        \> (28)\> \>  \textbf{return} $I_p$\}
                    \end{tabbing}
                \end{scriptsize}
            }
        \end{tabular}
    \end{center}\vspace{-3ex}
    \caption{\small The Divide-and-Conquer Algorithm.}
    \vspace{-5ex}
    \label{alg:dc}
\end{figure}

\section{Experimental Study}
\label{sec:exp}

\vspace{0.5ex}\noindent\textbf{Real/Synthetic Data Sets.} We tested
our proposed MQA processing algorithms over both real and
synthetic data sets. Specifically, for real data sets, we used two
check-in data sets, Gowalla \cite{cho2011friendship} and Foursquare
\cite{levandoski2012lars}. In the Gowalla data set, there are
196,591 nodes (users), with 6,442,890 check-in records. In the
Foursquare data set, there are 2,153,471 users, 1,143,092 venues,
and 1,021,970 check-ins, extracted from the Foursquare application
through the public API. Since most assignments happen in the same
cities, we extract check-in records within the area of San Francisco
(with latitude from \ang{37.709} to \ang{-122.503} and longitude
from \ang{37.839} to \ang{-122.373}), which has 8,481 Foursquare
check-ins, and 149,683 Gowalla check-ins for 6,143 users. We use
check-in records of Foursquare to initialize the location and
arrival time of tasks in the spatial crowdsourcing system, and
configure workers using the check-ins records from Gowalla. In other
words, we have 6,143 workers and 8,481 spatial tasks in the
experiments over real data. For simplicity, we first linearly map
check-in locations from Gowalla and Foursquare into a $[0,1]^2$ data
space, and then scale the arrival times of workers/tasks in the real
data accordingly. In order to generate workers/tasks for each time instance,
we evenly divide the entire time span of check-ins from two real
data sets into $R$ subintervals ($\in P$), and utilize check-ins in
each subinterval to initialize workers/tasks for the corresponding
time instance.

For synthetic data sets, we generate workers/tasks that join the
spatial crowdsourcing system for every time instance in the time instance set
$P$ as follows. We randomly produce locations of workers and tasks
in a 2D data space $[0, 1]^2$, following Gaussian $\mathcal{N}(0.5,
1^2)$ and Zipf distributions (skewness = 0.3), respectively. We also
test synthetic worker/task data with other distribution combinations
(e.g., Uniform-Zipf) and achieve similar results (see Appendix D).

For both real and synthetic data sets, we simulate the velocity
$v_i$ of each worker $w_i$ with Gaussian distribution $\mathcal{N}(\frac{v^-+v^+}{2},
(v^+-v^-)^2)$ within range
$[v^-, v^+]$, for $0< v^- \leq v^+ <1$, and the unit price $C$
w.r.t. the traveling distance $dist(\cdot, \cdot)$ varies from $5$
to $25$. Regarding the time constraint (i.e., the deadline $e_j$) of
spatial tasks $t_j$, we produce the arrival deadlines of tasks
within the range $[e^-, e^+]$, which are given by the remaining time
for workers to arrive at tasks after these tasks join the system.
Moreover, for the total quality score, $q_{ij}$, of worker-and-task
assignments, we randomly generate $q_{ij}$ with Gaussian
distributions within $[q^-, q^+]$. In our experiments, we also test
the size, $w$, of the sliding window with values from $1$ to $5$,
and the number, $R$, of time instances in $P$ from $10$ to
$25$.

\begin{table}[t]
	\begin{center}\vspace{-6ex}
		\caption{\small Experimental Settings.} \label{table2}\vspace{-2ex}
		{\small\scriptsize
			\begin{tabular}{l|l}
				{\bf \qquad \qquad \quad Parameters} & {\bf \qquad \qquad \qquad Values} \\ \hline \hline
				the size of sliding windows $w$ & 1, 2, \textbf{3}, 4, 5 \\
				the budget $B$  & 100, \textbf{200}, 300, 400, 500\\
				the quality range $[q^-, q^+]$  & [0.25, 0.5], [0.5, 1], \textbf{[1, 2]}, [2, 3], [3, 4]\\
				the deadline range $[e^-, e^+]$  & [0.25, 0.5], [0.5, 1], \textbf{[1, 2]}, [2, 3], [3, 4]\\
				the velocity range $[v^-, v^+]$   & [0.1, 0.2], \textbf{[0.2, 0.3]}, [0.3, 0.4], [0.4, 0.5]\\
				the unit price w.r.t. distance $C$ & 5, \textbf{10}, 15, 20\\
				the number, $R$, of time instances & 10, \textbf{15}, 20, 25 \\
				the number, $m$, of tasks  & 1K,  \textbf{3K}, 5K, 8K, 10K\\
				the number, $n$, of workers  & 1K, \textbf{3K}, 5K, 8K, 10K \\
				\hline
			\end{tabular}
		}
	\end{center}\vspace{-8ex}
\end{table}

\vspace{0.5ex}\noindent\textbf{Measures and Competitors.} We
evaluate the effectiveness and efficiency of our MQA processing
approaches, in terms of the overall quality score and the CPU time.
Specifically, the overall quality score is defined in
Eq.~(\ref{eq:tbc_sc}), which can measure the quality of the
assignment strategy, and the CPU time is given by the average time
cost of performing MQA assignments at each time instance.

In our MQA problem, regarding the effectiveness, we will compare
our MQA approaches with a straightforward method that conducts the
assignment over current and future time instances separately (i.e., without
prediction), in terms of the overall quality score. Moreover, we will
also compare our MQA approaches, the \textit{MQA greedy} (GREEDY) and \textit{MQA divide-and-conquer} (D\&C) algorithms, with a random (RANDOM) method (which randomly assigns workers to
spatial tasks under the budget constraint). Since RANDOM
does not take into account the quality of tasks, it is expected to
achieve worse quality than our MQA approaches (although it is
expected to be more efficient than MQA approaches).

\vspace{0.5ex}\noindent\textbf{Experimental Settings.} Table
\ref{table2} shows our experimental settings, where the default
values of parameters are in bold font. In subsequent experiments,
each time we vary one parameter, while setting others to their
default values. All our experiments were run on an Intel Xeon X5675
CPU @3.07 GHZ with 32 GB RAM in Java.

\vspace{-1ex}
\subsection{Effectiveness of the MQA Approaches}
\label{subsec:effectiveness}

\vspace{0.5ex}\noindent {\bf The Comparison of the Prediction
Accuracy.} We first evaluate the prediction accuracy of future
workers/tasks in our MQA approach, by comparing the estimated
numbers, $est$, of workers/tasks in cells with actual ones, $act$,
in terms of the relative error (i.e., defined as
$\frac{|est-act|}{act}$), where the size, $w$, of the sliding window
to do the prediction (via linear regressions) varies from 1 to 5. In Figure \ref{fig:relative_errors}, we present the average relative error of our grid-based prediction method, which is the result of dividing the summation of the relative errors of all the cells by the number of cells (e.g., 400 cells).  For both synthetic
(marked with S) and real data (marked with R), average relative errors for
different window sizes $w$ are not very sensitive to $w$. Only for
real data, the average relative error of predicting the number of workers
slightly increases with larger $w$ value. This is because the
distribution of workers changes quickly over time in real data,
which leads to larger prediction error by using wider window size
$w$. Nonetheless, average relative errors of predicting the number of
workers/tasks remain low (i.e., less than 5.5\%) over all
real/synthetic data for different window sizes $w$, which indicates
good accuracy of our grid-based prediction approach.

We also conducted experiments on the synthetic dataset with varying window size $w$ from 1 to 5 on three different workers distributions (e.g., Gaussian, Zipf and Uniform) to show the influence of the distributions of workers on accuracies. Due to space limitations, we put the results in Appendix F.

\begin{figure}[t!]\centering\vspace{-4ex}
	\scalebox{0.22}[0.22]{\includegraphics{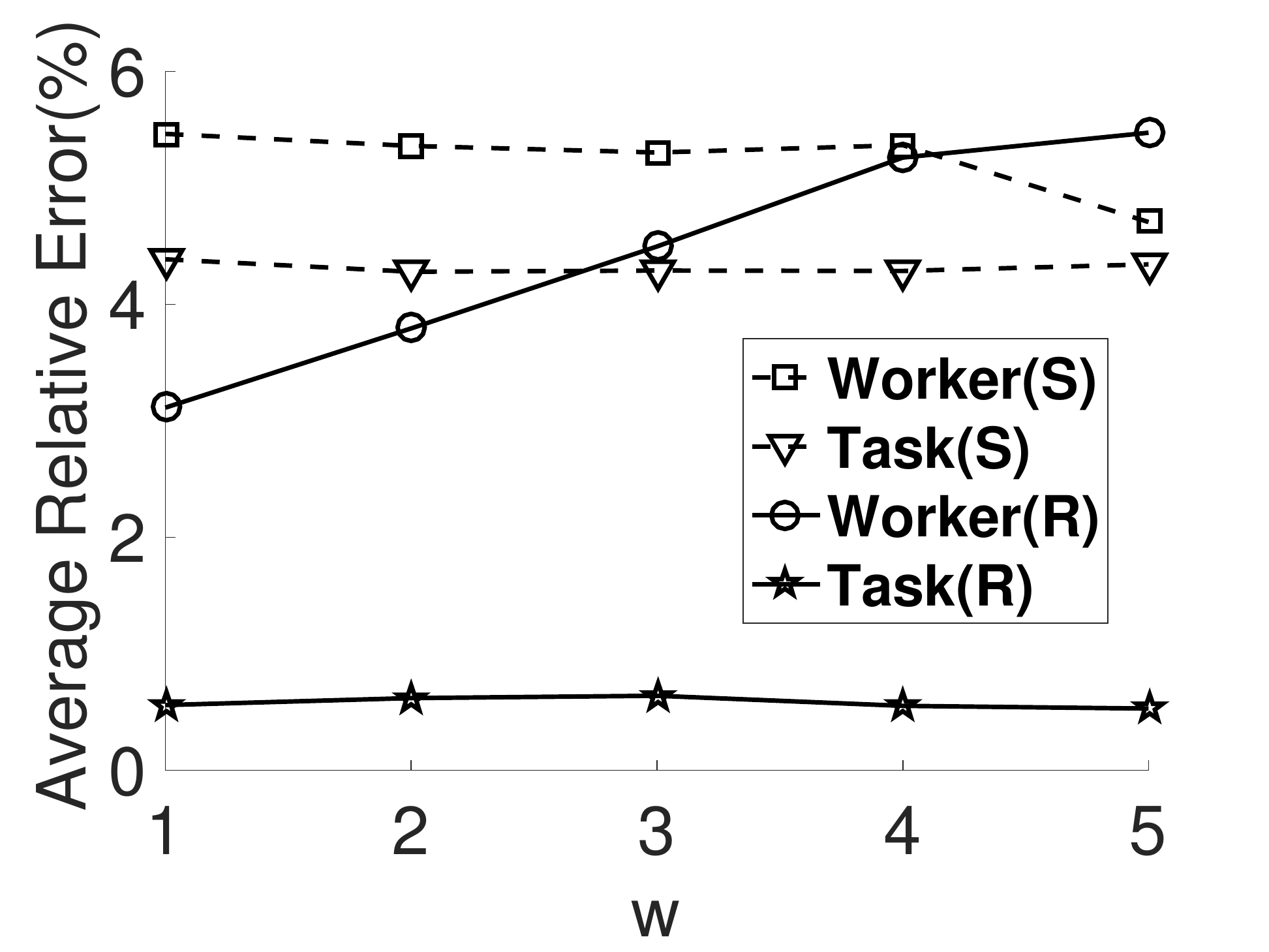}}\vspace{-2ex}
	\caption{\small The Prediction Accuracy vs. Window Size $w$.}\vspace{-5ex}
	\label{fig:relative_errors}
\end{figure}

\vspace{0.5ex}\noindent {\bf Comparison with a Straightforward
Method.} Figure \ref{fig:budget} compares the quality score of our
MQA approaches (with predicted workers/tasks) with that of the
straightforward method which selects assignments in current and next
time instances separately (without predictions), where budget $B$ varies
from $100$ to $500$. We denote MQA approaches with prediction as
GREEDY\_WP, D\&C\_WP, and RANDOM\_WP, and those without prediction
as GREEDY\_WoP, D\&C\_WoP, and RANDOM\_WoP, respectively.

In Figure \ref{subfig:b_score}, we can see that, the quality scores
of MQA with predictions (with solid lines) are higher than that of
MQA without prediction (with dash lines), for different budget
$B$. This indicates the effectiveness of our MQA approaches over
current/predicted workers/tasks, which can achieve better assignment
strategy than the ones without prediction (i.e., local optimality).
Moreover, either with or without predictions, D\&C incurs higher
quality scores than GREEDY (since D\&C is carefully designed to find
assignments with high quality scores via divide-and-conquer), and
RANDOM has the lowest score, which implies good quality of our
proposed assignment strategies.

Figure \ref{subfig:b_cpu} illustrates the average running time of
our GREEDY and D\&C approaches, compared with RANDOM, for each
time instance. In the figure, due to the prediction and merge costs,
D\&C\_WP requires the highest CPU time to solve the MQA problem,
which trades the efficiency for the accuracy. When the budget $B$
increases, the running time of D\&C decreases. This is because
larger budget $B$ leads to lower cost of selecting assignment pairs
under the budget constraint (i.e., lines 12-15 in procedure {\sf
MQA\_D\&C} of Figure \ref{alg:dc}). For other approaches, the time
cost remains low  (i.e., less than 5 seconds) for different $B$
values.

Due to space limitation, we put the results on varying other parameters in Appendix G.

\vspace{-1ex}
\subsection{Performance of the MQA Approaches}
\label{subsec:efficiency}

\vspace{0.5ex}\noindent {\bf The MQA Performance vs. the Range,
$[q^-, q^+]$, of Quality Score $q_{ij}$.} Figure \ref{fig:quality}
illustrates the experimental results on different ranges, $[q^-,
q^+]$, of quality score $q_{ij}$ from $[0.25,0.5]$ to $[3,4]$ on
real data. In Figure \ref{subfig:quality_score}, with the increase
of score ranges, quality scores of all the three approaches
increase. D\&C has higher quality score than GREEDY, which are both
higher than RANDOM. From Figure \ref{subfig:quality_cpu}, RANDOM is
the fastest (however, with the lowest quality score), since it
randomly selects assignments without considering the quality score
maximization. D\&C has higher running time than GREEDY (however,
higher quality scores than GREEDY). Nonetheless, the running times
of GREEDY remain low.

\vspace{0.5ex}\noindent {\bf The MQA Performance vs. the Range,
$[e^-, e^+]$, of Tasks' Deadlines $e_j$.} Figure \ref{fig:deadline}
shows the effect of the range, $[e^-, e^+]$, of tasks' deadlines
$e_j$ on the MQA performance over real data, where $[e^-, e^+]$
changes from $[0.25, 0.5]$ to $[2, 3]$. In Figure
\ref{subfig:deadline_score}, when the range $[e^-, e^+]$ becomes
larger, quality scores of all three approaches also increase. Since
a more relaxed (larger) deadline $e_j$ of a task $t_j$ can be
performed by more valid workers, it thus leads to higher quality
score (that can be achieved) and processing time (as confirmed by
Figure \ref{subfig:deadline_cpu}). Similar to previous results, D\&C
can achieve higher quality scores than GREEDY, and both of them
outperform RANDOM. Furthermore, GREEDY needs higher time cost than
RANDOM, and has lower running time than D\&C.

\begin{figure}[t!]
	\centering\vspace{-4ex}
	\subfigure[][{\scriptsize Quality Score}]{
		\scalebox{0.2}[0.2]{\includegraphics{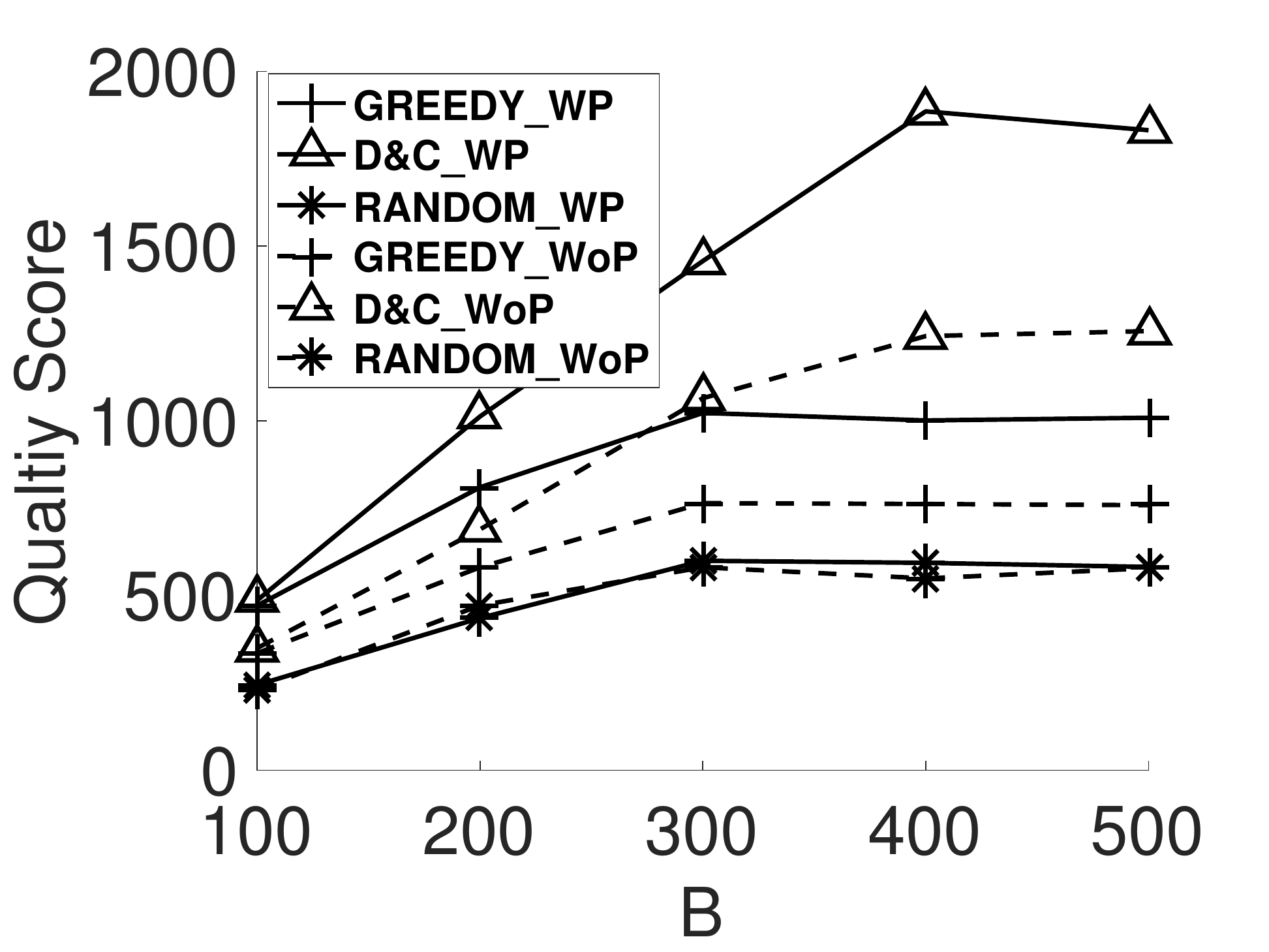}}
		\label{subfig:b_score}}
	\subfigure[][{\scriptsize Running Time}]{
		\scalebox{0.2}[0.2]{\includegraphics{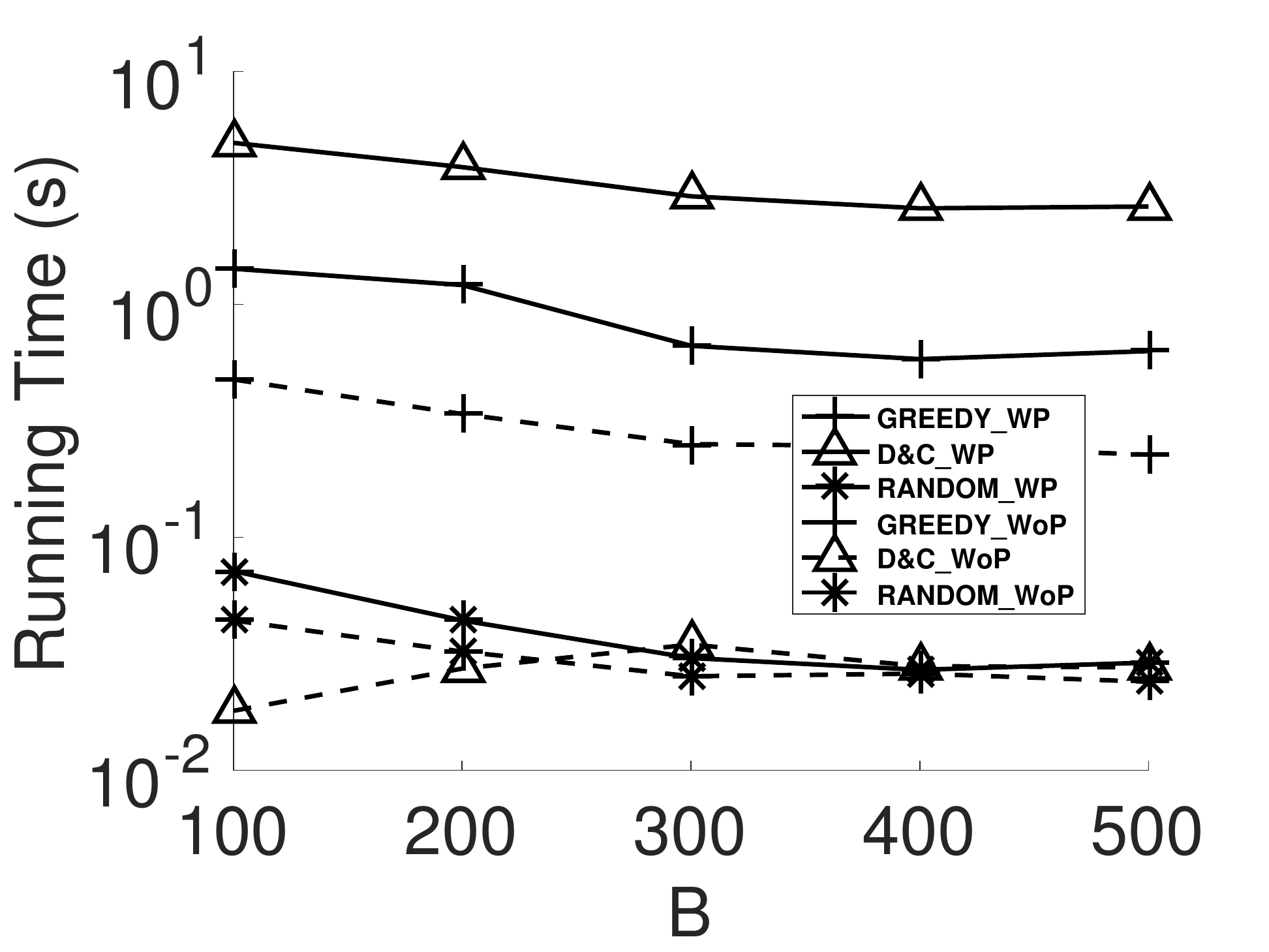}}
		\label{subfig:b_cpu}}\vspace{-2ex}
	\caption{\small Effect of the Budget $B$(Synthetic Data).}\vspace{-1ex}
	\label{fig:budget}
\end{figure}

\begin{figure}[t!]
	\centering\vspace{-2ex}
	\subfigure[][{\scriptsize Quality Score}]{
		\scalebox{0.2}[0.2]{\includegraphics{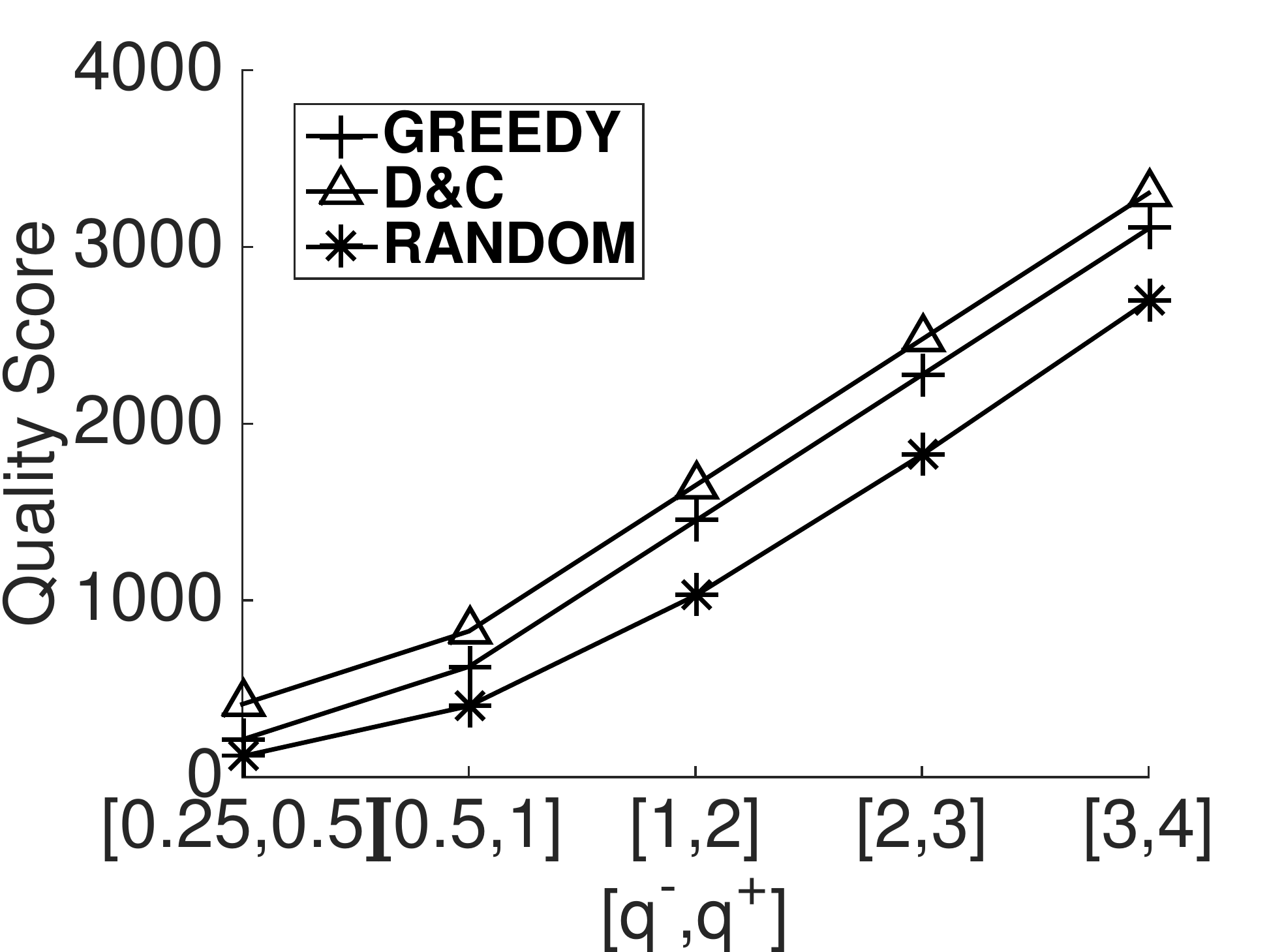}}
		\label{subfig:quality_score}}
	\subfigure[][{\scriptsize Running Time}]{
		\scalebox{0.2}[0.2]{\includegraphics{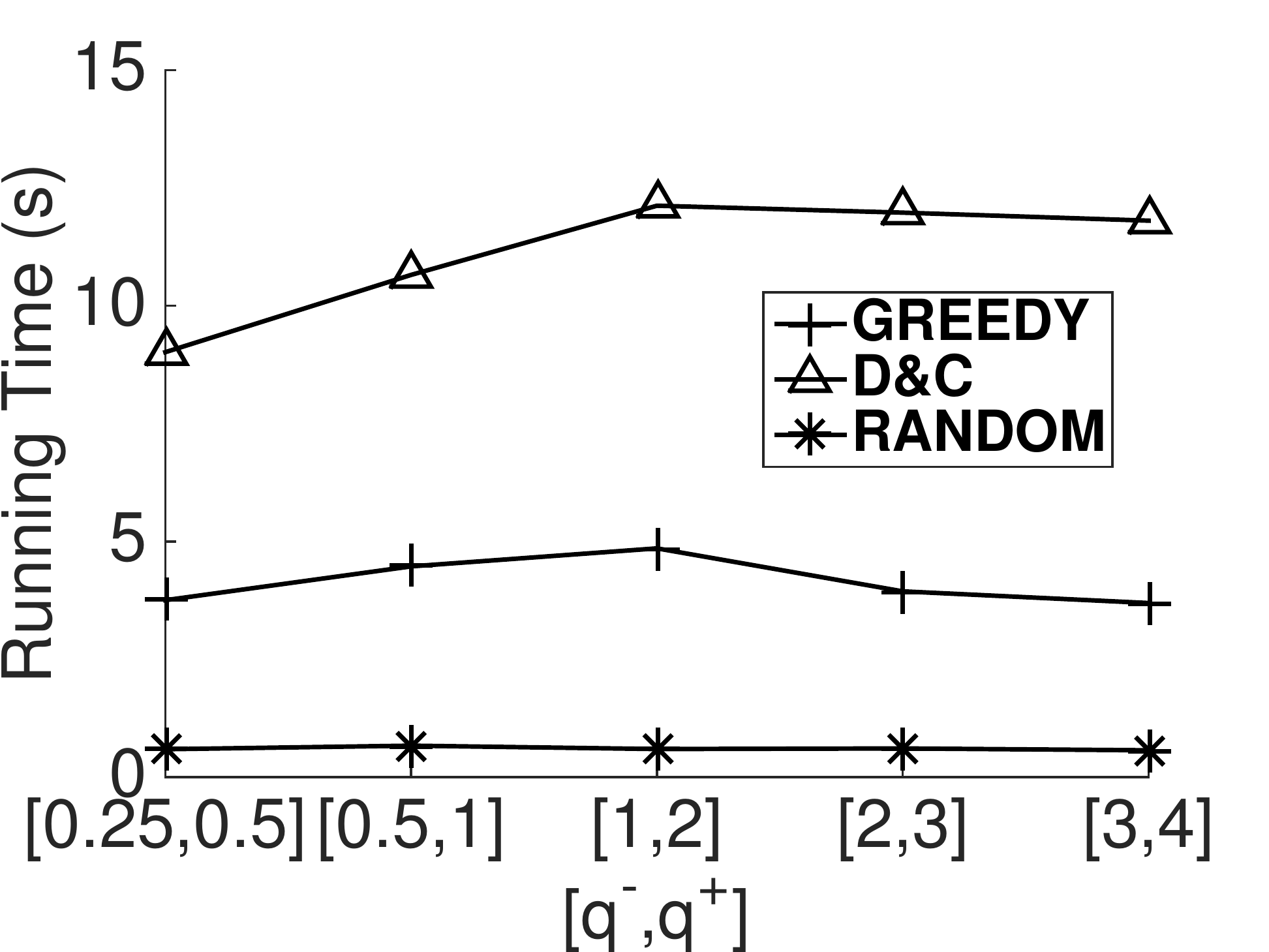}}
		\label{subfig:quality_cpu}}\vspace{-2ex}
	\caption{\small Effect of the Range of Quality Score $q_{ij}$  (Real Data).}\vspace{-5ex}
	\label{fig:quality}
\end{figure}

\vspace{0.5ex}\noindent {\bf The MQA Performance vs. the Range,
$[v^-, v^+]$, of Workers' Velocities $v_i$.} Figure
\ref{fig:velocity} presents the MQA performance with different
ranges, $[v^-, v^+]$, of workers' velocities $v_i$ from $[0.1, 0.2]$
to $[0.4, 0.5]$ on synthetic data, where default values are used for
other parameters. In Figure \ref{subfig:velocity_score}, when the
value range, $[v^-, v^+]$, of velocities of workers increases, the
total quality scores of all three approaches decrease. When the
range of velocities increases, some worker-and-task pairs with long
distances may become valid, which may quickly consume the available
budget $B$ and in turn reduce the number of selected pairs. Thus,
for larger workers' velocities, the resulting overall quality score
for the selected pairs in GREEDY and D\&C approaches decreases. With
different velocities, D\&C always has higher quality scores than
GREEDY, followed by RANDOM. In Figure \ref{subfig:velocity_cpu}, the
running times of GREEDY and D\&C increase for larger workers'
velocities, since there are more worker-and-task assignments to
process. Moreover, RANDOM has the smallest time cost (however, the
worst quality score), whereas GREEDY runs faster than D\&C, and
slower than RANDOM.

\vspace{0.5ex}\noindent {\bf The MQA Performance vs. the Number,
$m$, of Tasks .} Figure \ref{fig:task} varies the total number, $m$,
of spatial tasks for $R$ ($=15$ by default) time instance from $1K$ to
$10K$ on synthetic data sets, where other parameters are set to
their default values. From the experimental results, with the
increase of $m$, the total quality scores and time costs of all the
three approaches both increase smoothly. This is reasonable, since
more tasks lead to more valid pairs, which incur higher quality
score for selected assignment pairs, and require higher time cost to
process. D\&C can achieve higher quality score than GREEDY, which
indicates the effectiveness of our proposed D\&C approach. Moreover,
RANDOM has the worst quality score among the three tested
approaches.

\begin{figure}[t!]
	\centering\vspace{-4ex}
	\subfigure[][{\scriptsize Quality Score}]{
		\scalebox{0.2}[0.2]{\includegraphics{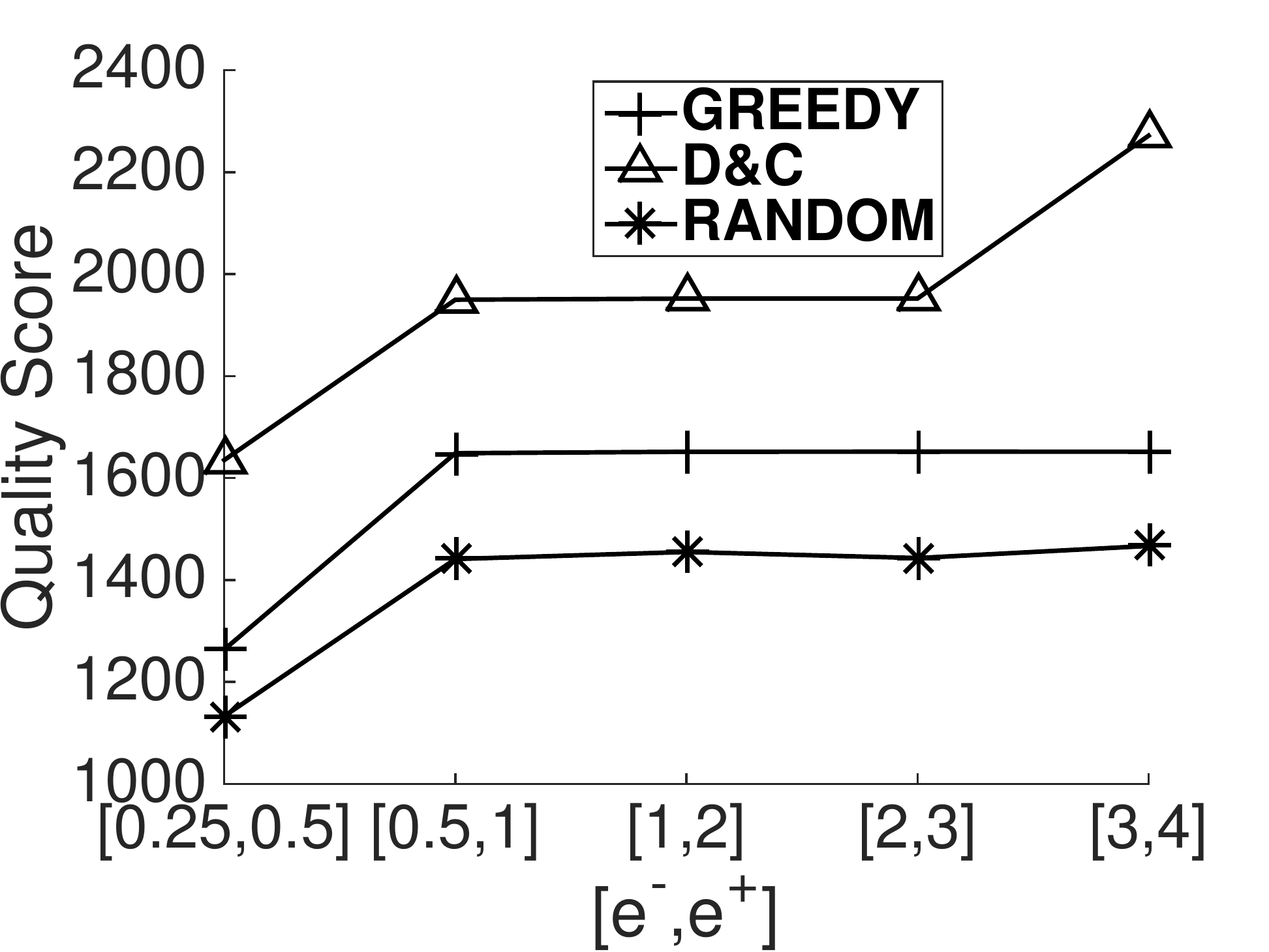}}
		\label{subfig:deadline_score}}
	\subfigure[][{\scriptsize Running Time}]{
		\scalebox{0.2}[0.2]{\includegraphics{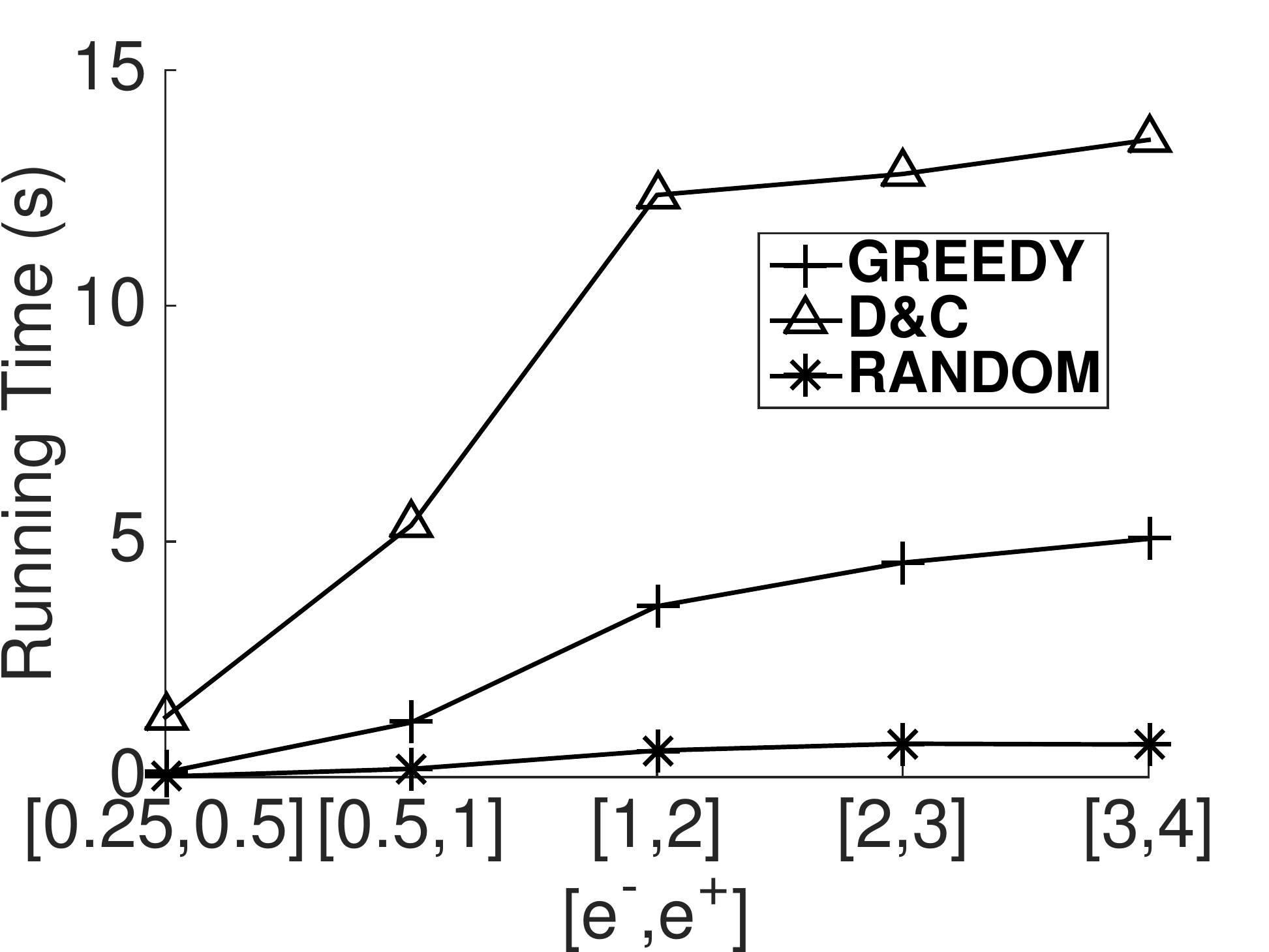}}
		\label{subfig:deadline_cpu}}\vspace{-2ex}
	\caption{\small Effect of the Range of Tasks' Deadlines $e_j$ (Real Data).}\vspace{-2ex}
	\label{fig:deadline}
\end{figure}

\begin{figure}[t]
	\centering
	\subfigure[][{\scriptsize Quality Score}]{
		\scalebox{0.2}[0.2]{\includegraphics{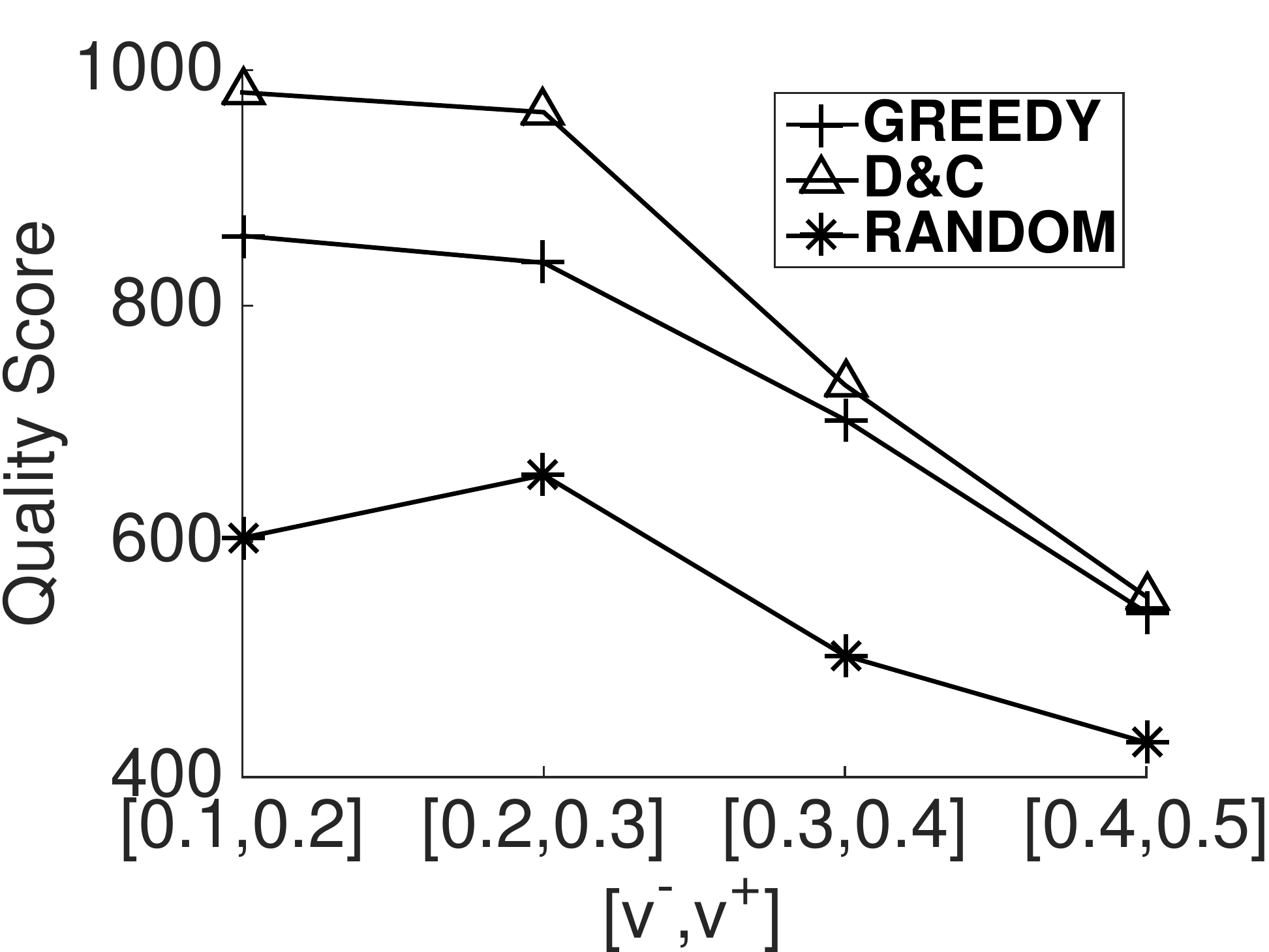}}
		\label{subfig:velocity_score}}
	\subfigure[][{\scriptsize Running Time}]{
		\scalebox{0.2}[0.2]{\includegraphics{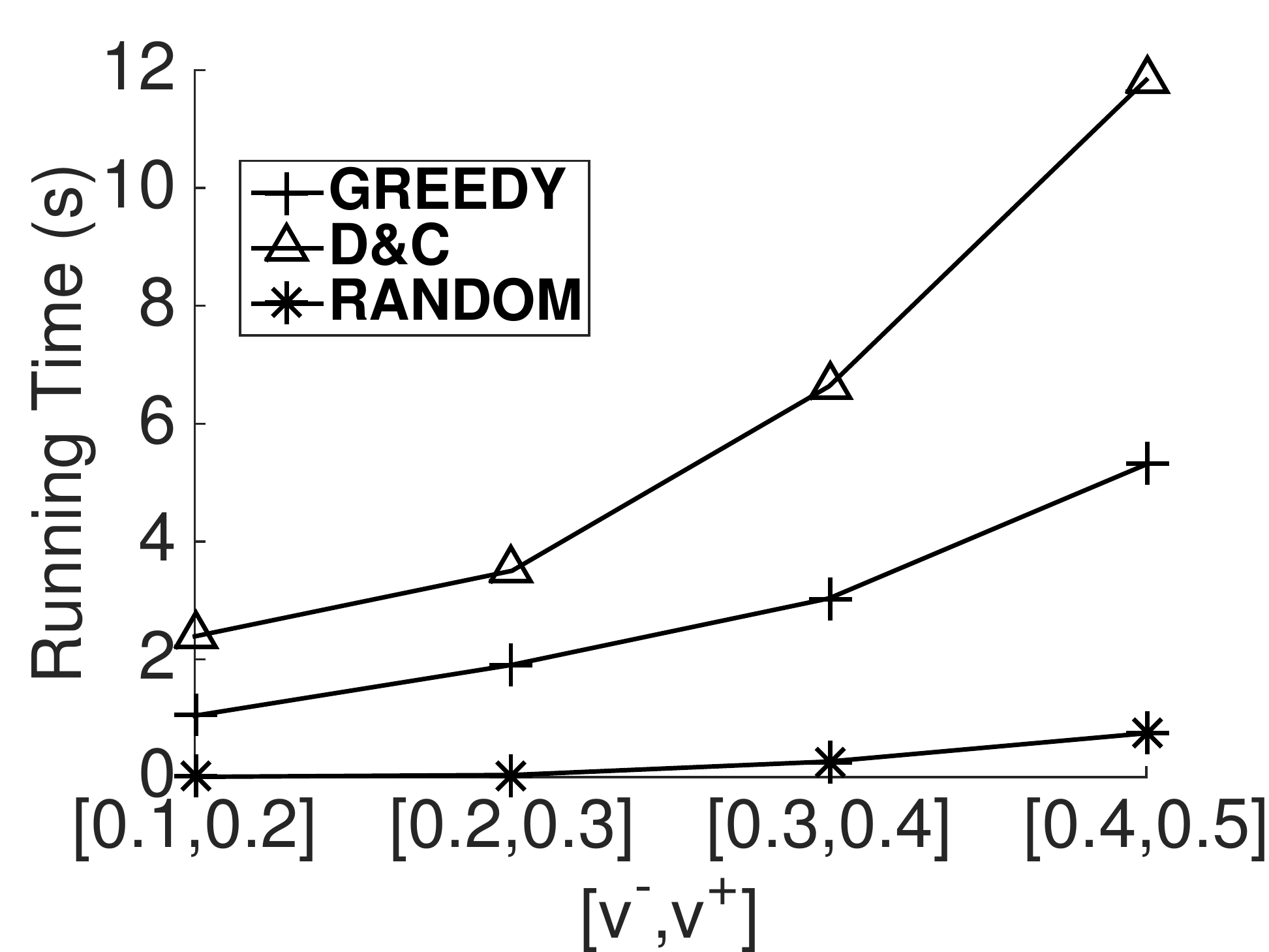}}
		\label{subfig:velocity_cpu}}\vspace{-2ex}
	\caption{\small Effect of the Range of Velocities $[v^-, v^+]$ (Synthetic Data).}\vspace{-5ex}
	\label{fig:velocity}
\end{figure}

\vspace{0.5ex}\noindent {\bf The MQA Performance vs. the Number,
$n$, of Workers.} Figure \ref{fig:worker} examines the effect of the
number, $n$, of workers for $R$ ($=15$ by default) time instances on the
MQA performance, where $n$ varies from $1K$ to $10K$ and other
parameters are set to default values. Similarly, both overall quality
scores and running times of three tested approaches increase, for
larger $n$ values. Our GREEDY and D\&C algorithms can achieve better
quality scores than RANDOM. Further, GREEDY has lower time costs than
D\&C. Nonetheless, the time cost of our approaches smoothly grows
with the increasing $n$ values, which indicates good scalability of
our MQA approaches.

Due to space limitations, please refer to the experimental 
results of the effects of the number,
$R$, of time instances and  the unit
price $C$ w.r.t. distance $dist(w_i, t_j)$ in Appendix E.

\section{related work}
\label{sec:related}

There are many previous studies in crowdsourcing systems
\cite{bulut2011crowdsourcing}, which usually allow
workers to accept task requests and accomplish tasks online.
However, these workers do not have to travel to some sites to perform
tasks. In contrast, the spatial crowdsourcing system
\cite{deng2013maximizing, kazemi2012geocrowd} requires workers
traveling to locations of spatial tasks, and completes tasks such as
taking photos/videos. For instance, some related studies
\cite{cornelius2008anonysense, kanhere2011participatory} studied the
problem of using smart devices (taken by workers) to collect data in
real-world applications.

Kazemi and Shahabi \cite{kazemi2012geocrowd} classified the spatial
crowdsourcing systems from two perspectives: workers' motivation and
publishing models. Regarding the workers' motivation, there are two
types, reward-based and self-incentivised, which inspires workers to
do tasks by rewards or volunteering, respectively. Moreover, based
on the publishing models, there are two modes, \textit{worker
selected tasks} (WST) \cite{deng2013maximizing} and
\textit{server assigned tasks} (SAT)
\cite{kazemi2012geocrowd,cheng2014reliable, liu2016cost, cheng2016multi}, in which tasks are accepted by workers or assigned by workers,
respectively. What is more, to handle the requests more quickly, existing studies \cite{tong2016exp, tong2016online, tong2017online} proposed methods to online process spatial crowdsourcing requests with quality guarantees.

Our MQA problem is reward-based and follows the SAT mode. Prior
studies in the SAT mode
\cite{kazemi2012geocrowd,cheng2014reliable, liu2016cost, cheng2016multi}
assigned existing workers to tasks in the spatial crowdsourcing
system with distinct goals, for example, maximizing the number of
the completed tasks on the server side \cite{kazemi2012geocrowd}, minimizing the traveling cost for completing a set of given tasks \cite{liu2016cost}, or the reliability-and-diversity score of
assignments \cite{cheng2014reliable}. In contrast, our MQA problem
in this paper has a different goal, that is, maximizing the overall
quality score of assignments and under the budget constraint. As a
result, we design specific \textit{MQA greedy} and \textit{MQA divide-and-conquer} algorithms for our
MQA problem, that maximize the quality score (under the budget
constraint) rather than other metrics (e.g., the
reliability-and-diversity score in \cite{cheng2014reliable}), which
cannot directly borrow from previous works.

Most importantly, different from the algorithms in existing studies \cite{cheng2014reliable, cheng2016multi} that deal with deterministic workers and tasks, our MQA approaches need to handle predicted workers and tasks, whose locations and other attributes (e.g., distances, quality scores, etc.) are variables (rather than fixed values). Thus, our
MQA approaches aims to design the assignment strategy over both
current and predicted workers/tasks, which has not been investigated
before. Therefore, we cannot directly apply previous techniques
(proposed for workers/tasks without prediction) to tackle our MQA
problem.

\begin{figure}[t!]
	\centering\vspace{-5ex}
	\subfigure[][{\scriptsize Quality Score}]{
		\scalebox{0.2}[0.2]{\includegraphics{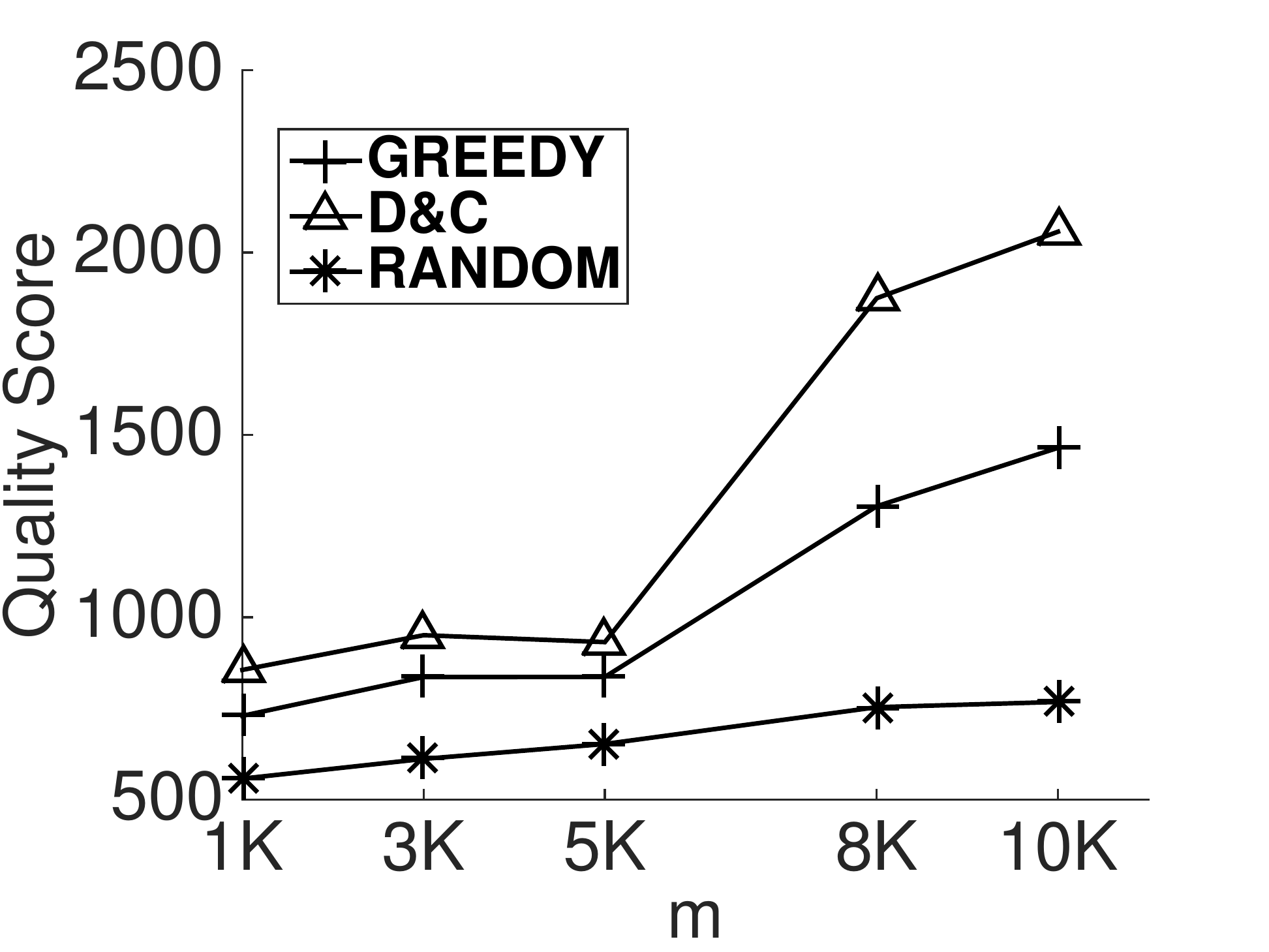}}
		\label{subfig:task_score}}
	\subfigure[][{\scriptsize Running Time}]{
		\scalebox{0.2}[0.2]{\includegraphics{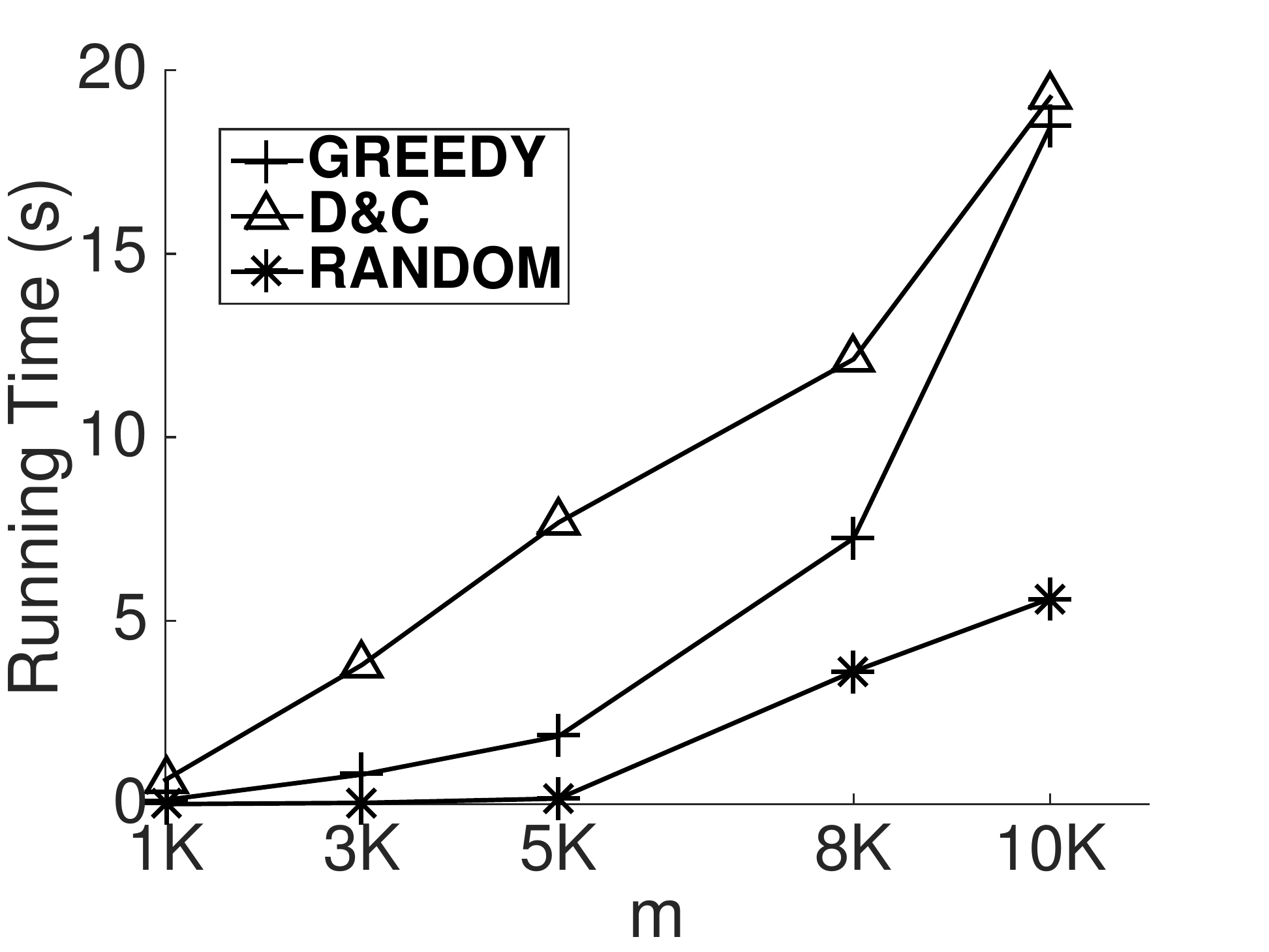}}
		\label{subfig:task_cpu}}\vspace{-2ex}
	\caption{\small Effect of the Number, $m$, of Tasks  (Synthetic Data).}
	\label{fig:task}\vspace{-4ex}
\end{figure}

\vspace{-1ex}
\section{conclusion}
\label{sec:conclusion}

In this paper, we studied a spatial crowdsourcing problem,
named \textit{maximum quality task assignment} (MQA), which
assigns moving workers to spatial tasks satisfying the budget
constraint of the traveling cost and achieving a high overall
quality score. In order to provide better global assignments, our MQA approaches are based on the assignment
selection strategy over both current and (predicted) future
workers/tasks. We propose an accurate prediction
approach to estimate both quality/location distributions of
workers/tasks. We prove that the MQA problem is NP-hard, and thus
intractable. Alternatively, we propose efficient heuristics, including \textit{MQA greedy} and \textit{MQA divide-and-conquer} approaches, to
obtain better global assignments. Extensive experiments have been conducted to
confirm the efficiency and effectiveness of our proposed MQA
processing approaches.

\section{Acknowledgment}

Peng Cheng and Lei Chen are supported in part by Hong Kong RGC Project N HKUST637/13, NSFC Grant No. 61328202, NSFC Guang Dong Grant No. U1301253, National Grand Fundamental Research 973 Program of China under Grant 2014CB340303, HKUST-SSTSP FP305, Microsoft Research Asia Gift Grant, Google Faculty
Award 2013 and ITS170.  Xiang Lian is supported by Lian Start Up No. 220981. Cyrus Shahabi's work has been funded in part by NSF grants IIS-1320149 and CNS-1461963.

\begin{figure}[t!]
	\centering\vspace{-5ex}
	\subfigure[][{\scriptsize Quality Score}]{
		\scalebox{0.2}[0.2]{\includegraphics{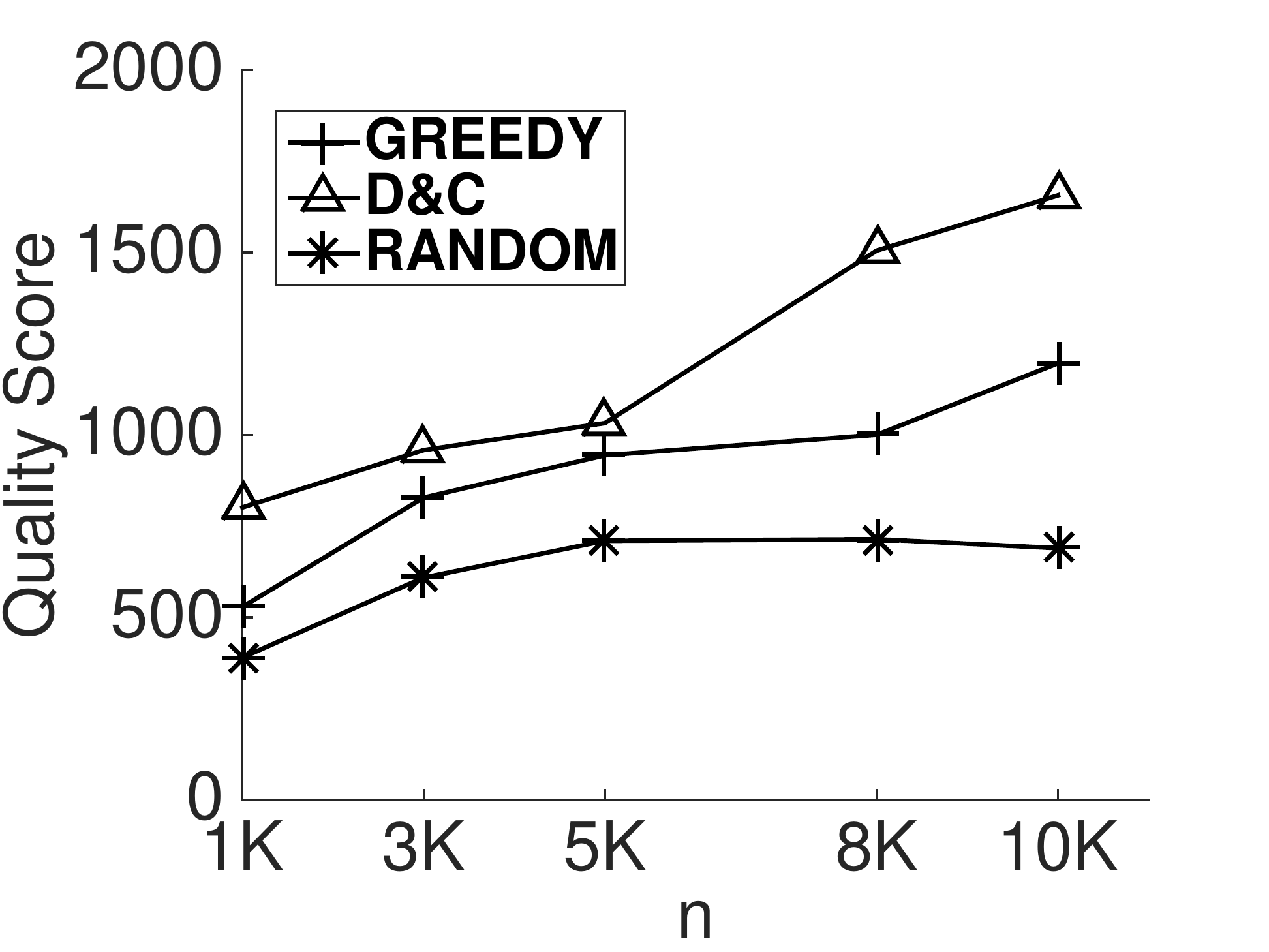}}
		\label{subfig:worker_score}}
	\subfigure[][{\scriptsize Running Time}]{
		\scalebox{0.2}[0.2]{\includegraphics{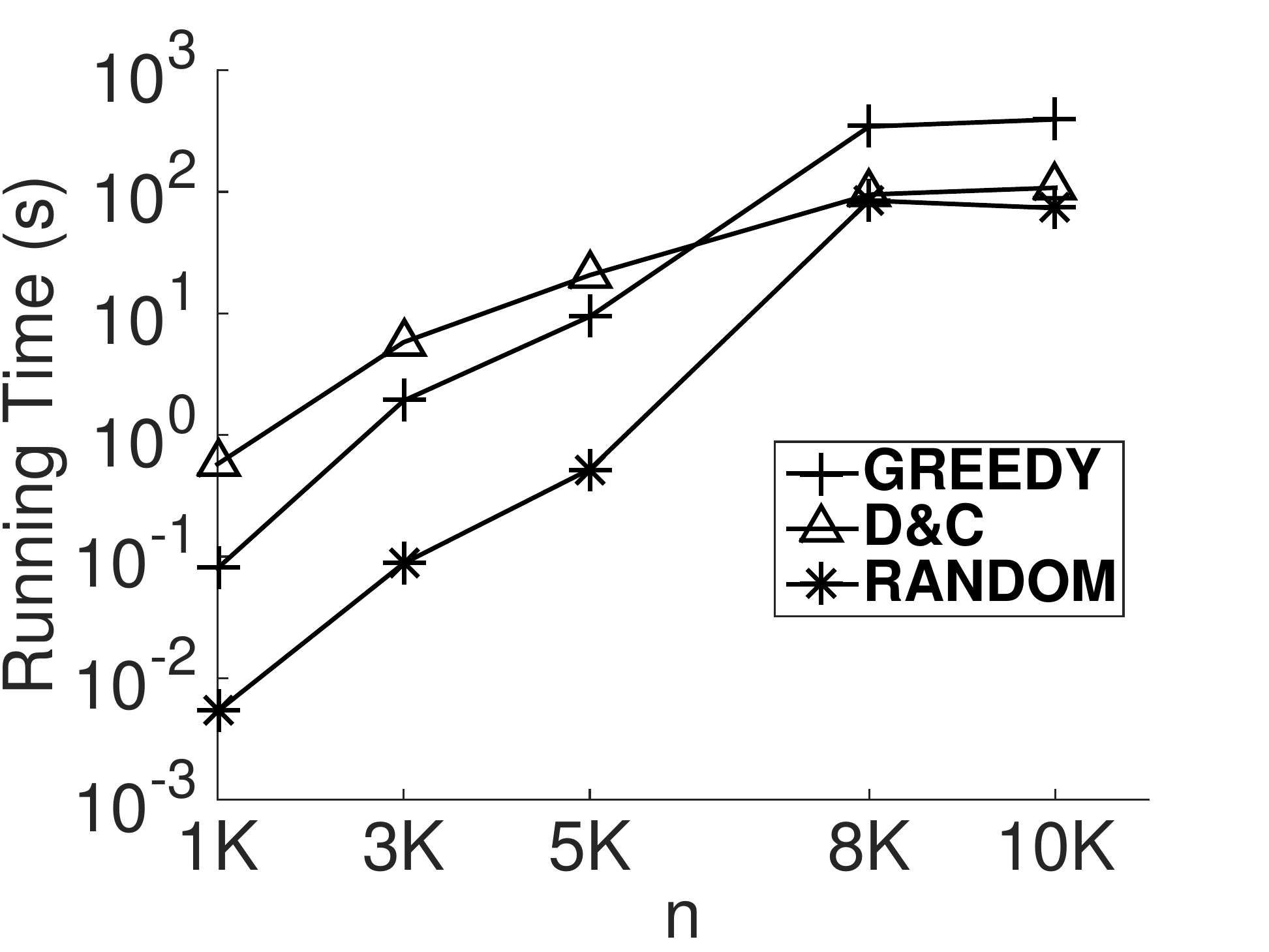}}
		\label{subfig:worker_cpu}}\vspace{-2ex}
	\caption{\small Effect of the Number, $n$, of Workers (Synthetic Data).}\vspace{-4ex}
	\label{fig:worker}
\end{figure}

\bibliographystyle{abbrv}
\let\xxx=\bibitem\def\bibitem{\par\vspace{-1mm}\xxx}
\bibliography{references/add}

\balance
\newpage

\appendix

\subsection{Proof of Lemma 2.1}
\begin{proof}
	We prove the lemma by a reduction from the 0-1 knapsack problem. A
	0-1 knapsack problem can be described as follows: Given a set, $C$,
	of $n$ items $a_i$ numbered from 1 up to $n$, each with a weight
	$w_i$ and a value $v_i$, along with a maximum weight capacity $W$,
	the 0-1 knapsack problem is to find a subset $C'$ of $C$ that
	maximizes $\sum_{a_i \in C'}v_i$ subjected to $\sum_{a_i \in C'}w_i
	\leq W$.

	For a given 0-1 knapsack problem, we can transform it to an instance
	of MQA as follows: at timestamp $p$, we give $n$ pairs of worker
	and task, such that for each pair of worker and task $\langle w_i,
	t_i\rangle$, the traveling cost $c_{ii} = w_i$ and the quality score
	$q_{ii} = v_i$. Also, we set the global budget $B=W$. At the same
	time, for any pair of worker and task $\langle w_i, t_j\rangle$
	where $t_j \neq t_i$, we make $c_{ij} \gg c_{ii}$ and $q_{ij} \leq
	q_{ii}$. Thus, in the assignment result, it is only possible to
	select worker-and-task pairs with same subscripts. Then, for this
	MQA instance, we want to achieve an assignment instance set $I_p$
	of some pairs of worker and task with same subscripts that maximizes
	the quality score $\sum_{\forall \langle w_i, t_i\rangle \in
		I_p}q_{ii}$ subjected to $\sum_{\forall \langle w_i, t_i\rangle \in
		I_p}c_{ii} \leq B$.
	
	Given this mapping, we can show that the 0-1 knapsack problem
	instance can be solved, if and only if the transformed MQA problem
	can be solved.
	
	This way, we can reduce 0-1 knapsack problem to the MQA problem.
	Since 0-1 knapsack problem is known to be NP-hard
	\cite{vazirani2013approximation}, MQA is also NP-hard, which
	completes our proof.\vspace{-1ex}
\end{proof}

\subsection{The Pseudo Code of the Grid-based Prediction Algorithm}

Figure \ref{alg:grid-predict} shows the pseudo code of our
grid-based prediction algorithm, namely {\sf MQA\_Prediction},
which predicts the number of workers/tasks in each cell, $cell_i$,
of the grid index by using the linear regression (lines 3-4), and
generates worker/task samples with the estimated sizes (lines 5-6).

    \begin{figure}[h]
        \begin{center}
            \begin{tabular}{l}
                \parbox{3.1in}{
                    \begin{scriptsize} \vspace{-4ex}
                        \begin{tabbing}
                            12\=12\=12\=12\=12\=12\=12\=12\=12\=12\=12\=\kill
                            {\bf Procedure {\sf MQA\_Prediction}} \{ \\
                            \> {\bf Input:} worker sets $\mathbb{W} = \{W_{p-w+1}, ..., W_p\}$ and task sets $\mathbb{T}=\{T_{p-w+1},$  \\
                            \>\>\>\> $...,$ $T_p\}$ at the latest $w$ time instances, and a grid index $\mathcal{I}$\\
                            \> {\bf Output:} future workers and tasks in $W_{p+1}$ and $T_{p+1}$, respectively, for the next\\
                            \> \> \> \> \>  time instance at timestamp $(p+1)$\\
                            \> (1) \> \> let $W_{p+1} = \emptyset$ and $T_{p+1} = \emptyset$\\
                            \> (2) \> \> \textbf{for} each cell $cell_i$ in $\mathcal{I}$\\
                            \> (3) \> \> \> estimate the future number, $|W_{p+1}^{(i)}|$, of workers in $cell_i$ by the \textit{linear regression}\\
                            \> (4) \> \> \> estimate the future number, $|T_{p+1}^{(i)}|$, of tasks in $cell_i$ by the \textit{linear regression}\\
                            \> (5) \> \> \> randomly generate $|W_{p+1}^{(i)}|$ worker samples for $cell_i$, and add them to $W_{p+1}$\\
                            \> (6) \> \> \> randomly generate $|T_{p+1}^{(i)}|$ task samples for $cell_i$, and add them to $T_{p+1}$\\
                            \> (7) \> \> \textbf{return} $W_{p+1}$ and $T_{p+1}$\\
                            \}
                        \end{tabbing}
                    \end{scriptsize}
                }
            \end{tabular}
        \end{center}\vspace{-2ex}
        \caption{\small The Grid-Based Worker/Task Prediction Algorithm.}

        \label{alg:grid-predict}
    \end{figure}

\subsection{Cost-Model-Based Estimation of the Best Number of the Decomposed Subproblems}
\label{subsec:DC_cost_model}

\noindent {\bf The Cost, $F_D$, of Decomposing Subproblems.} From
Algorithm {\sf MQA\_Decomposition} (in Figure
\ref{alg:decomposing}), we first need to retrieve all valid
worker-and-task assignment pairs (line 3), whose cost is $O(m'\cdot
n')$, where $m'$ and $n'$ are the numbers of both current/future
tasks and workers, respectively. Then, we will divide each problem
into $g$ subproblems, whose cost is given by $O(m'\cdot g+m')$ on
each level. For level $k$, we have $m'/g^k$ tasks in each subproblem
$M^{(k)}_i$. We will further partition it into $g$ more subproblems,
$M^{(k+1)}_j$, recursively, and each one will have $m'/g^{k+1}$
tasks. In order to obtain $m'/g^{k+1}$ tasks in each subproblem
$M^{(k+1)}_j$, we first need to find the anchor task, which needs
$O(m'/g^{k})$ cost, and then retrieve the rest tasks, which needs
$O(m/g^{k+1})$ cost. Moreover, since we have $g^{k+1}$ subproblems
on level $(k+1)$, the cost of decomposing tasks on level $k$ is
given by $O(m'\cdot g+m')$.

Since there are totally $log_g(m')$ levels, the total cost of
decomposing the MQA problem is given by:

{\scriptsize
    $$F_D=  m'\cdot n+(m' \cdot g+m')\cdot log_g(m').$$
    \vspace{-3ex}
}

\begin{figure}[ht!]\centering\vspace{-2ex}
	\scalebox{0.31}[0.31]{\includegraphics{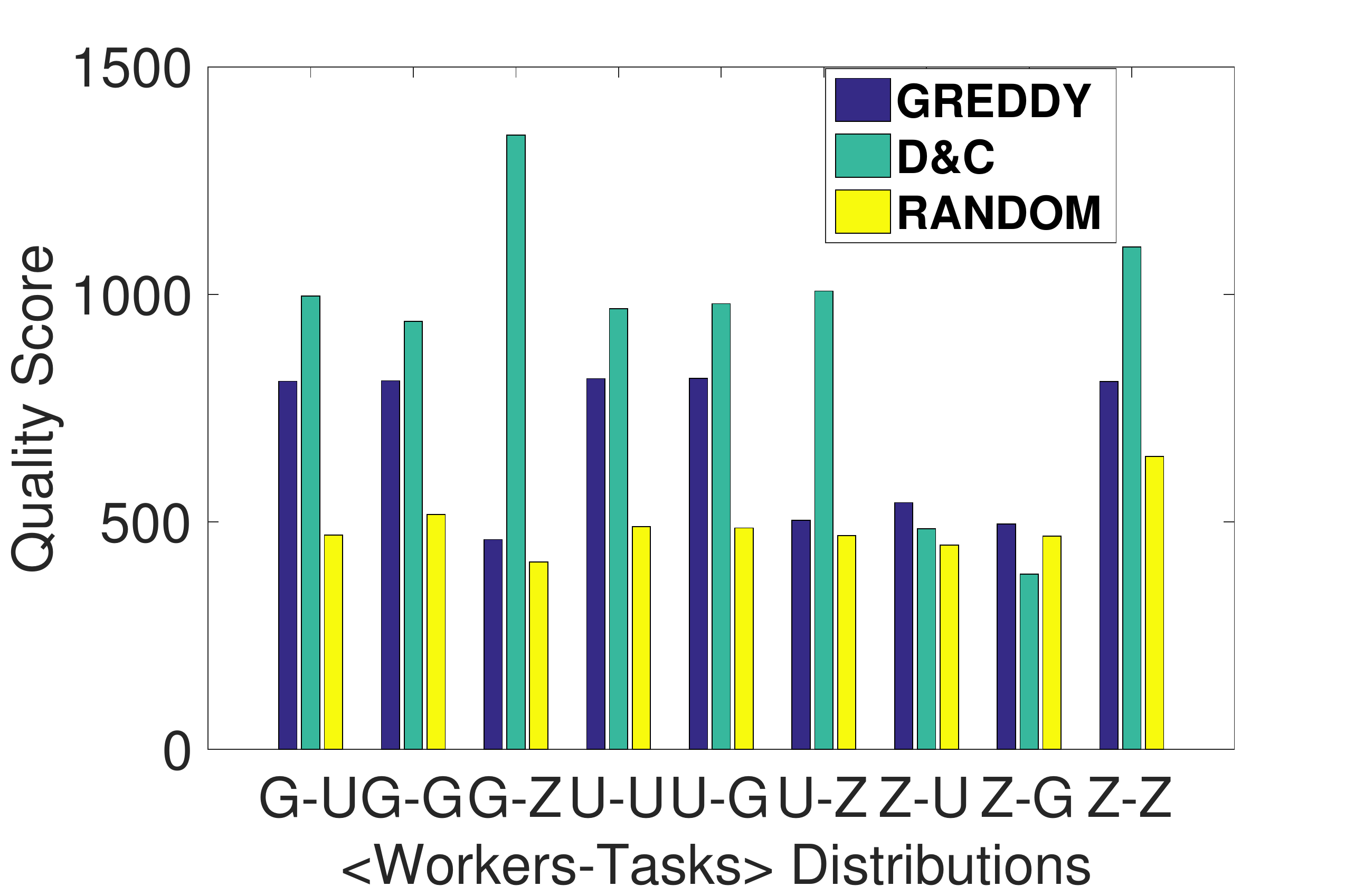}}
	\caption{\small Effect of Distributions of Workers  and Tasks on Quality Score (Synthetic Data).}
	\label{fig:distribution_score}
\end{figure}

\begin{figure}[ht!]\centering\vspace{-2ex}
	\scalebox{0.29}[0.29]{\includegraphics{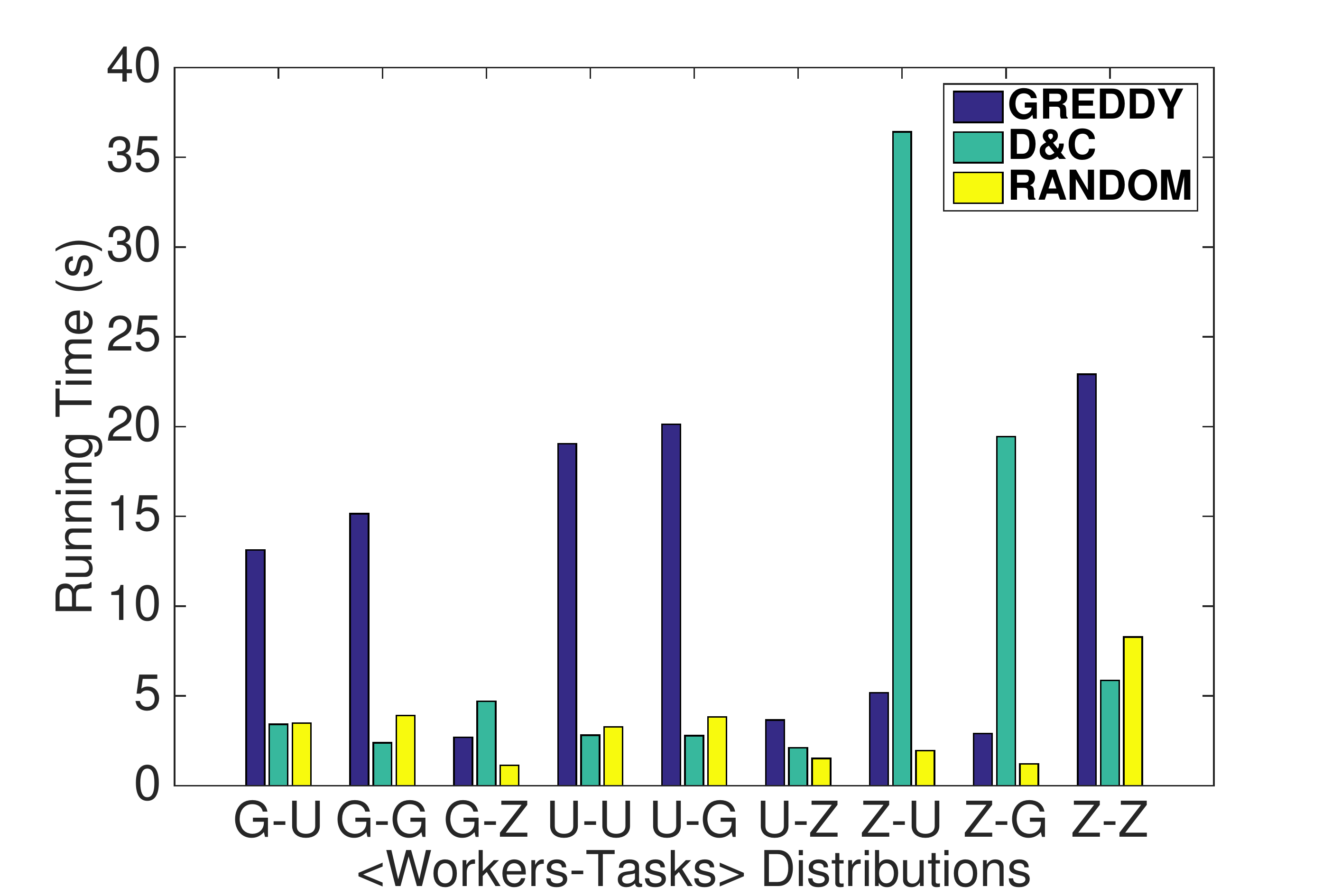}}
	\caption{\small Effect of Distributions of Workers  and Tasks on Running Time (Synthetic Data).}
	\label{fig:distribution_cpu}
\end{figure}

\noindent {\bf The Cost, $F_C$, of Recursively Conquering
    Subproblems.}

Let function $F_C(x)$ be the total cost of recursively
conquering a subproblem, which contains $x$ spatial tasks. Then, we
have the following recursive function:

{\scriptsize
    $$F_C(m') = g\cdot F_C(\left\lceil \frac{m'}{g}\right\rceil).$$
    \vspace{-3ex}
}

Assume that $deg_t$ is the average number of tasks in the assignment
instance set $I_p$. Then, the base case of function $F_C(x)$ is the
case that $x=1$, in which we greedily select one worker-and-task
pair from $deg_t$ pairs to achieve the highest quality score (via
the \textit{greedy algorithm}). Thus, we have $F_C(1) = 2deg_t^2$.

From the recursive function $F_C(x)$ and its base case, we can
obtain the total cost of the recursive invocation on levels from 1
to $log_g(m)$ below:

{\scriptsize
    $$\sum_{k = 1}^{\log_g(m')}F_c(m'/g^k) =
    \frac{2(m'-1)deg_t^2}{g-1}.$$ \vspace{-3ex}
}

\noindent {\bf The Cost, $F_M$, of Merging Subproblems.} Next, we
provide the cost, $F_M$, of merging subproblems by resolving
conflicts. For level $k$, we have $g^k$ subproblems and $m'/g^k$
tasks in each subproblem $M^{(k)}_i$. When merging the result of one
subproblem $M^{(k)}_i$ to the current maintained assignment instance
set $I_p$, there are at most $m'/g^k$ conflicted workers as each
task is only assigned with one worker. In addition, for level $k$,
we need to merge the results of $(g^k - 1)$  subproblems to the
current maintained assignment instance set. Moreover, when resolving
conflict of one worker, we may need to greedily pick one available
worker from $deg_t$ workers, which costs $2deg_t^2$.

Therefore, the worst-case cost of resolving conflicts during resolving conflicts of workers can be given by:
{\scriptsize
    $$F_M = \sum_{k=1}^{\log_g(m')}2deg_t^2(g^k - 1)\frac{m}{g^k}= 2deg_t^2(\frac{m\log m}{\log g} - \frac{g(m-1)}{g-1}).$$\vspace{-3ex}
}

\noindent {\bf The Cost, $F_B$, of the Budget Adjustment on
    Subproblems.} Then, we provide the cost, $F_B$, of the budget
adjustment in line 15 of the D\&C algorithm (in Figure
\ref{alg:dc}), or lines 17-28 in procedure {\sf
    MQA\_Budget\_Constrained\_Selection} (Figure
\ref{alg:conflict_reconcile}). Since each task is assigned with at
most one worker, the number of candidate pairs is at most same
as the number of tasks. For a subproblem with $m'/g^k$ tasks on
level $k$, lines 20-25 of Figure \ref{alg:conflict_reconcile} need
at most $(m'/g^k) ^2$ cost.  In addition, line 26 needs $|S_p|$
cost, which is also at most $(m'/g^k)^2$ cost. There are at most
$(m'/g^k)$ iterations, each of which selects at least one pair.

Therefore, the cost of the budget adjustment while merging
subproblems can be given by:\vspace{-3ex}

{\scriptsize
    \begin{eqnarray}
    F_B =
    \sum_{k=0}^{\log_g(m')}g^k\cdot2(\frac{m'}{g^k})^3=\sum_{k=0}^{\log_g(m')}
    \frac{2m'^2}{g^{2k}}= \frac{2g^2(m'^2-1)}{g^2-1}.
    \end{eqnarray}\vspace{-2ex}
}

\noindent {\bf The Total Cost of the $g$-D\&C Approach.} The total
cost, $cost_{D\&C}$, of the D\&C algorithm can be given by summing
up four costs, $F_D$, $F_C$, $F_M$, and $F_B$. That is, we
have:

{\scriptsize
    \begin{eqnarray}
    cost_{D\&C} &=& F_D + \sum_{k = 1}^{log_g(m')}F_c(m'/g^k) + F_M + F_B. \label{eq:D&C_cost}
    \end{eqnarray}
    \vspace{-3ex}
}

We take the derivation of $cost_{D\&C}$ (given in
Eq.~(\ref{eq:D&C_cost})) over $g$, and let it be 0. In particular,
we have:

{\scriptsize
    \begin{eqnarray}
    \frac{\partial cost_{D\&C}}{\partial g}
    &=&\frac{m'\log m'(g\log g - g - 1 - 2deg_t^2)}{g\log^2 (g)} \notag\\
    && -\frac{4g(m'^2-1)}{(g^2-1)^2} =0.
    \end{eqnarray}\vspace{-3ex}
}

We notice that when $g=2$, $\frac{\partial cost_{D\&C}}{\partial g}$
is much smaller than 0 but increases quickly when $g$ grows. In
addition, $g$ can only be an integer. Then we can try the integers
like 2, 3, and so on, until $\frac{\partial cost_{D\&C}}{\partial
    g}$ is above 0.

\begin{figure}[ht!]
	\centering\vspace{-2ex}
	\subfigure[][{\scriptsize Quality Score}]{
		\scalebox{0.2}[0.2]{\includegraphics{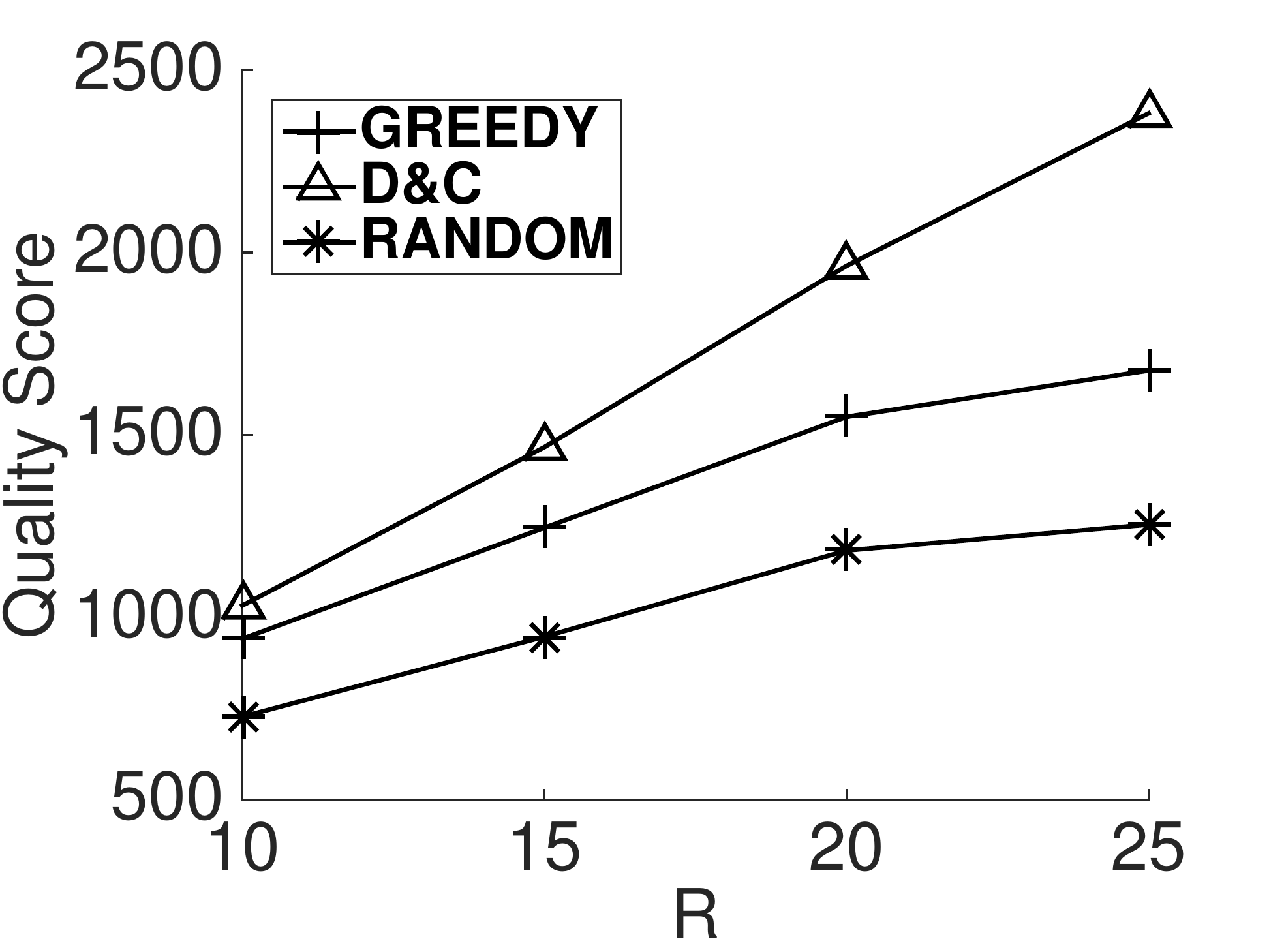}}
		\label{subfig:r_score}}
	\subfigure[][{\scriptsize Running Time}]{
		\scalebox{0.2}[0.2]{\includegraphics{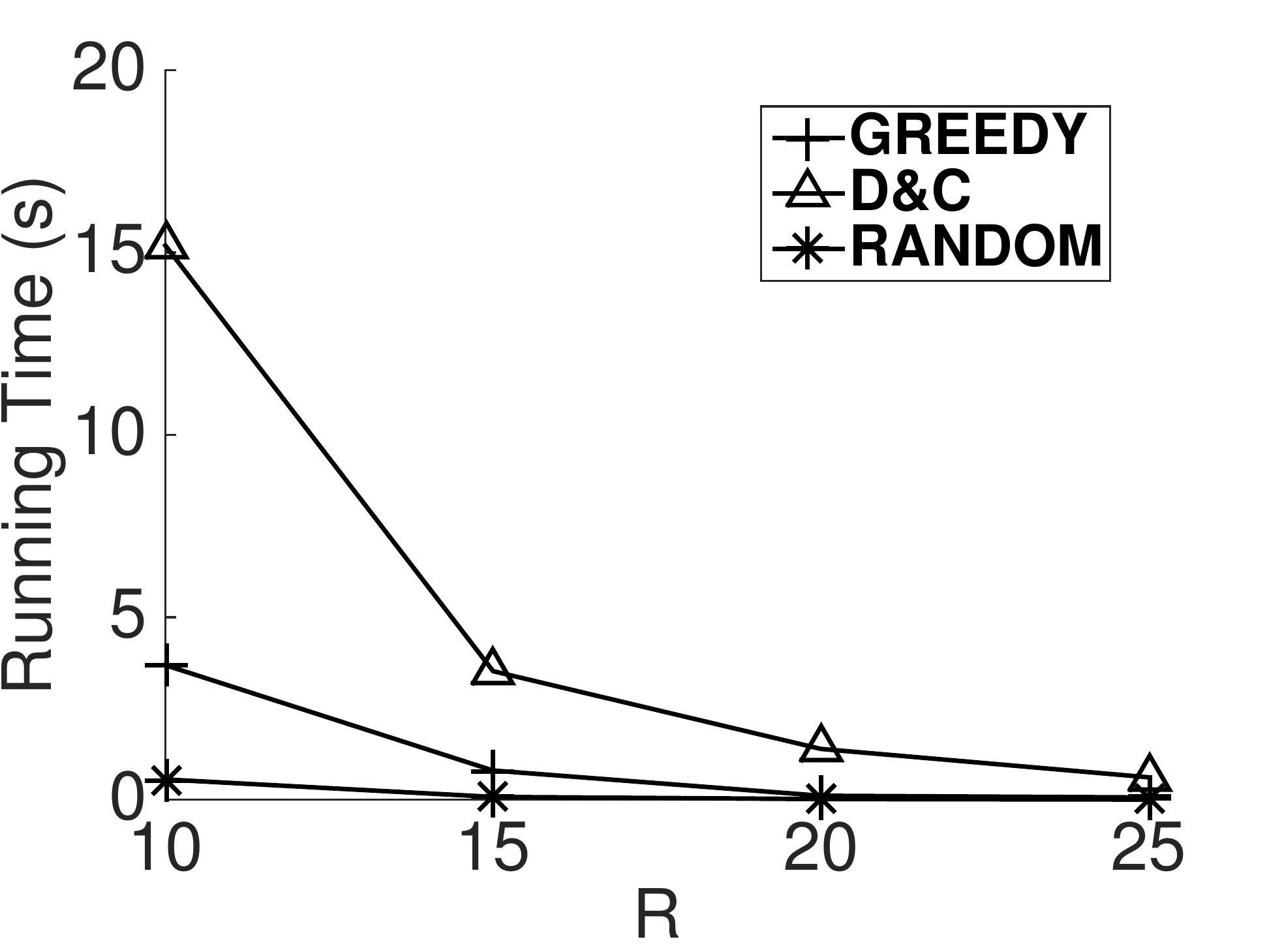}}
		\label{subfig:r_cpu}}\vspace{-2ex}
	\caption{\small Effect of the Number, $R$, of Time Instances (Synthetic Data).}
	\label{fig:round}
\end{figure}

\begin{figure}[ht!]
	\centering\vspace{-2ex}
	\subfigure[][{\scriptsize Quality Score}]{
		\scalebox{0.2}[0.2]{\includegraphics{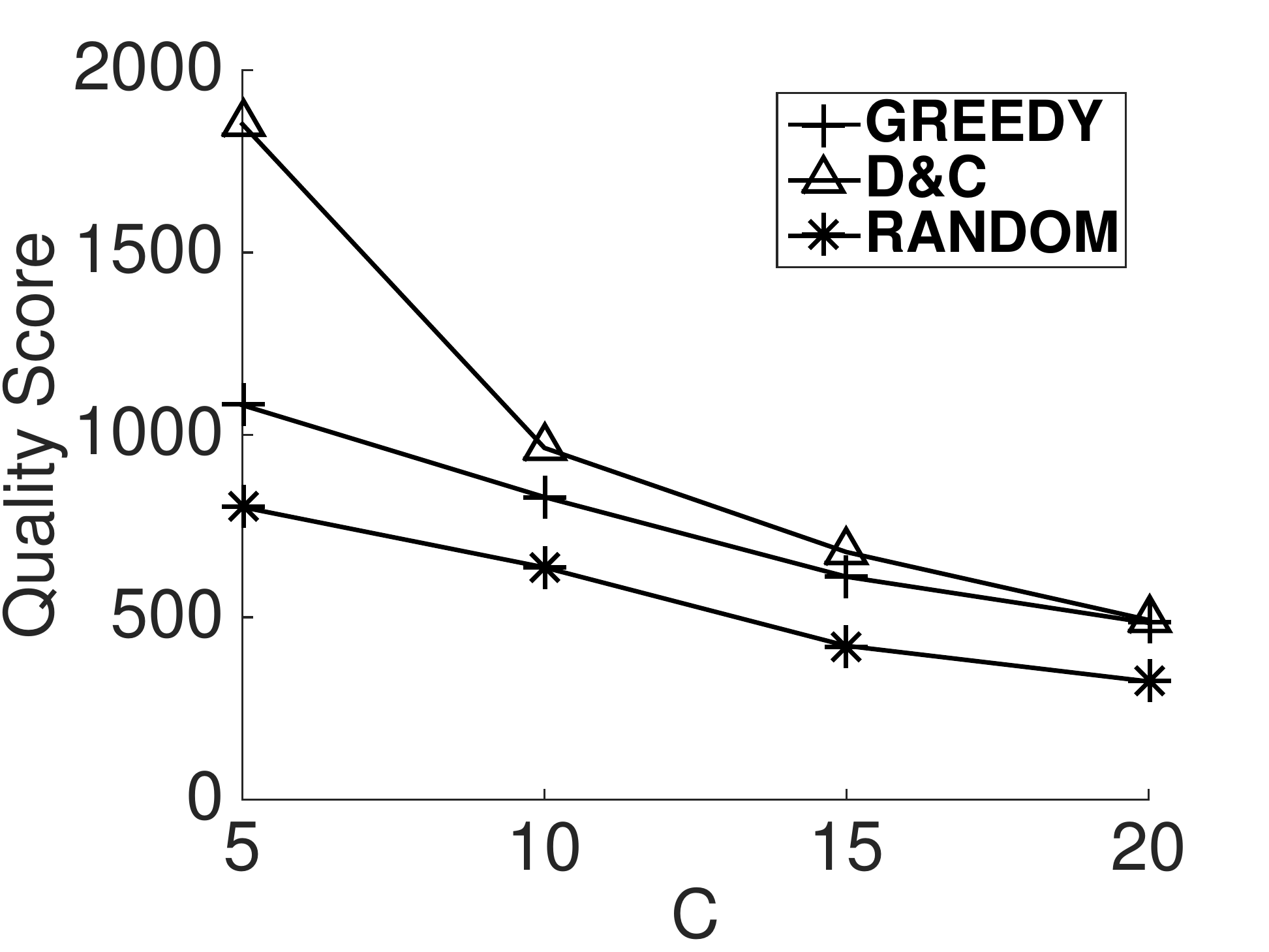}}
		\label{subfig:c_score}}
	\subfigure[][{\scriptsize Running Time}]{
		\scalebox{0.2}[0.2]{\includegraphics{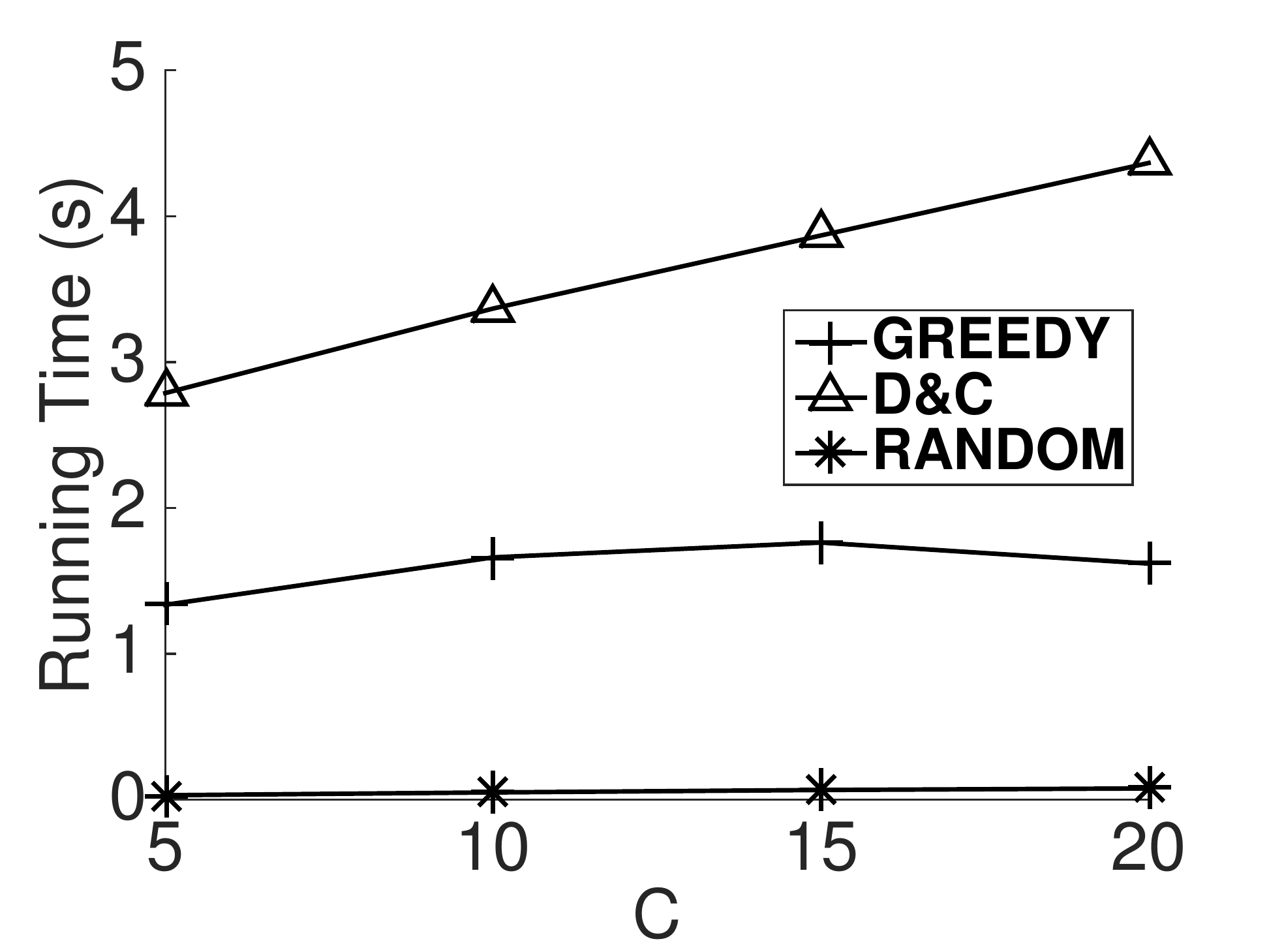}}
		\label{subfig:c_cpu}}\vspace{-2ex}
	\caption{\small Effect of the Unit Price $C$ w.r.t. Distance $dist(w_i, t_j)$  (Synthetic Data).}
	\label{fig:constant}
\end{figure}

\begin{figure}[ht!]
	\centering\vspace{-2ex}
	\subfigure[][{\scriptsize Quality Score (GAUS)}]{
		\scalebox{0.15}[0.15]{\includegraphics{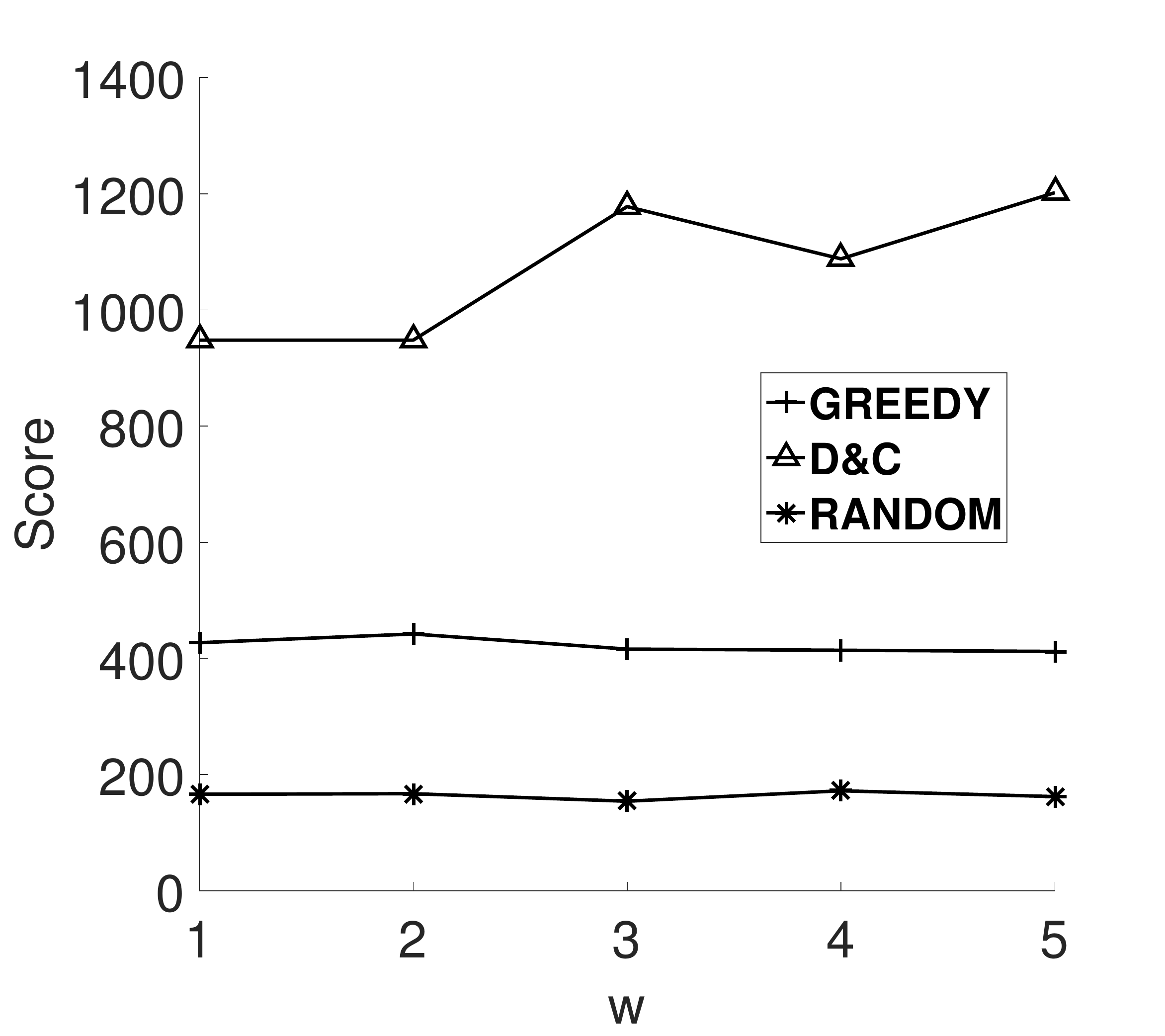}}
		\label{subfig:quality_score_gaussian}}
	\subfigure[][{\scriptsize Quality Score (UNIF)}]{
		\scalebox{0.15}[0.15]{\includegraphics{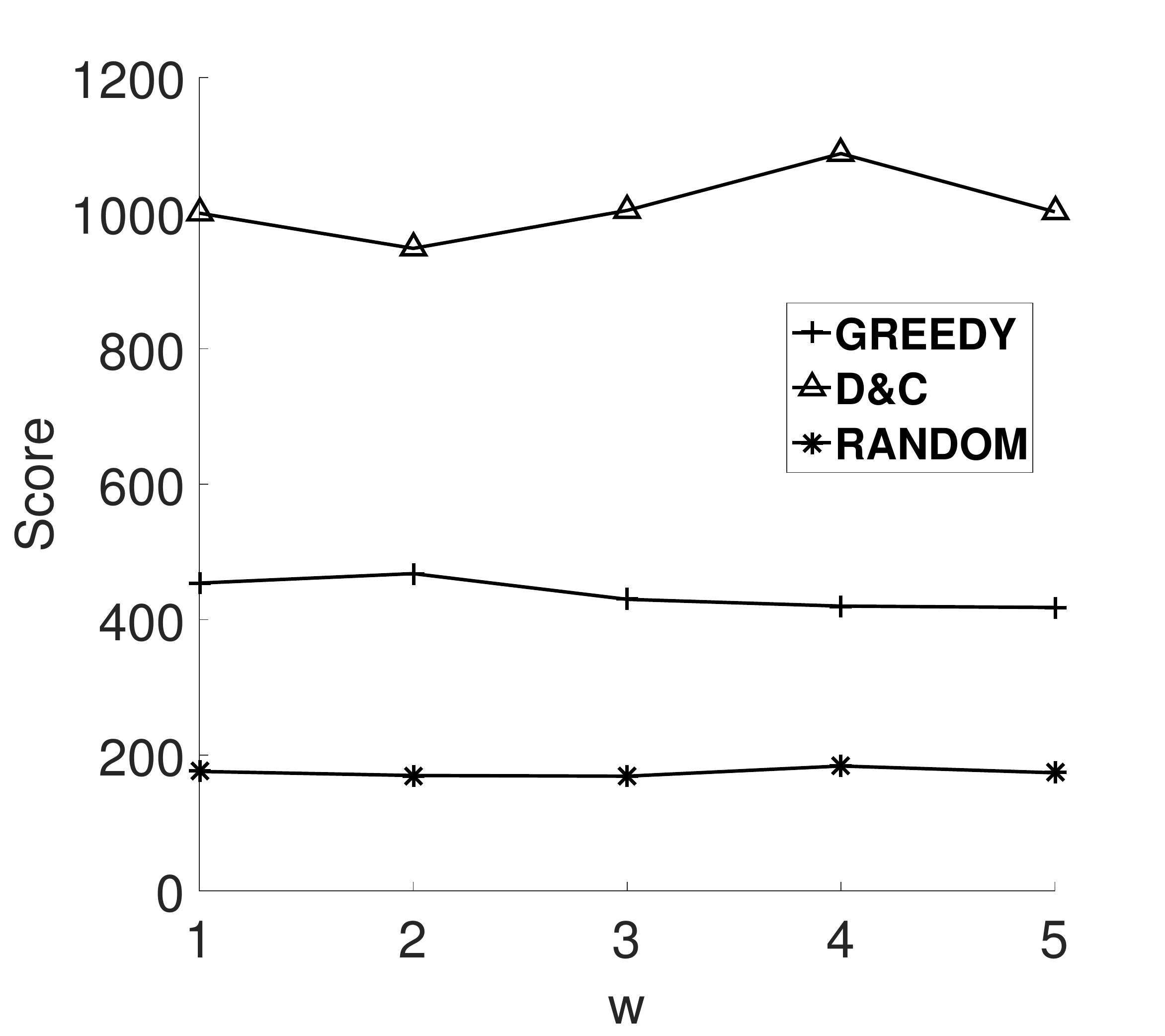}}
		\label{subfig:quality_score_uniform}}
	\subfigure[][{\scriptsize Quality Score (ZIPF)}]{
		\scalebox{0.15}[0.15]{\includegraphics{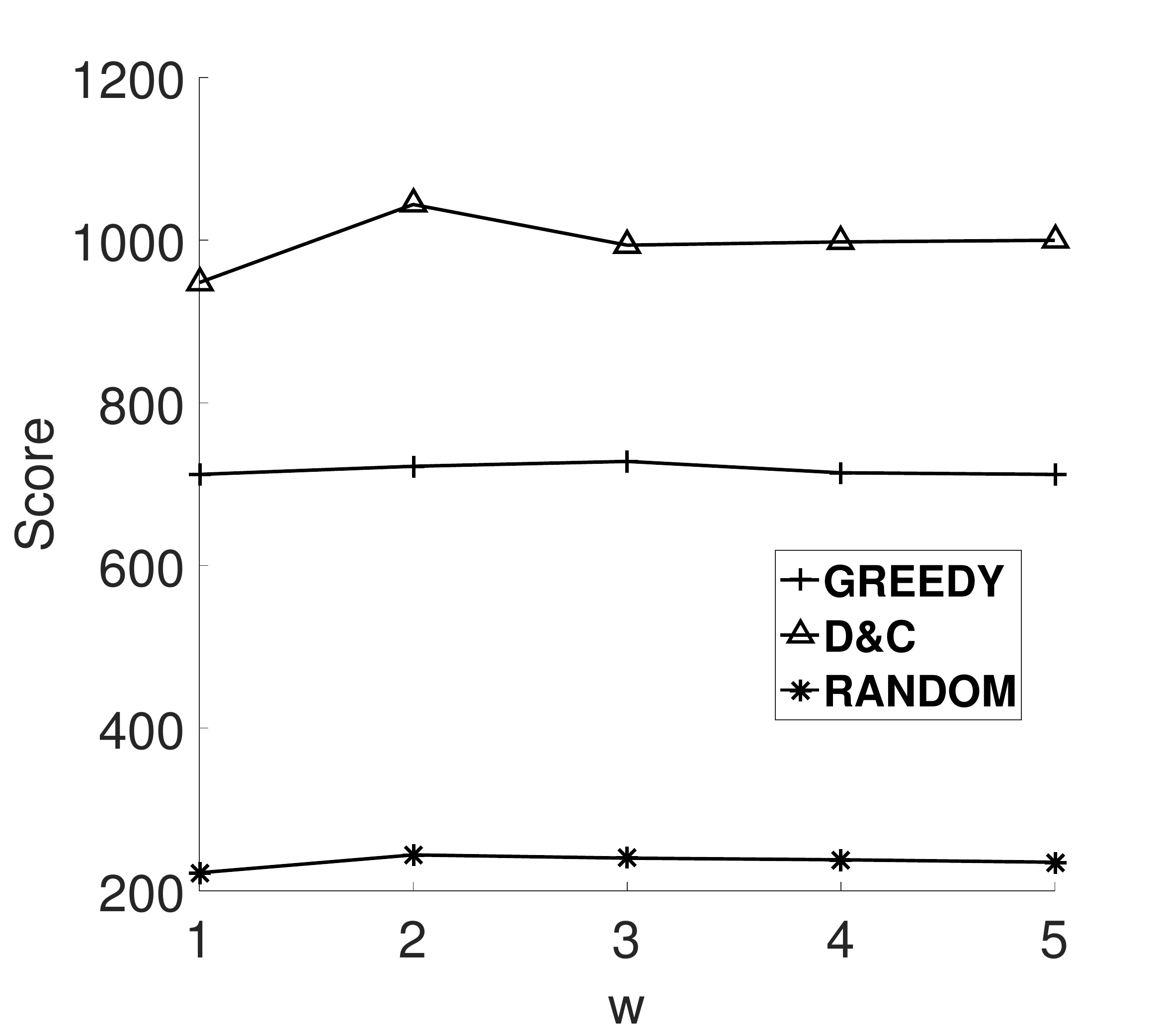}}
		\label{subfig:quality_score_zipf}}\vspace{-2ex}
	\caption{\small Effect of the Window Size $w$  (Synthetic Data).}
	\label{fig:windowsize}
\end{figure}

\subsection{Results with Different Worker-Task Distributions}

In this section, we present the experimental results for workers and
tasks with different location distributions, where parameters of
synthetic data are set to default values. We denote the Uniform
distribution as U, the Gaussian distribution as G, and the Zipf
distribution as Z. Then, for $\langle worker$$-$$task \rangle$
distributions, we tested the quality score and running time over 9
distribution combinations, including G-U, G-G, G-Z, U-U, U-G, U-Z,
Z-U, Z-G, and Z-Z, and the results are shown in Figures
\ref{fig:distribution_score} and \ref{fig:distribution_cpu}.

Similar to previous results, as shown in Figure
\ref{fig:distribution_score}, the D\&C algorithm can achieve the
highest quality score, compared with GREEDY and RANDOM, over all the
9 worker/task distribution combinations. For the running time, as
illustrated in Figure \ref{fig:distribution_cpu}, with different
combinations of worker/task distributions, D\&C can achieve low time
cost in most cases. Only for Z-U and Z-G, D\&C incurs higher time
cost than GREEDY and RANDOM, due to the unbalanced distributions of
workers and tasks. In particular, GREEDY and RANDOM iteratively
assign one valid pair and maintain the rest of valid pairs in each
iteration. When the distributions of workers and tasks are similar,
for example, G-G, U-U, and Z-Z, running times of GREEDY and RANDOM
become longer than D\&C. Especially, for Z-Z, almost all the workers
can reach all the tasks, which leads to the highest number of valid
pairs among all the 9 distribution combinations. As a result, both
GREEDY and RANDOM need much higher running time than that of other
distribution combinations. In general, with different worker and
task distributions, our GREEDY and D\&C can both achieve high
quality scores (with small time cost).

\subsection{The MQA Performance vs. the Number,
	$R$, of Time Instances and  the Unit
	Price $C$ w.r.t. Distance $dist(w_i, t_j)$}

\noindent {\bf The MQA Performance vs. the Number,
	$R$, of Time Instance.} 
Figure \ref{fig:round} reports the experimental
results for different numbers, $R$, of time instance from 10 to 25 on
synthetic data sets, where other parameters are set to default
values. In Figure \ref{subfig:r_score}, when the number, $R$, of
time instances increases, the total quality score of three MQA approaches
also increases. Since we consider a fixed time interval $P$ with
more time instances (each with budget $B$), the total quality score within
interval $P$ expects to increase for more time instances. D\&C can achieve
higher quality scores than GREEDY.

In Figure \ref{subfig:r_cpu}, when $R$ becomes larger, the running
time of all the three tested approaches decreases. This is because,
given $m$ tasks and $n$ workers within time interval $P$, for more
time instances, the average number of workers/tasks per time instance decreases,
which leads to lower time cost per time instance. Similar to previous
results, the running time of GREEDY is lower than that of D\&C.

\begin{figure}[ht!]
	\centering\vspace{-2ex}
	\subfigure[][{\scriptsize Quality Score}]{
		\scalebox{0.2}[0.2]{\includegraphics{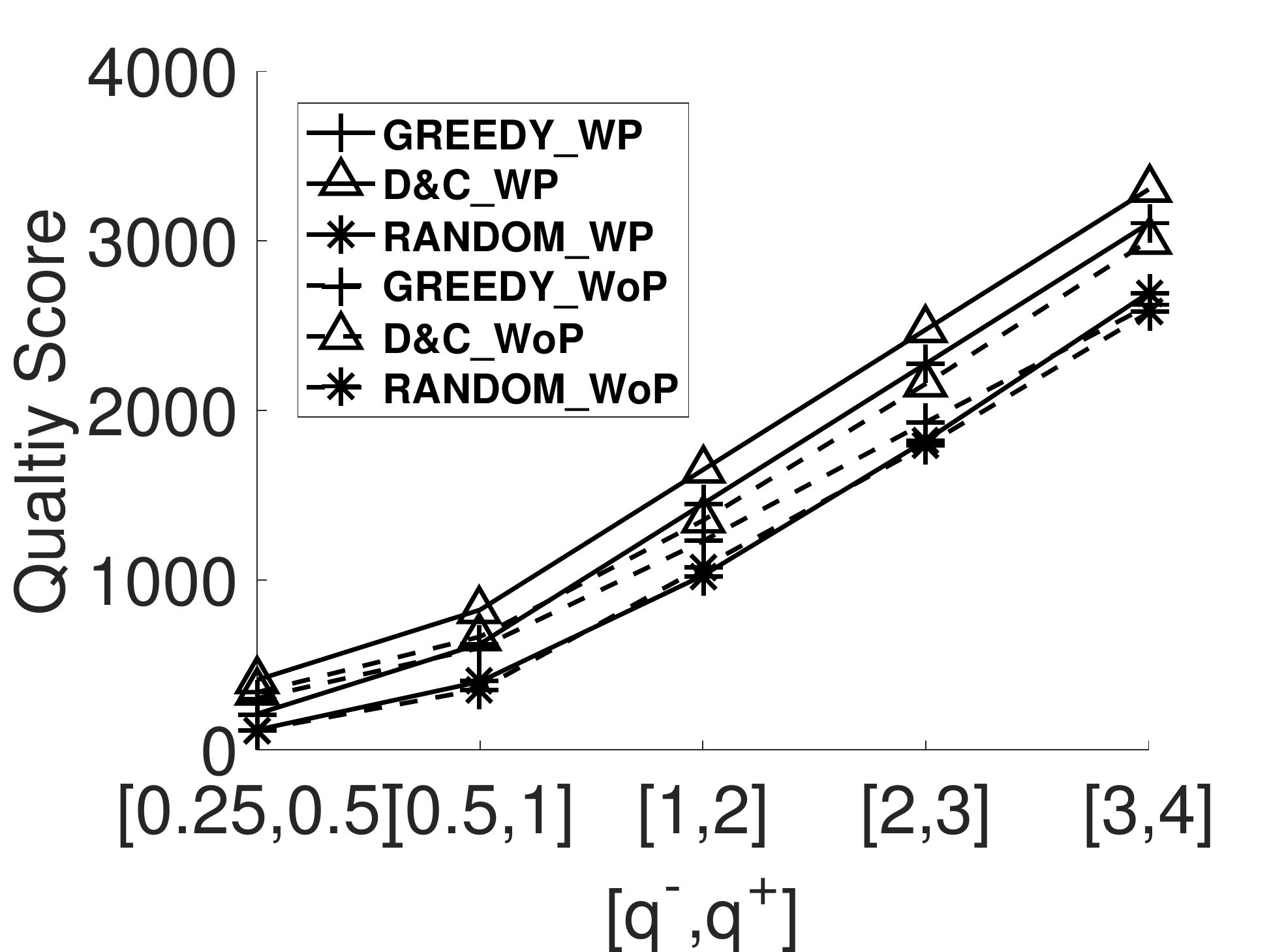}}
		\label{subfig_c:quality_score}}
	\subfigure[][{\scriptsize Running Time}]{
		\scalebox{0.2}[0.2]{\includegraphics{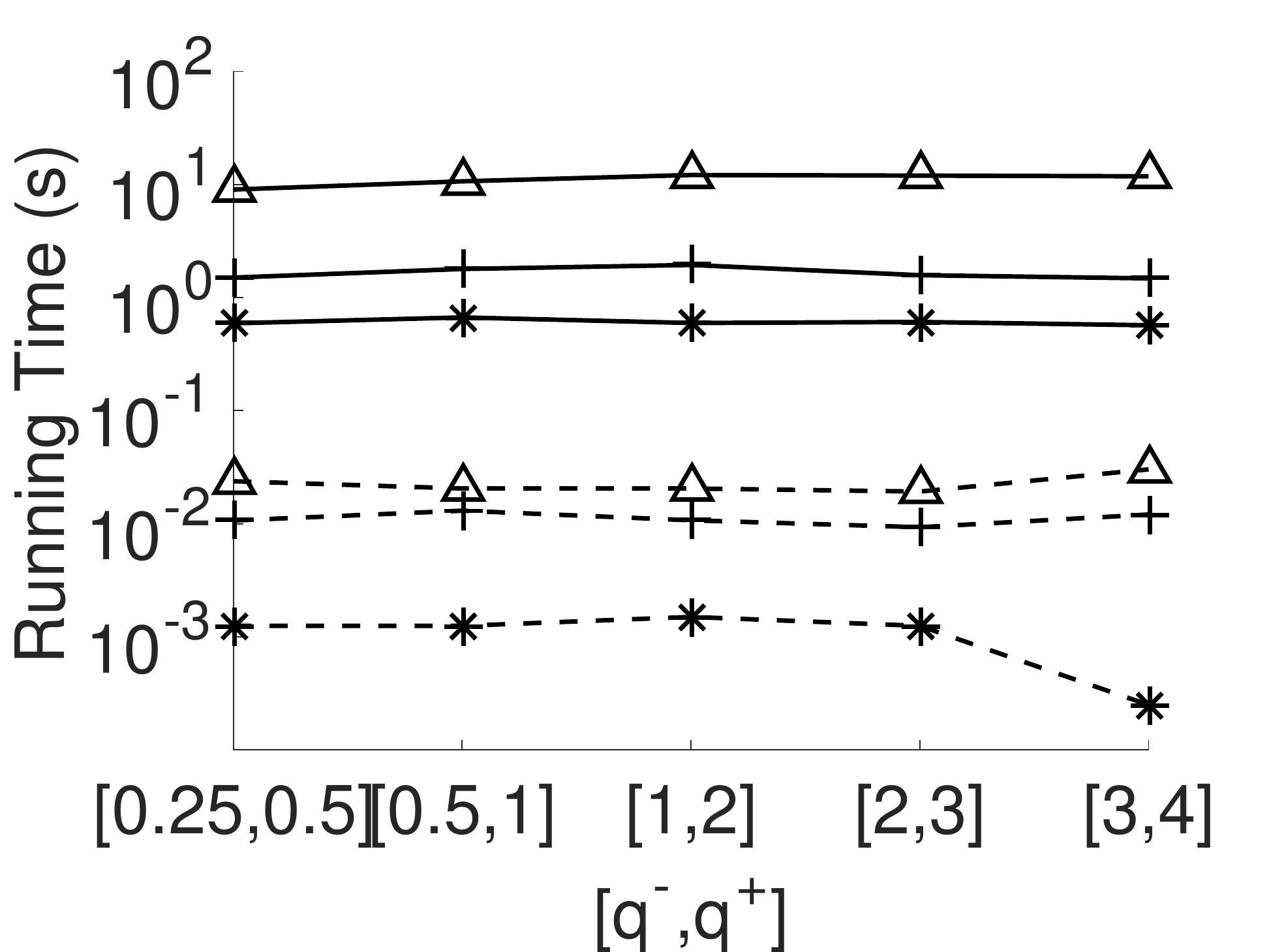}}
		\label{subfig_c:quality_cpu}}\vspace{-2ex}
	\caption{\small Effect of the Range of Quality Score $q_{ij}$  (Real Data).}\vspace{-2ex}
	\label{fig_c:quality}
\end{figure}

\begin{figure}[ht!]
	\centering\vspace{-2ex}
	\subfigure[][{\scriptsize Quality Score}]{
		\scalebox{0.2}[0.2]{\includegraphics{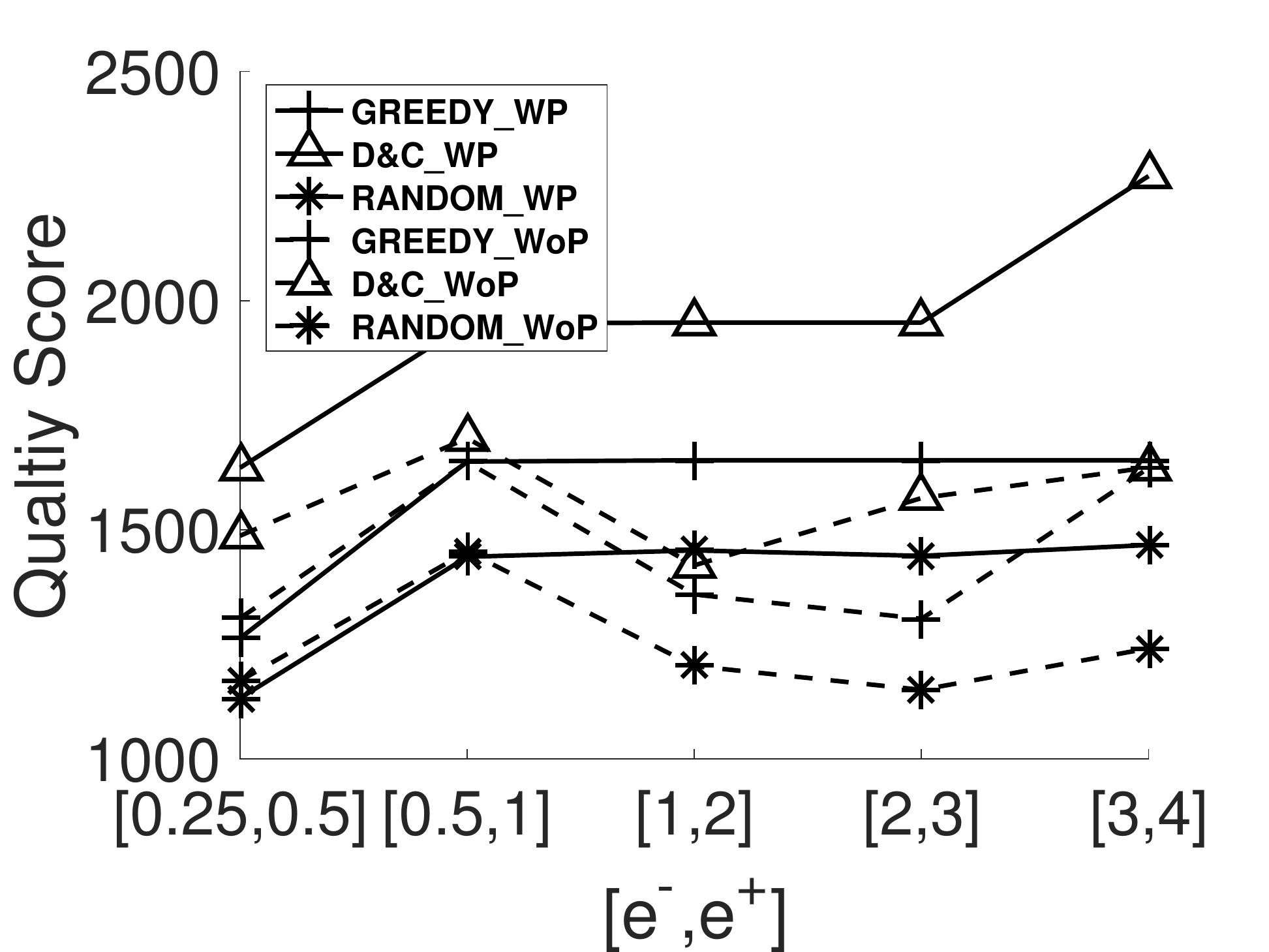}}
		\label{subfig_c:deadline_score}}
	\subfigure[][{\scriptsize Running Time}]{
		\scalebox{0.2}[0.2]{\includegraphics{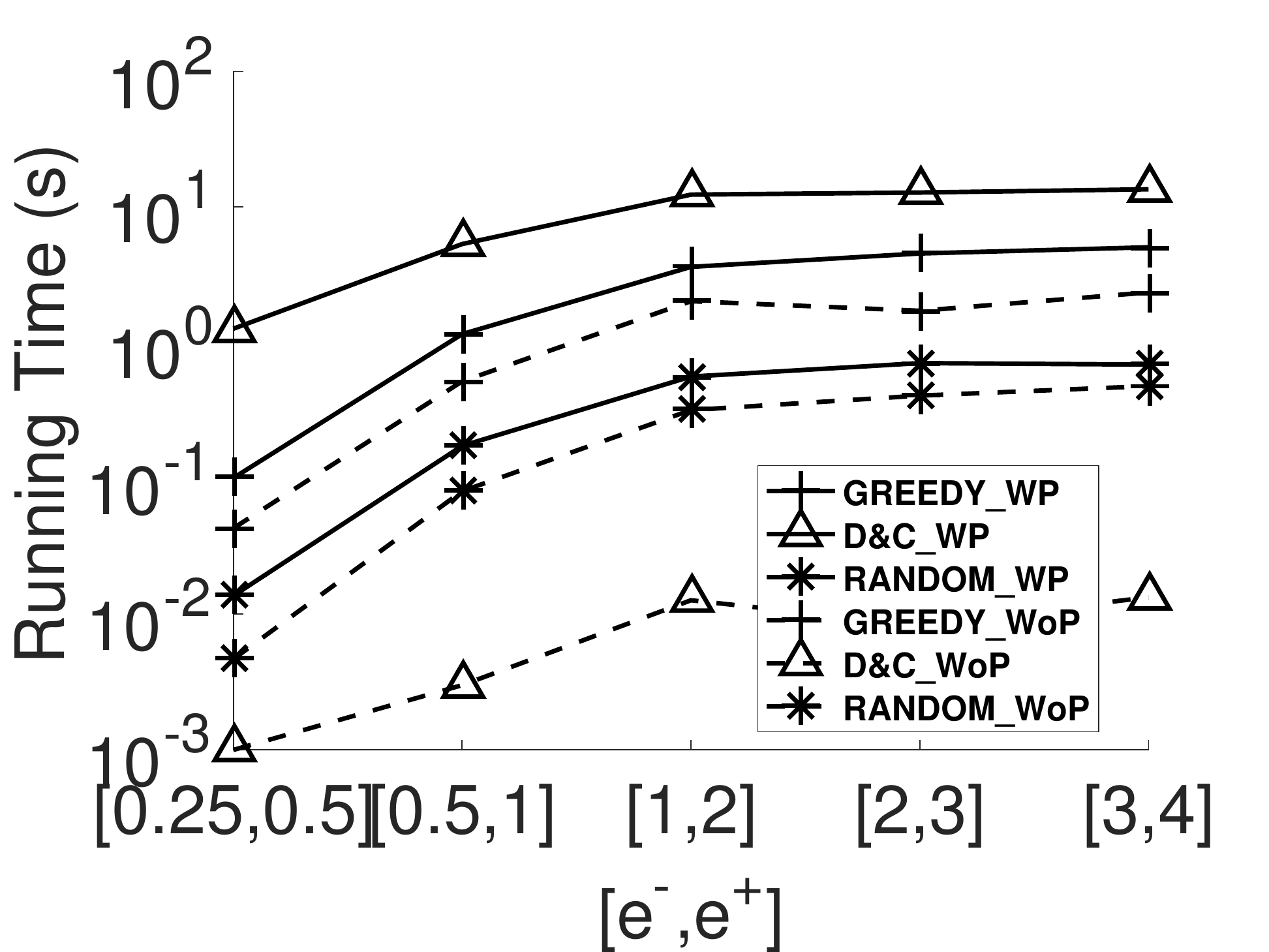}}
		\label{subfig_c:deadline_cpu}}\vspace{-2ex}
	\caption{\small Effect of the Range of Tasks' Deadlines $e_j$ (Real Data).}\vspace{-2ex}
	\label{fig_c:deadline}
\end{figure}

\begin{figure}[ht!]
	\centering\vspace{-2ex}
	\subfigure[][{\scriptsize Quality Score}]{
		\scalebox{0.2}[0.2]{\includegraphics{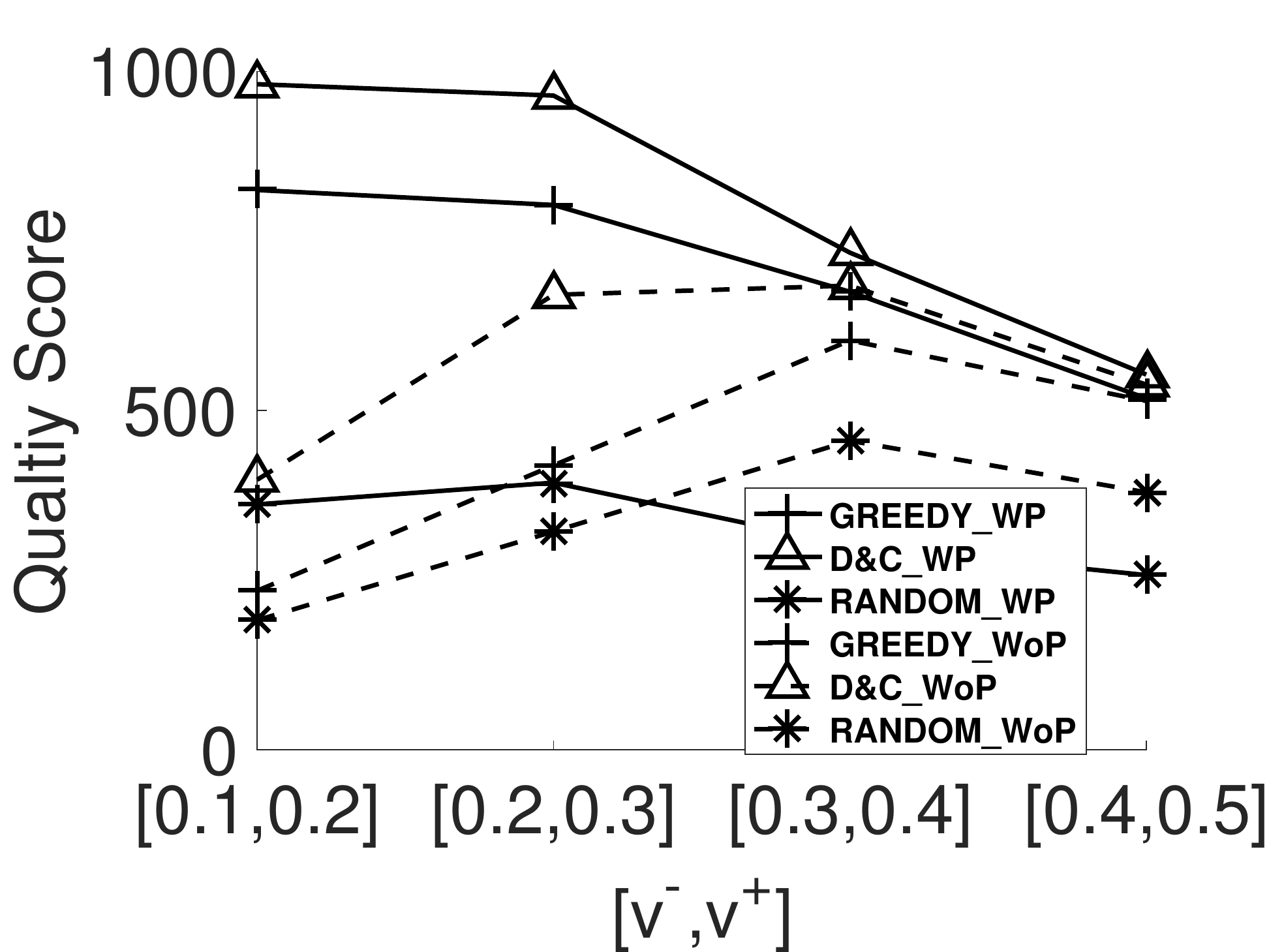}}
		\label{subfig_c:velocity_score}}
	\subfigure[][{\scriptsize Running Time}]{
		\scalebox{0.2}[0.2]{\includegraphics{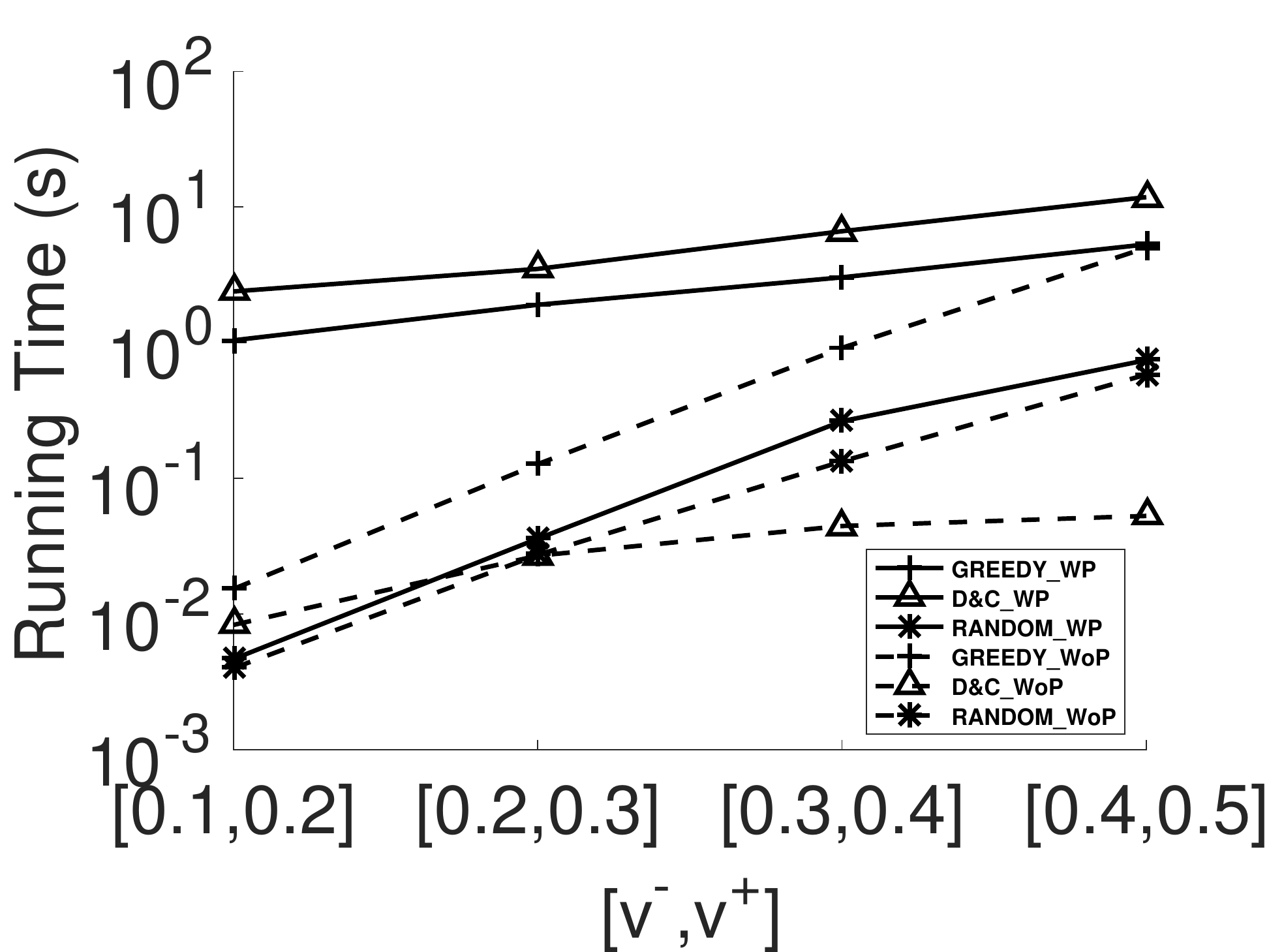}}
		\label{subfig_c:velocity_cpu}}\vspace{-2ex}
	\caption{\small Effect of the Range of Velocities $[v^-, v^+]$ (Synthetic Data).}\vspace{-2ex}
	\label{fig_c:velocity}
\end{figure}

\noindent {\bf The MQA Performance vs. the Unit
	Price $C$ w.r.t. Distance $dist(w_i, t_j)$.} Figure
\ref{fig:constant} illustrates the experimental results on different
unit prices $C$ w.r.t. distance $dist(w_i, t_j)$ from $5$ to $20$
over synthetic data, where other parameters are set to default
values. In Figure \ref{subfig:c_score}, when the unit price $C$
increases, the overall quality scores of all the three approaches
decrease. This is because for larger $C$, the number of valid
worker-and-task pairs for each time instance decreases, under the budget
constraint. Thus, the overall quality score of all the selected
assignments also expects to decrease for large $C$. Similar to
previous results, D\&C has higher quality scores than GREEDY. In
Figure \ref{subfig:c_cpu}, running times of GREEDY and RANDOM are
not very sensitive to $C$. However, with large $C$, the running time
of D\&C increases, since we need to check the constraint of budget
from lower divide-and-conquer levels, which increases the total
running time. For different $C$ values, GREEDY has lower running
times than D\&C.

\begin{figure}[ht!]
	\centering\vspace{-2ex}
	\subfigure[][{\scriptsize Quality Score}]{
		\scalebox{0.2}[0.2]{\includegraphics{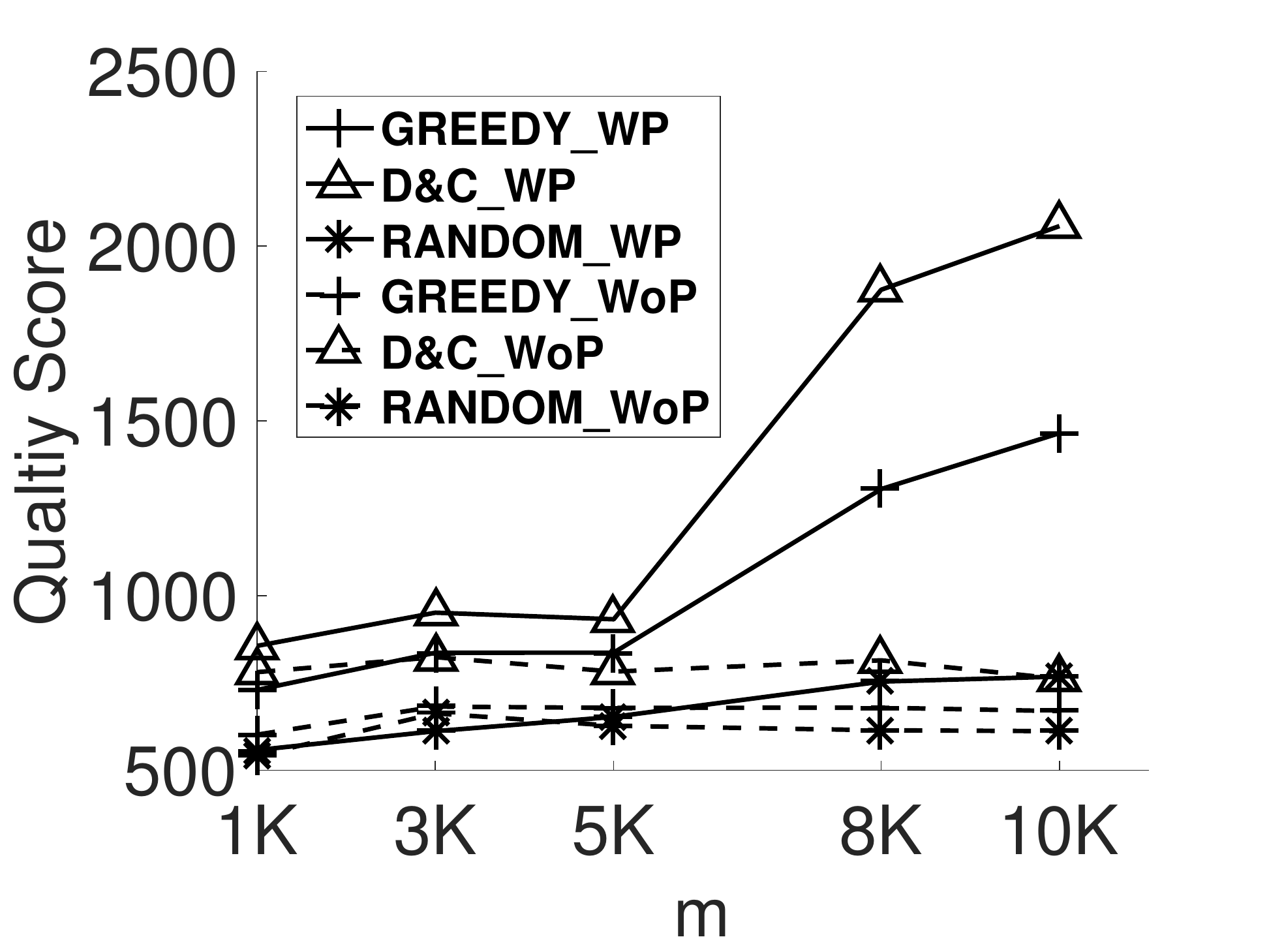}}
		\label{subfig_c:task_score}}
	\subfigure[][{\scriptsize Running Time}]{
		\scalebox{0.2}[0.2]{\includegraphics{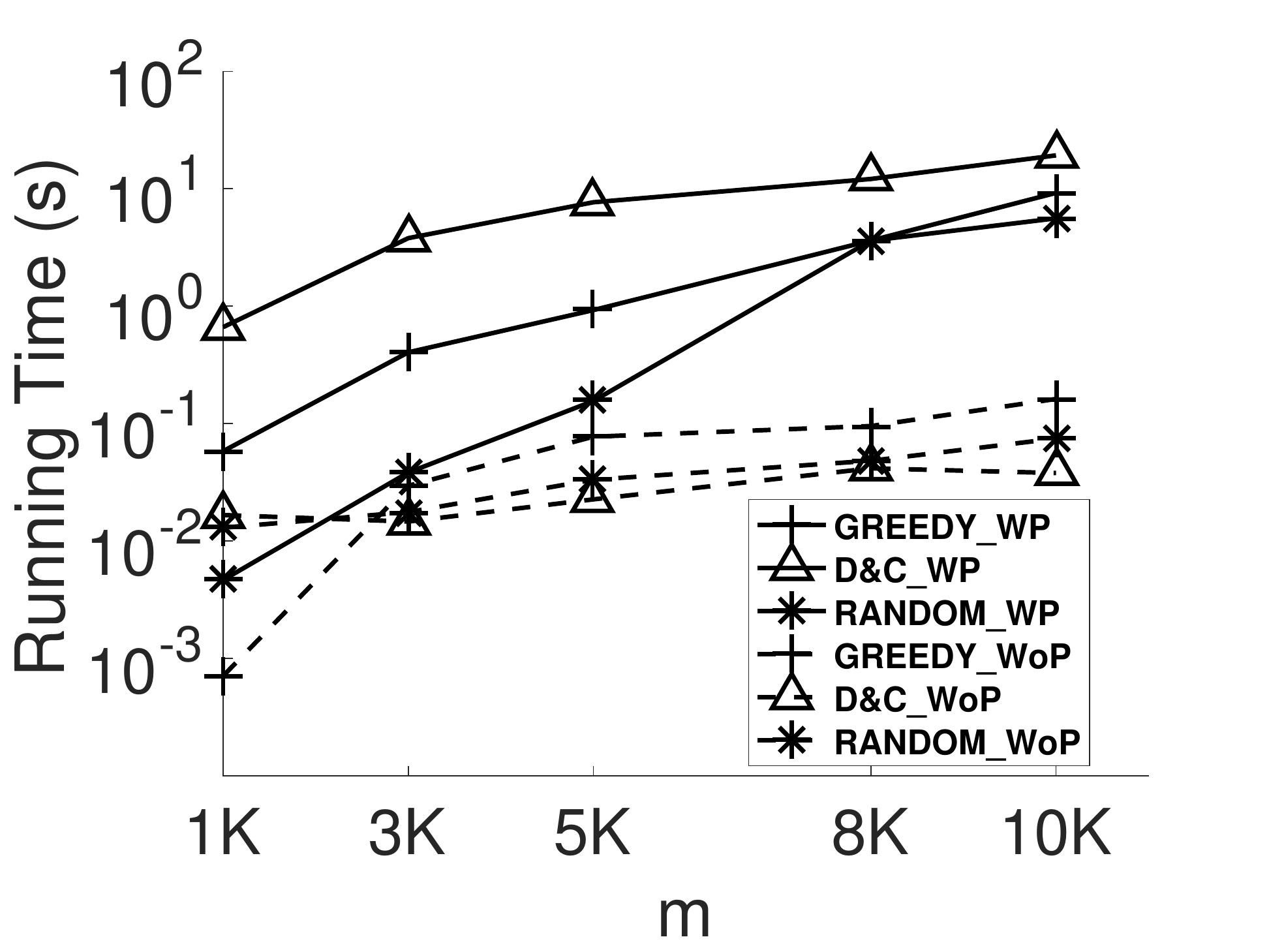}}
		\label{subfig_c:task_cpu}}\vspace{-2ex}
	\caption{\small Effect of the Number, $m$, of Tasks  (Synthetic Data).}\vspace{-2ex}
	\label{fig_c:task}
\end{figure}

\begin{figure}[ht!]
	\centering\vspace{-2ex}
	\subfigure[][{\scriptsize Quality Score}]{
		\scalebox{0.2}[0.2]{\includegraphics{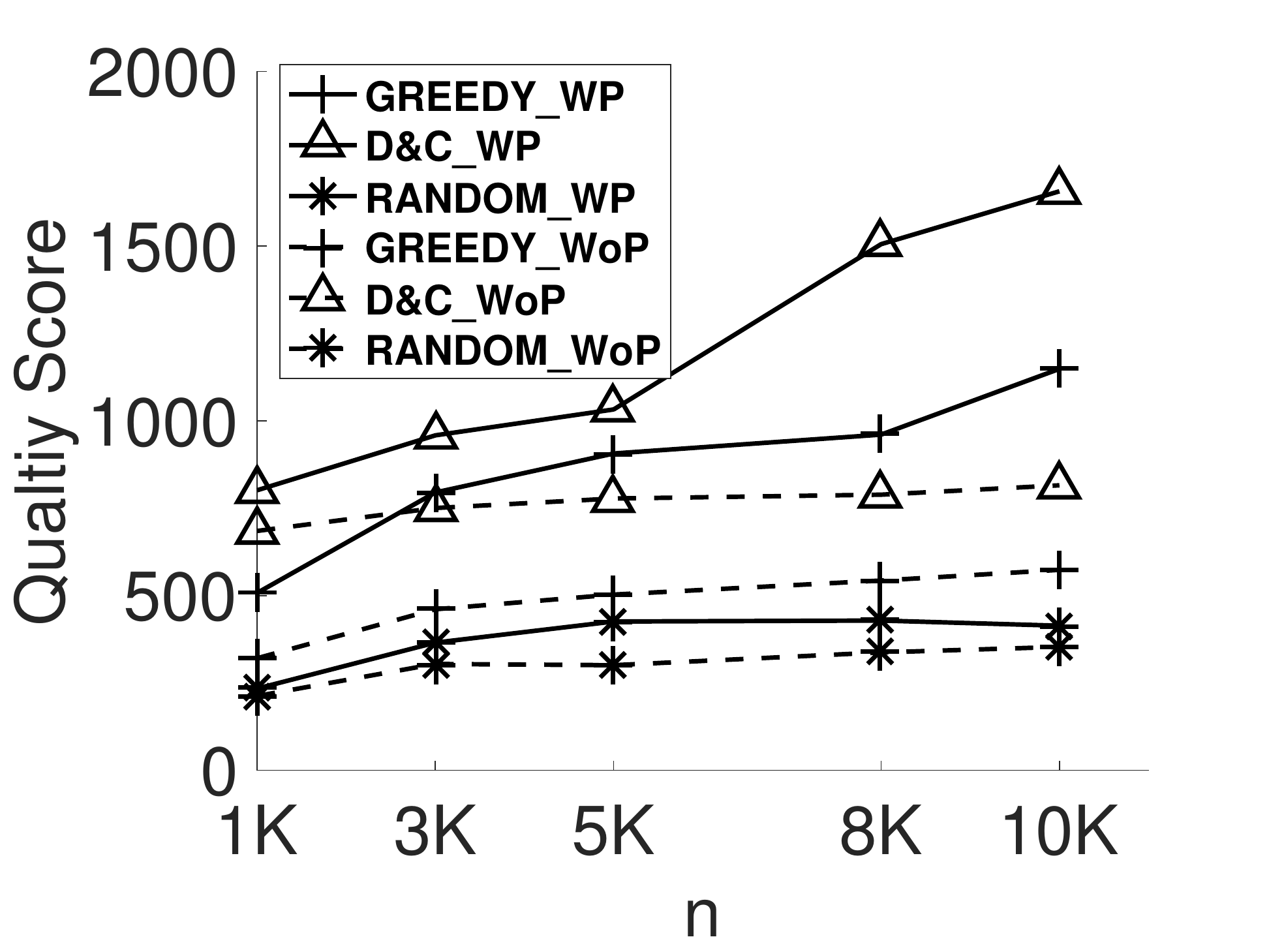}}
		\label{subfig_c:worker_score}}
	\subfigure[][{\scriptsize Running Time}]{
		\scalebox{0.2}[0.2]{\includegraphics{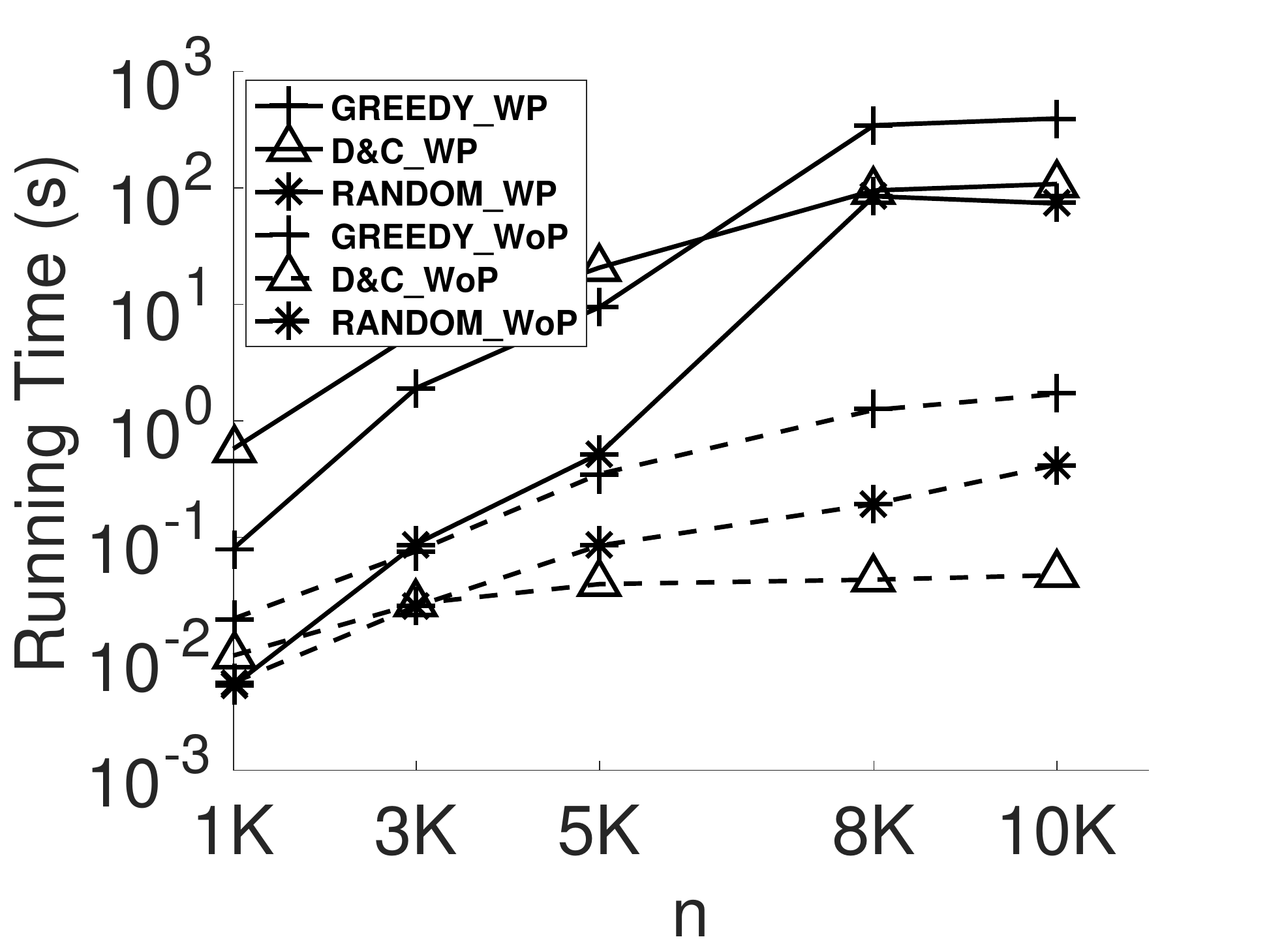}}
		\label{subfig_c:worker_cpu}}\vspace{-2ex}
	\caption{\small Effect of the Number, $n$, of Workers (Synthetic Data).}\vspace{-2ex}
	\label{fig_c:worker}
\end{figure}

\subsection{The MQA Performance vs. the Window Size $w$}

We show the results of quality scores by varying window size on different workers distributions in Figure \ref{fig:windowsize}. We can see the sliding window size affects the quality score slightly for GREEDY and RANDOM. For D\&C, it can achieve the highest quality score when the window size equals to 3 on workers with Gaussian Distribution (as shown in Figure \ref{subfig:quality_score_gaussian}), and equals to 4 and 2 on workers with Uniform and Zipf distribution respectively.

\subsection{Results of Comparison with Straightforward Methods}

Figure \ref{fig_c:quality} to Figure \ref{fig_c:worker} compare the quality scores and running times of our
MQA approaches (with predicted workers/tasks) with that of the
straightforward method which selects assignments at current and next
time instances separately (without predictions) by varying the range [$q^−$, $q^+$] of quality score $q_{ij}$, the range [$e^-$, $e^+$] of tasks' deadlines $e_j$, the range [$v^-$, $v^+$] of workers' velocities $v_i$, the number of tasks $m$ and the number of workers $n$. We denote MQA approaches with prediction as
GREEDY\_WP, D\&C\_WP, and RANDOM\_WP, and those without prediction
as GREEDY\_WoP, D\&C\_WoP, and RANDOM\_WoP, respectively.

\end{document}